%% file: main.tex
\documentclass[11pt]{article}

\usepackage{fullpage}
\usepackage[]{hyperref}
\usepackage[dvipsnames]{xcolor}
\usepackage{natbib}
\usepackage{palatino}
\usepackage{float}
\usepackage{tikz}
\usetikzlibrary{decorations.pathreplacing}
\usetikzlibrary{positioning, fit, calc}
\definecolor{ForestGreen}{rgb}{0.0333,0.4451,0.0333}
\definecolor{DarkRed}{rgb}{0.65,0,0}
\definecolor{Red}{rgb}{1,0,0}
\hypersetup{
    unicode=false,          % non-Latin characters in Acrobatӳ bookmarks
    colorlinks=true,        % false: boxed links; true: colored links
    linkcolor=BrickRed,          % color of internal links (change box color with linkbordercolor)
    citecolor=OliveGreen,        % color of links to bibliography
    filecolor=magenta,      % color of file links
    urlcolor=cyan           % color of external links
}

\newcommand{\eps}{\epsilon}
\newcommand{\mcP}{\mathcal{P}}
\newcommand{\poly}{\operatorname{poly}} % richard changed it

\newcommand{\bound}{\operatorname{bound}}

\input{prologue.tex}

\title{Planar Length-Constrained Minimum Spanning Trees}

\author{%
   \begin{tabular}{cc}
     D Ellis Hershkowitz\thanks{Supported in part by NSF grant CCF-2403236.} & \hspace{0.5cm} Richard Z Huang\footnotemark[1] \\\\
     \multicolumn{2}{c}{Brown University} \\
     \multicolumn{2}{c}{}
   \end{tabular}
 }

\date{}
\begin{document}
\maketitle

\begin{abstract}
    In length-constrained minimum spanning tree (MST) we are given an $n$-node graph $G = (V,E)$ with edge weights $w : E \to \mathbb{Z}_{\geq 0}$ and edge lengths $l: E \to \mathbb{Z}_{\geq 0}$ along with a root node $r \in V$ and a length constraint $h \in \mathbb{Z}_{\geq 0}$. Our goal is to output a spanning tree of minimum weight according to $w$ in which every node is at distance at most $h$ from $r$ according to $l$.

    We give a polynomial-time algorithm for planar graphs which, for any constant $\eps > 0$, outputs an $O\left(\log^{1+\eps} n\right)$-approximate solution with every node at distance at most $(1+\eps)h$ from $r$. Our algorithm is based on new length-constrained versions of classic planar separators and divisions which may be of independent interest. 
    Additionally, our algorithm works for length-constrained Steiner tree and bounds the integrality gap of the natural linear program as $O(\log ^ 3 n /\eps)$, again with $(1+\eps)$ slack in the length constraint. 
    Complementing this, we show that any algorithm on general graphs for length-constrained MST in which nodes are at most $2h$ from $r$ cannot achieve an approximation of $O\left(\log ^{2-\epsilon} n\right)$ for any constant $\epsilon > 0$ under standard complexity assumptions; as such, our results separate the approximability of length-constrained MST in planar and general graphs. 

    % The length-constrained minimum spanning tree problem seeks a minimum-weight spanning tree $T$ of an $n$-node edge-weighted graph $G$ rooted at $r$, such that the length (distance from any node to the root $r$ in $T$) is at most an input $h\geq1$. 
    
    % When the input graph is planar, we give a poly-time algorithm that outputs a spanning tree whose length is at most $(1+\eps)h$ for any $\eps>0$ and weight is -approximate, beating new inapproximability results for the problem on general graphs. Our algorithm generalizes naturally to the Steiner version of the problem and to general edge length functions. As part of our algorithm we introduce a new length-constrained planar separator, which may be of independent interest. 
    
\end{abstract}
\thispagestyle{empty}
\newpage
\pagenumbering{gobble}
\newpage

\begingroup
\hypersetup{linkcolor=ForestGreen}
\tableofcontents
\endgroup
\newpage
\pagenumbering{arabic}

\setcounter{page}{1}

\setcounter{page}{1}

\section{Introduction}
The minimum spanning tree problem is among the most well-studied algorithms problems \cite{cormen2022introduction,karger1995randomized,chazelle2000minimum,pettie2002optimal}. In \textbf{minimum spanning tree (MST)} we are given an $n$-node graph $G = (V,E)$ and an edge weight function $w : E \to \mathbb{Z}_{\geq 0}$. Our goal is to output a spanning tree $T$ of $G$ connecting all nodes of $V$ of minimum total edge weight $w(T) := \sum_{e \in T}w(e)$. %\textbf{Steiner tree} is identical to MST but we are also given a terminal set $S \subseteq V$ and our output tree $T$ must only connect all nodes of $S$.

A great deal of work has focused on MST because it formalizes how to cheaply design a simple communication network (with weights treated as costs). However, while MST guarantees connectivity between nodes, connectivity alone may not be sufficient for fast and reliable communication. For example, if sending a message over an edge incurs some latency and two nodes are very far in the MST, then communication may be prohibitively slow. Likewise, if sending a message over an edge has some failure probability, then long communication paths lead to more chances for failure and therefore less reliable communication \cite{akgun2011new, de2018extended}.

By forcing output networks to communicate over short paths, length-constrained MST aims to address these shortcomings.\footnote{This problem has appeared under many names in the literature, including: hop-bounded MST \cite{haeupler2021tree, chekuri2024approximation}, bounded diameter MST \cite{ravi1998bicriteria,kapoorsarwat2007bdmst} and shallow-light tree \cite{kortsarz1999shallowlight}.} \textbf{Length-constrained MST} is identical to MST but we are also given a root $r\in V$, edge lengths $l : E \to \mathbb{Z}_{\geq 0}$ and a length constraint $h \geq 1$, and every node in our output spanning tree $T$ must be at distance at most $h$ in $T$ from $r$ according to $l$.\footnote{Another way of defining these problems has no root and requires that all pairs of nodes are at distance at most $h$ from each other; all of our results hold for this version with an extra factor of $2$ in the length guarantee.} %\textbf{Length-constrained Steiner tree} is identical but we are also given a terminal set $S \subseteq V$ and we must only guarantee that each node in $S$ is at distance at most $h$ in $T$ from $r$ according to $l$.

Unfortunately, length constraints make problems significantly harder. First, length constraints destroy much of the structure that makes classic problems feasible. For example, the natural notion of distance to consider in MST (the shortest path distances where the distance between nodes $u$ and $v$ is the minimum weight of a $u$-$v$ path according to $w$) forms a metric. On the other hand, the natural notion of distance to consider in length-constrained MST (the length-constrained distances where the distance between nodes $u$ and $v$ is the minimum weight of a $u$-$v$ path according to $w$ that has length at most $h$ according to $l$) is inapproximable by a metric \cite{haeupler2021tree}. Second, length constraints cause algorithmic hardness: while MST is poly-time solvable \cite{karger1995randomized}, length-constrained MST has no poly-time $o(\log n)$-approximation (under standard assumptions) \cite{naor1997retractedshallowlight}.

Despite these difficulties, there has been considerable work on length-constrained MST over the past several decades. These works have studied length-constrained MST in general graphs \cite{kortsarz1999shallowlight, charikar1999approximation}, metrics \cite{althaus2005approximating}, when the length constraint is a small constant \cite{bar2001generalized, alfandari1999approximating} and length-constrained MST in the Euclidean plane \cite{laue2008approximating,fakcharoenphol2025ptas}. More generally, there are many recent works on length-constrained network design problems \cite{boehm2022hop,khani2016improved,hajiaghayi2009approximating,chimanispoerhase2014shallowlight, haeupler2021tree,filtser2022hop,chekuri2024approximation,haeupler2023maximum,ghaffari2021hop, haeupler2022hop,haeupler2024new,haeupler2023parallel,haeupler2024dynamic,haeupler2025length,haeupler2024low}. We discuss these in detail in \Cref{sec:relWork}. For length-constrained MST, remarkably, a poly-time $O\left(n^{\eps}\right)$-approximation for any constant $\eps > 0$ of \cite{charikar1999approximation} is the best poly-time approximation more than two decades on and no poly-log approximations are known. 

However, perhaps surprisingly, poly-time poly-log approximations are known for length-constrained MST \emph{if we allow for a large slack in the length}. In particular, we say that an algorithm for length-constrained MST has \textbf{length slack} $s \geq 1$ if the returned solution guarantees that every node is at distance at most $h s$ from $r$ according to $l$. Then, \cite{ravi1998bicriteria} gave a poly-time $O(\log n)$ approximation with $O(\log n)$ length slack. More recently, \cite{hershkowitz2024simplelengthconstrainedminimumspanning} gave a poly-log approximation with length slack $O(\log n / \log \log n)$. Indeed, length slack is often necessary for efficient algorithms in length-constrained graph algorithms: even just poly-time computing exact length-constrained distances for real-valued lengths is NP-hard  \cite{ashvinkumar2025faster,garey2002computers}. However, to-date every poly-log approximation for length-constrained MST requires super-constant length slack. 

\noindent This leads us to the central question of this work.
\begin{quote}\centering
    \textit{ Is there a non-trivial graph class which admits poly-time poly-log approximations for length-constrained MST with constant length slack?}
\end{quote}
In particular, the focus of this work is planar graphs; note that length-constrained MST is still NP-hard even when restricted to planar graphs \cite{friggstad2023planardst}.

% \begin{quote}\centering
%     \textit{\textbf{Question 2.} Is length-constrained MST easier on planar graphs?}
% \end{quote}

\subsection{Our Contributions}
We settle this question by giving (what is essentially) a poly-time $O(\log n)$-approximation with $1+\epsilon$ length slack for length-constrained MST in planar graphs. This is summarized by our main result, as follows.\footnote{We say that function $f(n)$ is $\poly(n)$ if there exists a constant $c$ such that $f(n) = O\left(n^c\right)$.}
\begin{tcolorbox}
\begin{restatable}{theorem}{thmmainreal}\label{thm:planarmain2}
    For any constant $\eps>0$, there is an $O\left(\log^{1+\eps}n\right)$-approximation for length-constrained MST with length slack $1+\eps$ in planar graphs running in time $\poly(n) \cdot n^{O\left(1/\eps^2\right)}$. 
\end{restatable}
\end{tcolorbox}
\noindent Our algorithm extends to various additional settings. It gives an $O(\log n)$ approximation with length slack $1+o(1)$ in quasi-polynomial time (\Cref{sec:qpt}). Likewise, a slight modification gives an algorithm for the length-constrained Steiner tree problem (\Cref{sec:steiner}) which is identical to length-constrained MST but where we must connect some terminal set $U \subseteq V$ to $r$ with length at most $h$ rather than all vertices. Furthermore, we show that our techniques can be adapted to give an $O(\log ^ 3 n)$-approximation with length slack $1+\eps$ that is competitive with the natural linear program (LP); see \Cref{sec:intgap}. 
% This too can be adapted to bound the integrality gap for length-constrained Steiner tree with $1+\epsilon$ length slack, but with a slightly worse approximation of $O(\log^3 |U|)$.
This bounds the integrality gap of the LP as poly-logarithmic (with length slack $1+\eps$) which can be contrasted with the fact that in general graphs the LP has a polynomial integrality gap \cite{zosin2002directed,li2024polynomial}. 
% using a much simpler variant of our techniques (\Cref{thm:intgap}).

Complementing our algorithm, we observe that such a result is impossible in general graphs.
\begin{tcolorbox}
    \begin{restatable}{theorem}{thmlowerbound}\label{thm:newlowerbound}
    For every fixed $\eps>0$, length-constrained MST cannot be $O\left(\log^{2-\eps}n\right)$-approximated in poly-time with length slack $s$ for $s < 2$ unless $\textsf{NP}\subseteq \textsf{ZTIME}\left(n^{\poly(\log n)}\right)$.
\end{restatable}
\end{tcolorbox}
%\noindent% Notably, this surpasses the previously-mentioned impossibility of \cite{naor1997retractedshallowlight} which ruled out $o(\log n)$ approximations for length-constrained Steiner tree.
\noindent Thus, our results  show that planar length-constrained MST is formally easier than length-constrained MST on general graphs (under standard complexity assumptions). 

\subsection{Overview and Intuition}
In this section, we give an overview of the challenges in obtaining our results and our techniques.

\subsubsection{Algorithm Intuition}
We begin with an overview of our algorithm.

\paragraph{Algorithm High-Level.} At a very high-level, the idea of our algorithm is to reduce solving our instance of planar length-constrained MST to a collection of instances of (non-planar) length-constrained Steiner tree at a cost of $1+\eps$ in length slack and an additive cost of $O(\log n)\cdot \OPT$ where $\OPT$ is the weight of the optimal length-constrained MST. These instances of length-constrained Steiner tree will be such that:
\begin{itemize}
    \item For each considered instance of length-constrained Steiner tree, the optimal length-constrained MST of the original graph (appropriately restricted to the instance) is feasible.
    \item No edge of the optimal solution appears in more than  $O(\log n)$ of the instances of length-constrained Steiner tree.
    \item Each considered instance of length-constrained Steiner consists of at most $O(\log n)$ terminals.
\end{itemize}
The utility of the last fact is that the aforementioned $O\left(n^{\eps}\right)$-approximation of \cite{charikar1999approximation} for length-constrained MST also gives an $O\left(t^{\eps}\right)$-approximation for length-constrained Steiner tree where $t$ is the number of terminals. Thus, if we apply the algorithm of \cite{charikar1999approximation} to our instances of length-constrained Steiner tree, its guarantees along with the above gives an $O\left(\log ^{1+\eps} n \right)$-approximation.
Of course, the real challenge here is to construct a collection of instances satisfying the above.

\paragraph{Why Classic Planar Separators Don't Work.} 
If our goal is a collection of instances in which each edge of the optimal solution participates in at most $O(\log n)$ instances, a natural strategy is that of \emph{planar separators}. In particular, it is known that, given any spanning tree of a planar graph, there exists a fundamental cycle of this tree (where the edge in the fundamental cycle but not in the tree may or may not be in the graph) where at most $2/3$ of the graph lies ``inside'' the cycle and at most $2/3$ of the graph lies ``outside'' the cycle. Even better, since this holds for any spanning tree, one can use a shortest path tree so that the cycle can be covered by two shortest paths \cite{lipton1979separator}. This cycle can then be used as a ``separator'' for divide-and-conquer approaches where we repeatedly recurse on the inside and outside of the cycle. As each edge will appear in $O(\log n)$ recursive calls, this gives a natural approach for $O(\log n)$-approximations \cite{klein2013structured}.

What typically makes such divide-and-conquer approaches feasible for graph problems with just edge weights is that one is free to fully use the shortest paths which cover the cycle in the output solution. In particular, any shortest path has weight at most the diameter of the graph---that is, the maximum distance between two nodes in the graph---and the diameter is typically a lower bound on the optimal solution (up to constants). As such, one can include these shortest paths in the output solution while remaining competitive with the optimal solution. Likewise, if one must guarantee a connection between two nodes of bounded weight, it is typically safe for an algorithm to connect these two nodes along a shortest path.

However, these sorts of guarantees do not clearly carry over to length-constrained MST where we have both edge weights and lengths. In particular, we would have to choose a separator which is covered by shortest paths with respect to weights or lengths, not both. However, it is easy to see that any path which is shortest according to weights has weight which (within a factor of $2$) lower bounds the weight of the optimal length-constrained MST. However, such a path could have length far greater than $h$ and so be much too long to use in our solution. Likewise, (assuming our input is feasible) a shortest path according to lengths certainly has length at most $2h$ and so is roughly on the right order in terms of length but could have weight far greater than that of the optimal solution. Rather, what is needed is a separator which can be covered with paths that have both low weight and low length.

% However, this is much too weak of a lower bound. This is most easily seen by the fact that one can, without changing the optimal solution, always add a weight $0$ Hamiltonian path to the graph of length $\gg h$ edges so that all weight-shortest paths have weight $0$ but consist of edges which cannot be used. Similarly, $h$ may be much smaller than $O(\sqrt n)$ and so it is not clear that a separator consisting of $O(\sqrt{n})$ nodes is even of the right ``scale'' for length-constrained MST. Rather, what is needed is a separator of length roughly $h$ whose weight is comparable to that of the optimal solution.

\paragraph{Length-Constrained Planar Separators.} We give such a separator in the form of $h$-length separators. As with classic separators, our length-constrained separators break the graph into two parts each of which contain at most $2/3$ of the input graph. Like classic separators, our length-constrained separators can be covered by paths whose weight (up to constants) lower bound the optimal solution. However, unlike classic separators, the lengths of the paths covering our separators is at most $O(h)$. %Much of the challenge, here, is in finding the right definition.

Specifically, the sense in which the weight of the paths covering our separators lower bound the weight of the optimal solution is somewhat different than the classic setting. In particular, lower bounding the optimal solution by way of the diameter of the graph (according to weights) cannot work in the length-constrained setting. In particular, one can always add a weight $0$ Hamiltonian path to the graph of length $\gg h$ edges which does not change the optimal solution's weight but makes the diameter (according to weights) $0$. Instead, what is needed for such a separator is a notion of distance that more tightly lower bounds the optimal solution. Such a lower bound can be gotten by $h$-length-constrained distances, defined as follows. Say that a path is $h$-length if its length according to $l$ is at most $h$. Then, the \textbf{$h$-length-constrained distance} between two nodes is the minimum weight of an $h$-length path between them and the $h$-length-constrained diameter $D^{(h)}$ is defined as the maximum $h$-length-constrained distance between two nodes. We will also say that an $h$-length path is $h$-length shortest if it is a minimum weight $h$-length path between its endpoints. In addition to having length $O(h)$, the paths covering our separators will have weight at most $O\left(D^{(h)}\right)$; it is easy to see that $D^{(2h)}$ lower bounds the weight of the optimal length-constrained MST and so (up to doubling $h$) the weight of our separators lower bound the weight of the optimal solution.

At this point, it might seem easy to get length-constrained planar separators: just take a shortest path tree with respect to $h$-length-constrained distances (rather than shortest path distances as in the classic case) and apply the aforementioned cycle separator. Unfortunately, shortest path trees with respect to $h$-length-constrained distances do not exist. Namely, it is easy to construct examples of (even planar) graphs where any rooted spanning tree contains root to leaf paths which are not $h$-length shortest; see \Cref{fig:noST}.

\begin{figure}[H]
    \centering
    \begin{subfigure}[t]{.32\textwidth}
      \centering
        \includegraphics[width=.9\linewidth]{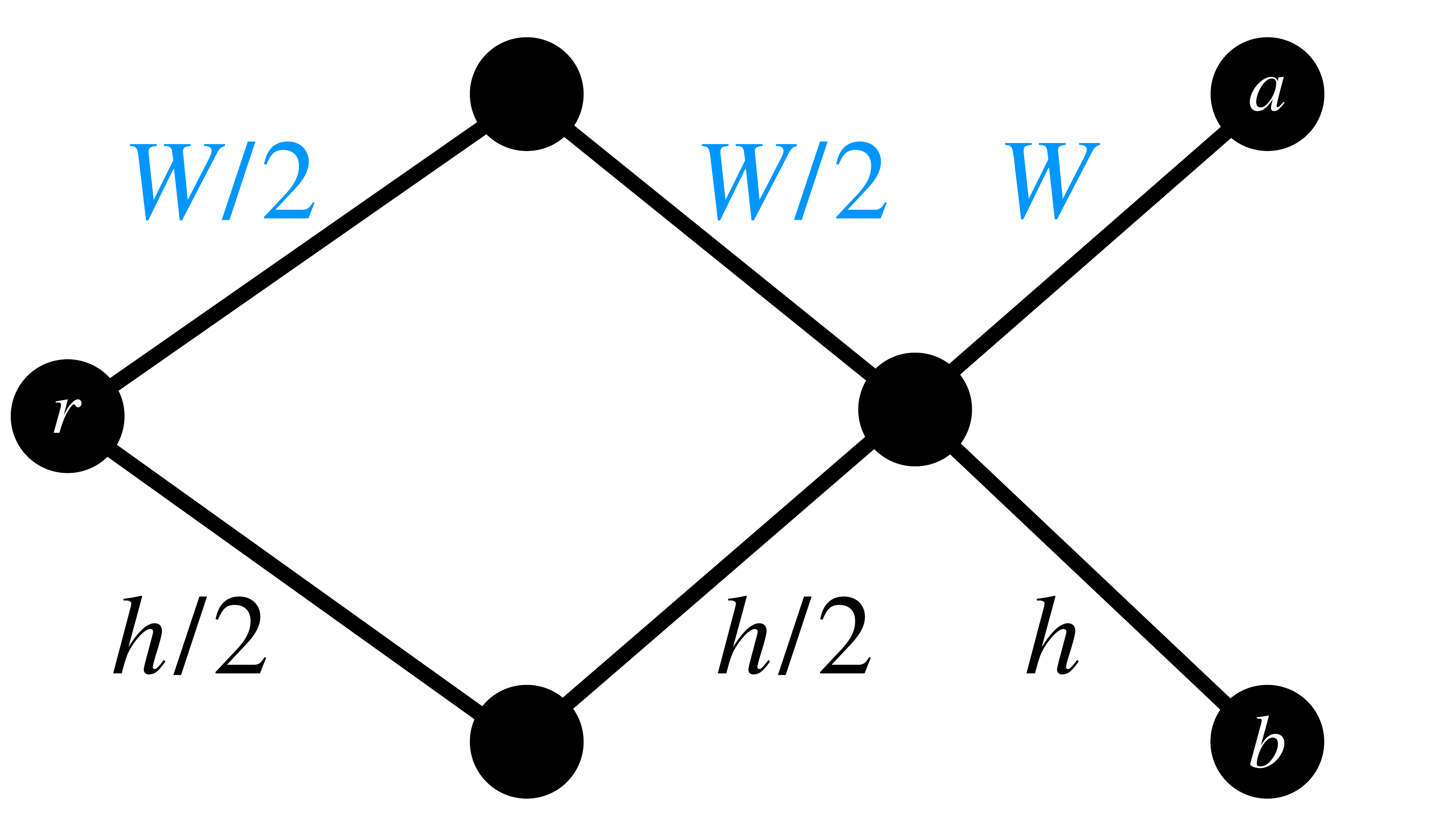}
      \caption{Input graph.}\label{sfig:noST1}
    \end{subfigure}%
    ~
    \begin{subfigure}[t]{.32\textwidth}
      \centering
      \includegraphics[width=.9\linewidth]{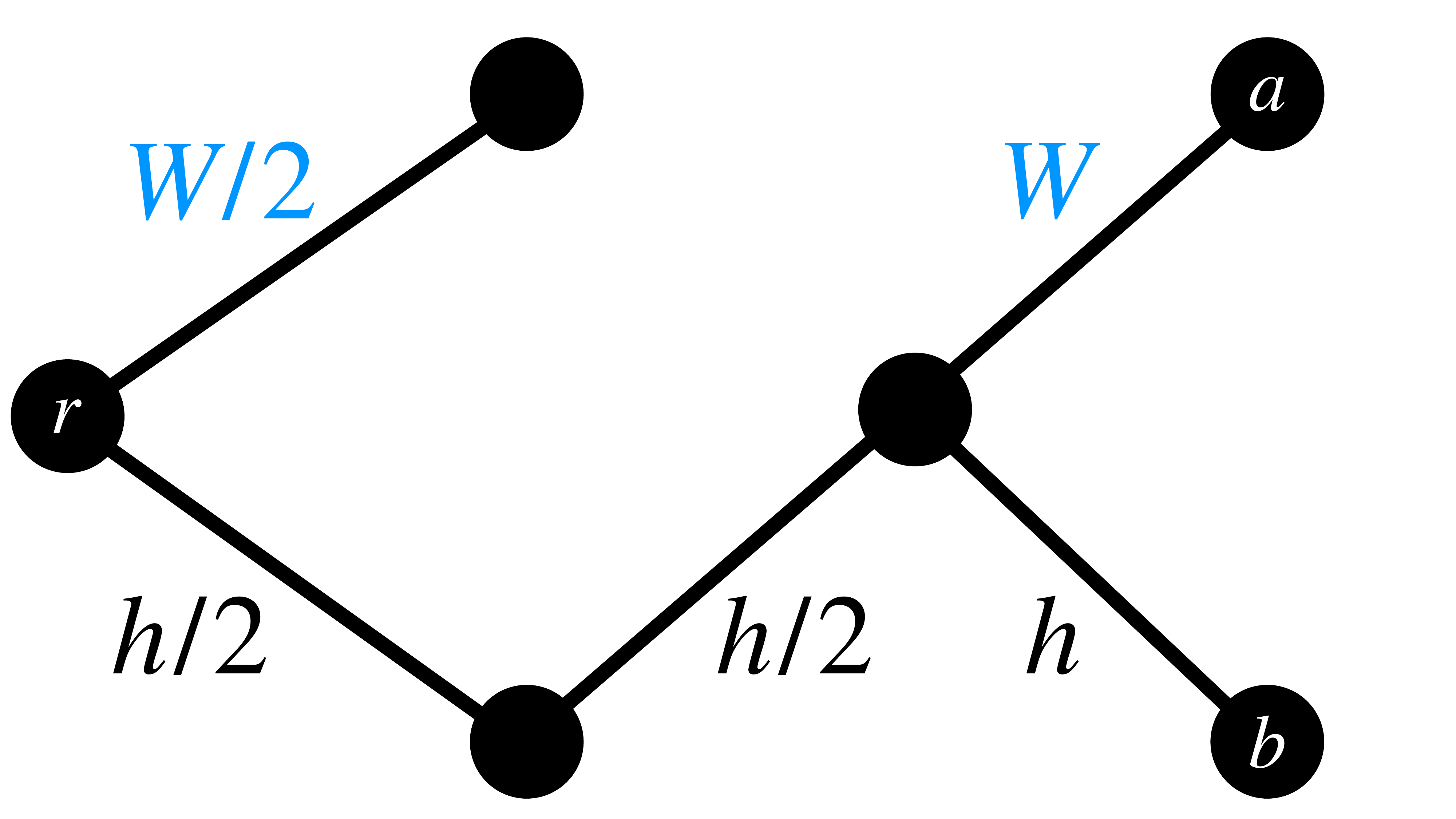}
      \caption{$r$ to $a$ $h$-length path.}\label{sfig:noST2}
    \end{subfigure}%
    ~
    \begin{subfigure}[t]{.32\textwidth}
      \centering
      \includegraphics[width=.9\linewidth]{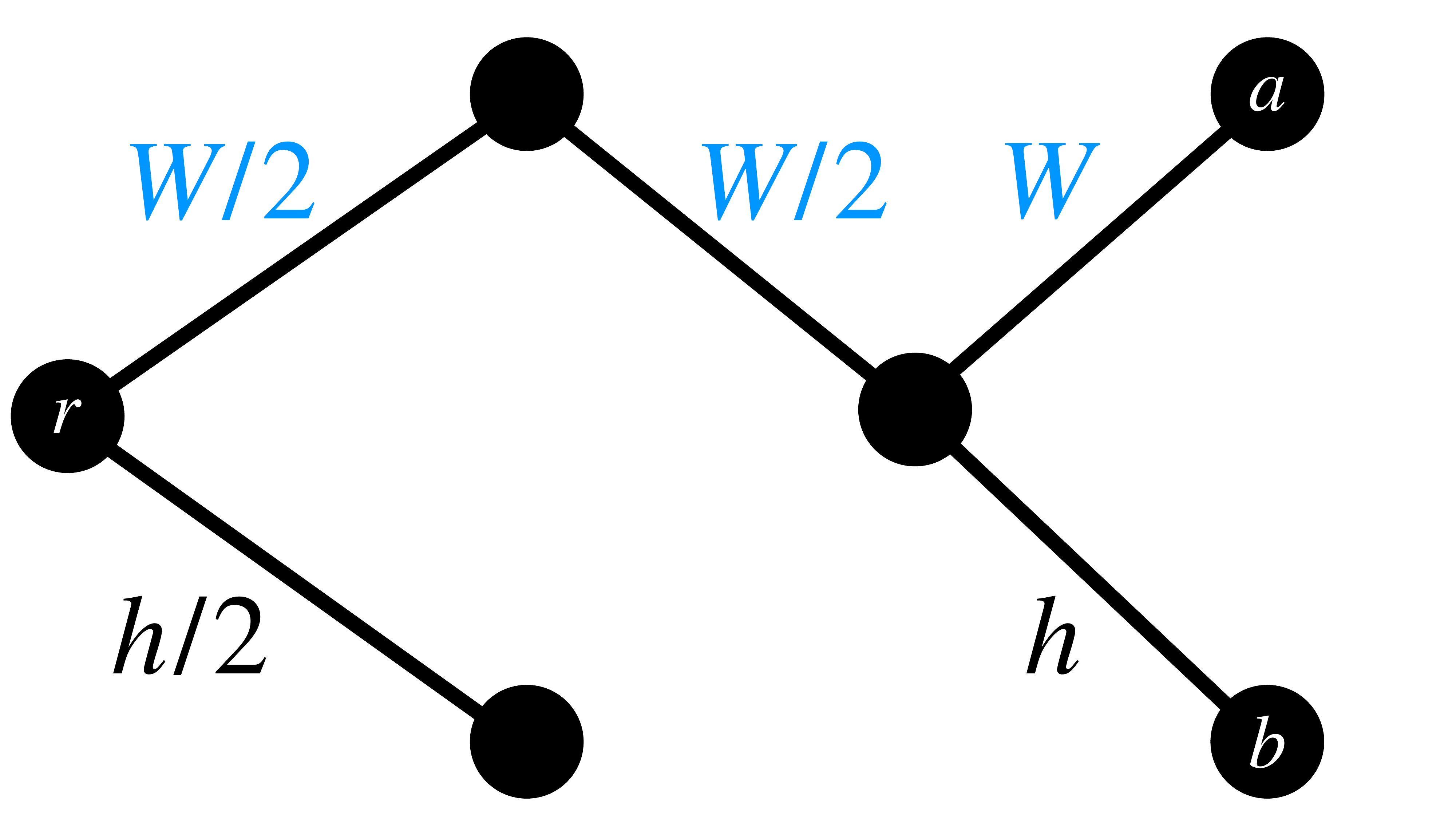}
      \caption{$r$ to $b$ $h$-length path.}\label{sfig:noST3}
    \end{subfigure}%
    \caption{A graph in which no spanning tree contains only $h$-length shortest paths from the root $r$ for any $h, W > 0$. Edge lengths in black; edge weights in blue; all unlabeled edge lengths and weights $0$. In \ref{sfig:noST2} the unique $r$ to $a$ $h$-length shortest path is included, preventing the inclusion of the unique $r$ to $b$ $h$-length shortest path. \ref{sfig:noST3} gives the flipped case.}
    \label{fig:noST}
\end{figure}

Nonetheless, we show, that such trees approximately exist in an appropriate sense. In particular, there exist spanning trees all of whose root to leaf paths are $2h$-length with weight at most $2D^{(h)}$ (\Cref{lemma:mixture}). Applying the classic planar cycle separator to such trees then gives us our length-constrained separators. We prove the existence of and give algorithms for finding these trees by taking a careful linear combination of lengths and weights to get a single new weight function with respect to which we then take a shortest path tree (\Cref{defn:mixturemetric}); this is similar to the ``mixture metric'' of \cite{haeupler2021tree}.

Summarizing, we show how to compute separators which break the graph into parts that are at most $2/3$ the size of the original graph and which can be covered by paths of length $O(h)$ and weight $2D^{(h)}$. See \Cref{defn:lc-sep} for the formal definition of our length-constrained separators and \Cref{lemma:lengthseparator} for the algorithm that computes them.

\paragraph{Length-Constrained Planar Divisions.}
We next use our length-constrained separators to construct length-constrained versions of divisions of planar graphs. Roughly, an $\alpha$-division is a separator that breaks the graph into $\alpha$ parts or ``regions'', each with a $1/\alpha$ fraction of the size of the input graph. Typically, an $\alpha$-division is gotten by taking a planar separator and then recursing on the inside and outside of the separator to recursive depth $\log \alpha$.

However, doing so with our length-constrained separators is unsuitable for our purposes as far as length is concerned. In particular, what we ultimately need for our algorithm is a division where the parts of the separator to which each region is adjacent---a.k.a.\ the boundary of that region---can be covered by edges of total length $O(h)$. However, if we just naively apply length-constrained separators to recursive depth  $\log \alpha$, a region can end up adjacent to $(\log \alpha)$-many length-constrained separators and so, it seems, all one can say is that its adjacent separators can be covered by length $O(h \cdot \log \alpha)$.

We get around this by using the fact that planar separators can be made to work for general node weightings. In particular, we alternate between separators which break up the graph itself and separators which reduce the total length it takes to cover the boundary. Since each time we apply a separator we increase the length it takes to cover the boundary by $h$ and each time we separate the separators we reduce the total length by $2/3$ multiplicatively, we end up with regions whose boundaries can be covered by edges of total length $O(h)$, as desired. As far as weight is concerned, we show that (after appropriately modifying edge lengths), the length-constrained diameter never increases as we take our separators and so the total weight of our length-constrained $\alpha$-division is at most $O\left(\alpha \cdot D^{(h)}\right)$. %Additionally, we ultimately need that the number of connected components of the boundary is at most $O(1)$ which we can likewise guarantee by interleaving one more separator with an appropriate node-weighting. 

Summarizing, we give a length-constrained $\alpha$-division which breaks a planar graph up into $\alpha$ regions each consisting of a $1/\alpha$-fraction of the graph where there is a collection of edges of total weight $O\left(\alpha \cdot D^{(h)}\right)$ that covers the boundaries of all regions together and which on each region's boundary has  total length $O(h)$. See \Cref{defn:LC-division} for the formal definition of our divisions and \Cref{lemma:lc-div} for the algorithm that computes them.

\paragraph{Length-Constrained Planar Division Hierarchies.}
Next, we use our length-constrained divisions to build a tree-type hierarchy of length-constrained $\alpha$-divisions of depth $O\left(\log_\alpha n\right)$. To do so, we simply recursively apply our length-constrained $\alpha$-divisions with the lengths of edges on the boundary zero-ed out. Using the fact that length-constrained diameter lower bounds $\OPT$, we end up with a collection of regions organized in a tree and a collection of edges which cover the boundaries of all of these regions where these edges have weight $O(\alpha \cdot \log n \cdot \OPT)$ and these edges when restricted to the boundary of a single region have length at most $O(h)$. See \Cref{defn:hierarchy} for the definition of our hierarchies and \Cref{lemma:hierarchy} for their algorithm.

\paragraph{Breaking Up Region Boundaries.}
Our goal now becomes to use our hierarchy to solve length-constrained MST. Our algorithm will use all of the edges which cover the boundaries of all of our regions of our $\alpha$-division hierarchy with $\alpha = O\left(\log^ \eps n\right)$. This costs us only $O(\alpha \log n)\cdot \OPT = O\left(\log^{1+\eps} n\right) \cdot \OPT$ and so is consistent with our approximation guarantee but may look inadequate as far as length is concerned: in particular, traversing the boundary of a single region in our hierarchy may cost us $O(h)$ in length and since the hierarchy has depth $O\left(\log_\alpha n\right)$, it looks like this will cost us a prohibitive $O\left(h \cdot \log_\alpha n\right)$ overall in length.\footnote{In fact, using the edges covering the boundary of a length-constrained $\alpha$-division hierarchy as a solution for length-constrained MST immediately gives an $O(\alpha \cdot \log n)$ approximation with $O\left(\log_\alpha n\right)$ length slack.}

We get around this by breaking apart the edges covering each boundary into smaller length ``pieces'' where we limit traversal to be internal to each of these pieces. In particular, we break each boundary's edges into $O(\beta)$ pieces, each with diameter (according to $l$) at most $h/\beta$ for $\beta \approx \frac{\log n}{\eps^2 \cdot \log \log n }$ (see \Cref{lemma:low-diam-partition}). Since the depth of our hierarchy is $\log_\alpha n \approx \log n / (\eps \log \log n)$, we are free to traverse within one piece at each level of our hierarchy on the way to the root and this only costs us at most an extra $\frac{h}{\beta}\cdot \log_\alpha n \approx \frac{h \eps^2 \cdot \log \log n}{\log n} \frac{\log n}{\eps \log \log n} \leq h\eps$ overall in length.

\paragraph{Reducing to Length-Constrained Steiner Tree with Few Terminals.} Lastly, we solve a series of instances of length-constrained Steiner tree to connect the boundary pieces of a parent to each of the boundary pieces of each of its children. In particular, for each region in our hierarchy and each of its child regions, we solve an instance of length-constrained Steiner tree where we treat each boundary piece of the child as a terminal and the pieces of the parent as the root of our instance. 

Furthermore, we set up this instance in such a way that the optimal weight of this instance is no more than the weight of the optimal length-constrained MST restricted to this region. Since the number of boundary pieces is $O(\beta) \approx \log_\alpha n/\eps \leq O(\log n)$, the algorithm of \cite{charikar1999approximation} gives an $O\left(\log ^\eps n\right)$-approximation for this instance of length-constrained Steiner tree. Also, since each level of our hierarchy partitions the optimal length-constrained MST, we pay $O\left(\log_\alpha n \cdot \log ^{\eps} n\right) \cdot \OPT = O\left(\log ^{1+\eps} n \right)\cdot \OPT$ overall for all of these instances. See \Cref{sec:chandInstance} for a description of these instances and note that these instances may not be planar.

Correctly setting up each of these instance of length-constrained Steiner tree requires that we have a reasonably accurate guess for the length the optimal length-constrained MST takes to reach each boundary piece. In particular, we guess this length for each boundary piece in multiples of $h/\beta$ for $\beta$-many guesses per piece. Guessing with this coarseness again loses us about an $h/\beta$ at each level of our hierarchy and so $\frac{h}{\beta} \cdot \log_\alpha n \approx h\eps$ overall in terms of length. Since each child's boundary is broken up into $\beta\approx \frac{\log n}{\eps^2 \log \log n}$ pieces, we must simultaneously guess this length for all $\beta$ pieces of each child for a total number of $\beta ^ \beta \approx n^{1/\eps^2}$ guesses per child, giving us a polynomial runtime overall for any constant $\epsilon$. Using a standard dynamic programming approach then allows us to perform at least as well as the ``correct'' guess. 

See \Cref{sec:DP} for the formal definition of the dynamic program and \Cref{sec:realalgdesc} for our final algorithm. The aforementioned  quasipolynomial-time algorithm with stronger approximation guarantees can be gotten by modifying $\alpha$ and $\beta$; see \Cref{sec:qpt}. Lastly, we show that a similar-in-spirit but much simpler algorithm that avoids the use of length-constrained divisions and dynamic programming achieves an LP-competitive $O(\log^2 n)$-approximation with $1+\eps$ length slack; see \Cref{sec:intgap}.

\subsubsection{Hardness of Approximation Intuition}
Our hardness of approximation is based on a very simple reduction from the group Steiner tree problem which we can describe nearly in full here. In group Steiner tree, we are given a graph and groups of vertices and a root $r$ and must connect at least one vertex from each group to $r$. In order to reduce group Steiner tree to length-constrained MST, we pull each node in each group away by a length $h$ and weight $0$ edge and connect each group together with  weight $0$ and length $0$ edges. This means that once a solution reaches one node in a group, it can trivially reach the rest. We then connect each original node to the root with a  weight $0$ but length $h$ edge so any solution can trivially connect all original nodes to the root but not in a way that allows it to cheat in group Steiner tree. Note that this hardness is not immediate from the reduction of DST to length-constrained MST since this reduction only holds for length slack $1$.

% Our hardness of approximation is based on an elementary reduction from the group Steiner tree.
\subsection{Additional Related Work}\label{sec:relWork}
Before moving onto our formal results, we give an overview of additional related work.

\paragraph{Details of Prior Work on Length-Constrained MST}\cite{kortsarz1999shallowlight} gave an approximation of $O\left(n^{\eps} \exp(1/\eps)\right)$ but with running time $n^{1/\eps} \poly(n)$ for $\eps > 0$. This work was later found to contain a bug by \cite{charikar1999approximation}, which fixed this bug and improved the approximation guarantee to $O\left(n^{\eps}/\eps^3\right)$\footnote{The original paper claims an $O\left(n^\eps/\eps^2\right)$-approximation, but this result was based on the initial statement of Zelikovsky's height reduction lemma in \cite{zelikovsky1997series} which had an error.}.
%Remarkably, more than two decades later, the $O\left(n^{\eps}/\eps^3\right)$-approximation of \cite{charikar1999approximation}, remains the best poly-time approximation if we must exactly respect the length-constraints.
\cite{hershkowitz2024simplelengthconstrainedminimumspanning} showed that, given any $\eps > 0$, one can $O\left(n^{\epsilon}/\epsilon\right)$-approximate length-constrained MST with length slack $O(1/\eps)$ whereas, conversely, \cite{kapoorsarwat2007bdmst} gave an $O(1/\eps)$-approximation with length slack $O\left(n^{\epsilon}/\epsilon\right)$ for any $\eps > 0$. Recently, \cite{chekuri2024approximation} gave a poly-time LP-competitive $\poly(\log n)$-approximation with $\poly(\log n)$ length slack.
We also note that \cite{naor1997retractedshallowlight} claimed to give an $O(\log n)$-approximation with length slack $2$ (for even the directed case) but later retracted this due to an error (see \cite{chimanispoerhase2014shallowlight}). We summarize these works and how they compare to our own in \Cref{fig:related}. Notably, no poly-logarithmic approximation for general graphs is currently known when allowing for at most constant length slack. 

\begin{figure}[H]
\begin{center}\renewcommand{\arraystretch}{1.25}
\begin{tabular}{|l|l|l|l|l|l|}
\hline
 Approximation & Length Slack &  Running Time &Graph Class & Citation  \\ \hline
 $O\left(n^\eps \cdot \exp(1/\eps)\right)$ &    $1$      &  $\poly(n) \cdot n^{1/\eps}$    & General  &\cite{kortsarz1999shallowlight} \\ \hline
$O\left(n^\eps /\eps^3\right)$ &    $1$      &  $\poly(n) \cdot n^{1/\eps}$    & General  &\cite{charikar1999approximation} \\ \hline
 $O(\log n)$ & $O(\log n)$    & $\poly(n)$     & General &  \cite{ravi1998bicriteria} \\ \hline
 $O(1/\eps)$ &  $O\left(n^\eps / \eps\right)$        &   $\poly(n)$   & General &\cite{kapoorsarwat2007bdmst} \\ \hline
 $O\left(n^\eps/\eps\right)$ & $O(1/\eps)$       & $\poly(n)$  &  General  &  \cite{hershkowitz2024simplelengthconstrainedminimumspanning}\\ \hline
 $O\left(\log^3 n\right)$ & $O\left(\log^3 n\right)$       &   $\poly(n)$   & General & \cite{chekuri2024approximation} \\ \hline
  $O\left(\log^{1+\eps} n\right)$ & $1+\eps$       &   $\poly(n) \cdot n^{O\left(1/\eps^2\right)}$   & Planar & \textbf{This work} \\ \hline
\end{tabular}
\end{center}\caption{A summary of work on length-constrained MST.}\label{fig:related}
\end{figure}
\noindent There are also several works which have given approximation algorithms for special cases of length-constrained MST, all with length slack $1$. \cite{ravi1998bicriteria} gave a $(1+\eps)$-approximation in poly-time for bounded treewidth graphs for constant $\eps > 0$. For the case where all edge lengths are $1$ and the input graph is a complete graph whose weights give a metric, \cite{althaus2005approximating} gave a poly-time $O(\log n)$-approximation. 
% \cite{segal2022bdmst} gave a poly-time $O\left(1+\frac{c_{\max}}{c_{\min}} \cdot \frac{h}{D(n-1)}\right)$-approximation where $D$ is the (hop) diameter of the input graph and $c_{\max}$ and $c_{\min}$ are the max and min edge weights respectively. 
For the related problem where our goal is a solution in which all nodes are at distance at most $h$, \cite{bar2001generalized} gave an $O(\log n)$ approximation, assuming all edge length are $1$ and $h \in \{4,5\}$. For this same problem if all edge weights are in $\{1,2\}$ and $h=2$ then a $5/4$-approximation in poly-time is possible \cite{alfandari1999approximating}. For the metric problem in the Euclidean plane, \cite{laue2008approximating} and \cite{fakcharoenphol2025ptas} both gave poly-time $(1+\eps)$-approximations.

\paragraph{Details on Length-Constrained Graph Algorithms.} Beyond length-constrained MST, there is a considerable and growing body of work on length-constrained graph algorithms. There is a good deal of work on the length-constrained version of many well-studied generalizations of MST \cite{boehm2022hop,khani2016improved,hajiaghayi2009approximating,chimanispoerhase2014shallowlight}. Along these lines, perhaps most notably, it is known that length-constrained distances can be embedded into distributions over trees  which gives poly-time poly-log approximations with poly-log length slack for numerous such generalizations, including length-constrained group Steiner tree, Steiner forest and group Steiner forest \cite{haeupler2021tree,filtser2022hop,chekuri2024approximation}. Beyond network design, there is recent work in length-constrained flows \cite{haeupler2023maximum} and length-constrained oblivious routing schemes \cite{ghaffari2021hop, haeupler2022hop} which aim to route in low congestion ways over \emph{short} paths. Additionally, the cost-distance and buy-at-bulk network design problems are two well-studied problems that also involve two edge weight/length functions \cite{awerbuch1997buy,chekuri2001costdist,meyerson2008cost,chekuri2024approximation}.

Lastly, there are many recent exciting developments in length-constrained expander decompositions which are, roughly, a small number of length increases so that reasonable demands in the resulting graph can be routed by low congestion flows over short paths \cite{haeupler2022hop,haeupler2024new,haeupler2023parallel,haeupler2024dynamic,haeupler2024low}. This line of work has recently culminated in, among other things, the first close-to-linear time algorithms for $O(1)$-approximate min cost multicommodity flow \cite{haeupler2024low}.

\paragraph{Directed Steiner Tree.}
Lack of progress on length-constrained MST is partially explained by the fact that length-constrained MST is closely related to directed Steiner tree (DST): it is easy to show that a poly-time $\alpha$-approximation for length-constrained MST (with length slack $1$) implies a poly-time $\alpha$-approximation for DST and vice versa (see \Cref{sec:dst-reduction}). The existence of a poly-time poly-log approximation for DST is an active area of research and major open question \cite{zosin2002directed, halperin2007integrality,rothvoss2011directed, grandoni2019log2,li2024polynomial,friggstad2024logarithmic}. Thus, a poly-time poly-log approximation for length-constrained MST (with length slack $1$) is the same open question. 

DST on planar graphs has recently proven a fruitful line of attack on this question. In particular, \cite{friggstad2023planardst} gave a poly-time $O(\log n)$-approximation for DST in planar graphs and it is known that DST in general graphs does not have an $O\left(\log^{2-\eps} n\right)$-approximation for any $\eps > 0$ under standard complexity assumptions \cite{halperin2003polyloginapx}. Their algorithm was also extended by \cite{chekuri2024directedpolymatroid} to give an LP-competitive $O(\log ^2 n)$ approximation algorithm. Thus, like our work, this shows a separation between the complexity of planar and general graph DST.  Notably, however, the reductions from length-constrained MST to DST do not preserve planarity and so it is not clear that one can, for example, invoke the approximation algorithm of \cite{friggstad2023planardst} to get algorithms for planar length-constrained MST.

\section{Preliminaries}
We begin by giving an overview of the notation and conventions we will use.

\paragraph{Graphs.} For any subgraph $H$ of a graph $G=(V,E)$ we let $V(H)$ and $E(H)$ denote the vertex and edge sets of $H$ respectively. We let $n$ be the number of vertices and $m$ be the number of edges of $G$. For a subset of vertices $U\subseteq V$ we let $G[U]$ be the induced subgraph on $U$. For two subsets of vertices $A,B$ we let $E(A,B)$ denote the set of edges with an endpoint in $A$ and an endpoint in $B$. Given a planar embedded graph $G=(V,E)$ and a cycle $C$ of $G$, we say a vertex/edge is \textit{inside} $C$ if it belongs to a component of $G\setminus C$ that is not adjacent to the outer face of $G$, and \textit{outside} of $C$ if it is not inside. A \textit{fundamental cycle} of a spanning tree $T$ of graph $G$ is a cycle $C=P_1\cup P_2\cup e_C$ where $e_C=(u,v)\in E(G)$ and $P_1,P_2$ are two simple edge-disjoint paths between $u,w$ and $w,v$ respectively in $T$ where $u,v,w\in V(G)$. A \textit{triangulated graph} is a planar graph where the addition of any edge breaks planarity. A \textit{region} of a graph $G$ is an edge-induced subgraph $H$ of $G$. A vertex in a region $H$ of $G$ is a \textit{boundary vertex} of $H$ if there exists an edge $e=(u,v)\in E(G)$ such that $e\not\in H$. We let $\bound_G(H)=\left\{v\in V(H):\exists\text{ }e=(u,v)\in E(G)\text{ s.t. }e\not\in E(H)\right\}$ be the set of boundary vertices of a region $H$ of $G$.

\paragraph{Lengths and weights.} We let $w(e)$ and $l(e)$ refer to the nonnegative weight and length of an edge $e$, respectively. For vertex-weighted graphs, we let $W(v)$ be the nonnegative weight of some vertex $v$. In other words, $w:E\to\mathbb{Z}_{\geq0},l:E\to\mathbb{Z}_{\geq0},W:V\to\mathbb{Z}_{\geq0}$. Let $l(E)=\sum_{e\in E}l(e),w(E)=\sum_{e\in E}w(e)$ for a set of edges $E$ and $W(V)=\sum_{v\in V}W(v)$ for a set of vertices $V$. We will often abuse notation and let $l(T)=l(E(T)),w(T)=w(E(T)),W(T)=W(V(T))$ for a subgraph $T$. 
% Given a graph $G=(V,E)$ we let $n:=|V|$ and $N:=\sum_{v\in V}=W(v)$ be the total node weight of $G$.

\paragraph{Distances.} We let $d(u,v)$ be the shortest path distance between $u,v$ under the length function $l$ unless otherwise stated. Sometimes we use a subscript to make clear which graph distances we are using, that is, $d_T(u,v)$ refers to the distance between $u,v$ in the subgraph $T$. The (weak) diameter of a subgraph $H$ is $D(H):=\max_{u,v\in V(H)}\left(d_H(u,v)\right)$.

Let $    \mathcal{P}^{(h)}\left(u,v\right) := \{P = \left(u, \ldots, v\right) : l(P) \leq h\}
$
be all $h$-length paths---i.e.\ all paths with length at most $h$---between $u$ and $v$. Likewise, let their $h$-length-constrained distance be
$$    d^{(h)}(u,v) := \begin{cases}
    \infty & \text{if }\mcP^{(h)}=\emptyset \\
    \min \{w(P) : P \in \mcP^{(h)}\left(u,v\right)\} & \text{o.w.}
\end{cases}
$$
and given a graph $G=(V,E)$ and root vertex $r$ we define 
$$ D^{(h)}(G) := \max_{u,v\in V}\left(d^{(h)}(u,v)\right)$$
to be the $h$-length-constrained diameter of $G$.

\paragraph{Length-Constrained Minimum Spanning and Steiner Tree.} 
In length-constrained minimum spanning tree, we are given a graph $G=(V,E)$ with a root $r\in V$ edge weight and length functions $w:E\to\mathbb{Z}_{\geq0},l:E\to\mathbb{Z}_{\geq0}$ and are required to return a spanning tree of $G$ with minimum total edge weight among spanning trees $T$ of $G$ satisfying $d_T(r,v)\leq h$ for each $v\in V$. In length-constrained Steiner tree, we are given the same input along with a terminal subset $U\subseteq V$ and are required to return a (not necessarily spanning) tree of $G$ with minimum total edge weight among trees $T$ of $G$ satisfying $d_T(r,t)\leq h$ for each $t\in U$.

For an instance of length-constrained MST with input $h$, we refer to feasible solutions as $h$-length spanning trees and optimal solutions as $h$-length MSTs, and let $\OPT$ denote the weight of an $h$-length MST of the instance. For the rest of the paper we fix a single optimal solution $T^*$ since we will need to reason about different subgraphs of a single optimal solution.

\section{Length-Constrained Planar Separators}\label{sec:sep}
In this section, we introduce length-constrained planar separators, which are separators with length bounded by $h$ and weight bounded by the $h$-length-constrained diameter. 
\begin{definition}[Balanced Separator]\label{defn:sep}
    Given a graph $G=(V,E)$ with vertex weights $W$, a subgraph $S$ of $G$ is a balanced separator if $V$ can be partitioned into three sets $A,B,V(S)$ such that 
    \begin{enumerate}
        \item \textbf{Separation}: $E(A,B)=\emptyset$.
        \item \textbf{Balanced}: $\sum_{v\in A}W(v)\leq 2W(G)/3$ and $\sum_{v\in B}W(v)\leq 2W(G)/3$.
    \end{enumerate}
\end{definition}
\noindent Our main contribution of this section is a length-constrained version of the above.
\begin{definition}[Length-Constrained Separator]\label{defn:lc-sep}
    Given a graph $G$ with edge lengths $l$, edge weights $w$, and vertex weights $W$, and $h\geq 1$, a subgraph $P$ of $G$ is an $h$-length separator if
    \begin{enumerate}
        \item \textbf{Balanced}: $P$ is a balanced separator.
        \item \textbf{Length-constrained}: $P$ is a path in $G$ of length at most $O(h)$ and weight at most $O\left(D^{(h)}(G)\right)$.
        \item \textbf{Inside/outside}: Adding an edge connecting the endpoints of $P$ creates a cycle $C$ such that all vertices of $A$ (resp. $B$) lie on the inside (resp. outside) of $C$.
    \end{enumerate}
\end{definition}
\noindent To prove that length-constrained separators exist (and that we can efficiently compute them), we use a classic planar separator result.  
%\begin{definition}[Cycle Separator]
    % , that is, $e_C=(u,v)$ and $P_1,P_2$ are two simple (edge-disjoint) paths to $u,v$ in $T$cycle $C$ of $G$ i
%    Given a graph $G=(V,E)$ and a spanning tree $T$ of $G$, a fundamental cycle $C=P_1\cup P_2\cup e_C$ of $T$ (that is, $e_C=(u,v)$ and $P_1,P_2$ are two simple edge-disjoint paths to $u,v$ in $T$) is called an $\alpha$-cycle separator if $C$ is an $\alpha$-separator and all vertices of $A$ (resp. $B$) lie on the inside (resp. outside) of $C$.
%\end{definition}
\begin{theorem}[Cycle Separator; Lemma 2 of
\cite{lipton1979separator}]\label{thm:cycle-separator}
    Given a connected triangulated planar graph $G=(V,E)$ with vertex weights $W$ and a spanning tree $T$ of $G$, one can compute a balanced separator $C$ of $G$ in $O(n)$ time such that $C=P_1\cup P_2\cup e_C$ is a fundamental cycle of $T$ and all vertices of $A$ (resp. $B$) lie on the inside (resp. outside) of $C$.
%\begin{enumerate}
 %   \item \textbf{\textit{Separation}}: $E(A,B)=\emptyset$.
  %  \item \textbf{\textit{Balancedness}}: $A,B$ have at most $2/3$ of the total node weight of $G$. 
%    \item \textbf{\textit{Cycle}}: $C=P_1\cup P_2\cup e_C$ where $P_1\cup P_2\cup e_C$ is a fundamental cycle of $T$, that is, $e_C=(u,v)$ and $P_1,P_2$ are two simple (edge-disjoint) paths to $u,v$. All vertices of $A$ (resp. $B$) lie on the inside (resp. outside) of $C$. 
%\end{enumerate}
\end{theorem}
% \noindent We remark that this version of the cycle separator is not the main result of \cite{lipton1979separator}; they go further to show that one can find a cycle separator of size $O\left(\sqrt{n}\right)$, but we do not need this size bound for our purposes (we need more length-related bounds). It is also important that our separators are made up of simple paths in the tree. Planar separators have been useful for many divide-and-conquer algorithms, as they allow the algorithm to break the graph into balanced components in each recursive step. 
\noindent In order to compute length-constrained separators, we will need a \textit{mixture metric} in addition to cycle separator. % to prove results in similar bicriteria network design problems.
\begin{definition}[Mixture Metric]\label{defn:mixturemetric}
    Given a graph $G=(V,E)$ with edge lengths $l$, edge weights $w$, %a root $r\in V$,
    an edge $e\in E$, and $h\geq1$ we let the $h$-mixture weight of $e$ be 
    $$ \wmix(e):=\frac{D^{(h)}(G)\cdot l(e)}{h}+w(e) $$
    and the $h$-mixture metric $\dmix(u,v)$ for every $u,v$ pair in $V$ be the shortest path distance metric where we use the mixture weights as the edge weights. 
\end{definition}
\noindent We note that mixture metrics have been used previously in \cite{meyerson2008cost,haeupler2021tree}. In particular, the mixture metric provides an easy way to compute paths that are simultaneously low length and low weight.
\begin{lemma}\label{lemma:mixture}
    Given a graph $G=(V,E)$ with edge lengths $l$, edge weights $w$, %a root $r$,
    and $h\geq 1$ such that $D(G)\leq h$ for any shortest $u,v$ path $P$ under the $h$-mixture metric, we have
    $ l(P)\leq 2h $
    and 
    $ w(P)\leq 2D^{(h)}(G)$.
\end{lemma}
\begin{proof}
    We show that any shortest path in the mixture metric will have mixture weight at most $2D^{(h)}(G)$, and any such path will have length at most $2h$ and weight at most $2D^{(h)}(G)$.
    
    Observe that for any pair of vertices $u,v$, there exists a $u,v$ path $P'$ with length at most $h$ and weight at most $D^{(h)}(G)$. Then any shortest path $P$ from $u$ to $v$ under the $h$-mixture metric satisfies
    \begin{align*}
        \wmix(P) &\leq \wmix(P') \\
        &=\sum_{e\in P'} \frac{D^{(h)}(G)\cdot l(e)}{h}+w(e) \notag\\
        &= \frac{D^{(h)}(G)\cdot l(P'
        )}{h}+\sum_{e\in P'}w(e) \notag\\
        &\leq D^{(h)}(G) + D^{(h)}(G)
         \\
        &= 2D^{(h)}(G) 
    \end{align*}
    where the first line is because $P,P'$ are both $u,v$ paths and $P$ is a shortest path under the mixture metric, and the fourth line is because $l(P')\leq h,w(P')\leq D^{(h)}(G)$.  Note that for $P$ we have
    \begin{align*}
        \wmix(P) &= \frac{D^{(h)}(G)\cdot l(P)}{h}+\sum_{e\in P}w(e)
    \end{align*}
    Now if $l(P')>2h$ then the first term in the equation above is strictly greater than $2D^{(h)}(G)$, and if $w(P')>2D^{(h)}(G)$ then the second term is strictly greater than $2D^{(h)}(G)$. Therefore, if $\wmix(P)\leq 2D^{(h)}(G)$ then $l(P')\leq 2h$ and $w(P)\leq 2D^{(h)}(G)$. Since $\wmix(P)\leq 2D^{(h)}(G)$ it follows that $l(P)\leq 2h$ and $w(P)\leq 2D^{(h)}(G)$. 
\end{proof}
\noindent We now have everything we need to compute length-constrained separators. 
\begin{lemma}[Length-Constrained Separator Existence and Algorithm]\label{lemma:lengthseparator}
    Given a planar graph $G$ with edge lengths $l$, edge weights $w$, vertex weights $W$, and $h\geq 1$, one can compute an $h$-length separator $P$ in polynomial time such that $l(P)\leq 4h$, $w(P)\leq 4D^{(h)}(G)$. 
\end{lemma}
\begin{proof}
    Our algorithm for finding a length-constrained separator is as follows: we compute a shortest path tree $T$ rooted at $r$ under the $h$-mixture metric, triangulate the graph 
    and apply \Cref{thm:cycle-separator} using $T$ to obtain a cycle $C=P_1\cup P_2\cup e_C$. Note that since we triangulate after fixing the tree $T$, it follows that all edges of $P_1\cup P_2$ are edges in $T$ (but $e_C$ might not be an edge in the original graph $G$). We can compute the mixture weights by running $O\left(n^2\right)$ shortest paths computations to find $D^{(h)}(G)$: for every pair $u,v\in V$ we compute the shortest $u,v$ path under $w$ restricted to paths with length at most $h$ which is possible in polynomial time by e.g.\ \cite{ashvinkumar2025faster}. Putting together the runtimes of our subroutines, we get that this is a polytime algorithm.
    
    Since $T$ is a shortest paths tree under the mixture metric, any simple path $P$ of $T$ must satisfy $l(P)\leq 2h,w(P)\leq2D^{(h)}(G)$ by \Cref{lemma:mixture}, so this holds for $P_1,P_2$. We can concatenate $P_1,P_2$ to obtain a single path $P$ and the lemma follows. 
\end{proof}

% \noindent The distinction between $P$ and the cycle $C$ containing $P$ that is obtained in the above proof is important. In particular, the edge $e_C$ may not even exist in the original input graph, so if our algorithm is based on buying edges of these separators then it cannot buy $e_C$. However, $C$ nicely defines the two sides of the separator, so it is also important to define it explicitly. For the remainder of the paper, we assume all separators are length-constrained.
%It is easy see that length-constrained separators immediately give a $\log n,\log n$ algorithm for the problem in planar graphs. Basically, the idea is that the balancedness from the separation allows for $\log n$-depth recursion, and since the length-constrained separator has length and weight bounded by $O\left(h\right)$ and  $O\left(d^{\left(h\right)}\left(r,V\right)\right)\leq O\left(\OPT\right)$ respectively, if we just ``buy'' the separator for each recursive call then we pay $O\left(h\right)$ and $O\left(\OPT\right)$ for $O\left(\log n\right)$ levels of recursion.

\section{Length-Constrained Divisions and Hierarchies}
We will use a generalization of our length-constrained separators to decompose our input graph into small subgraphs with more structure beyond balanced vertex weight. \textit{Divisions} are one path towards this. Divisions are the result of taking separators to some specified recursive depth; we'll use length-constrained versions of divisions based on our length-constrained separators to hierarchically decompose the graph.

\subsection{Length-Constrained Divisions}
In order to formalize our length-constrained divisions, we extend the well-known sense of \textit{regions} of a graph to \textit{length-constrained regions}. Recall a region is just an edge-induced subgraph of the graph. Roughly, a length-constrained region is one that has a specified \textit{boundary subgraph} separating the region from the rest of the graph and satisfying other properties. 
% Borrowing their terminology, we use the following definitions:
%\begin{definition}[Region]    A \textbf{region} of a graph $G$ is an edge-induced subgraph $H$ of $G$, i.e.\ a subgraph whose edges are $H\subseteq E$ and vertices are the endpoints of each edge in $H$. \end{definition}
%\begin{definition}[Boundary vertices] A vertex in a region $H$ of $G$ is a \textbf{boundary vertex} of $H$ if there exists an edge $e=(u,v)\in E(G)$ such that $e\not\in H$, and we let $\bound_G(H)$ be the set of boundary vertices of a region $H$ of $G$.\end{definition}
\begin{definition}[Length-Constrained Region]\label{defn:lc-region}
    An \textbf{$h$-length region} of a graph $G=(V,E)$ with edge lengths $l$ is a pair $\left(H, L_H\right)$ where $H$ is a region of $G$ and $L_H$ is a subgraph of $H$ with $V\left(L_H\right)=\bound_G(H)$ satisfying
    \begin{enumerate}  
        \item \textbf{\textit{Length-constrained boundary}}:  $l(L_H)\leq O(h)$.
        \item \textbf{\textit{Few components in boundary}}: $L_H$ contains at most $O(1)$ connected components. 
        \item \textbf{\textit{Separated boundary}}: for any two connected components of $L_H$, there is no path between distinct connected components of $L_H$ in $G-\left(E(H)\setminus E\left(L_H\right)\right)$.
    \end{enumerate}
\end{definition}

\noindent Given a length-constrained region $\left(H,L_H\right)$ we refer to $H$ as the \textit{region} and $L_H$ as the \textit{boundary}. Note that $(G,\emptyset)$ is a length-constrained region of $G$. We next define a length-constrained division of a graph, which informally is a collection of length-constrained regions that ``cover'' the graph, such that each length-constrained region is balanced in vertex weight and each boundary has bounded edge weight.  
%\begin{definition}[Hole]    A \textbf{hole} is a face of a region of a planar graph $G$ that is not a face of $G$. We \textbf{triangulate a hole} by adding a fake vertex inside of it, and adding an edge of length and weight $0$ between each vertex on the hole to its fake vertex. \end{definition} 
% \begin{definition}[Division]    A \textbf{division} of a graph $G=(V,E)$ is a set of regions $\mathcal{H}_G=\{H_1,H_2,\dots,H_k\}$ such that $\bigcup_{i\in[k]}E(H_i)=E$.
    % A vertex belonging to multiple regions in a division is called a \textbf{boundary vertex}. Given a division of $G$, we let $\bound(H)$ be the set of boundary vertices belonging to region $H$, and let $\bound(G)=\bigcup_{i\in[k]}\bound(H_i)$. We let $\holes(H)$ be the set of holes belonging to region $H$ and let $\holes(\mathcal{H}_G)=\bigcup_{i\in[k]}\holes(H_i)$.  \end{definition}
% \noindent We can view a cycle separator to divide a planar graph into two regions, the inside and outside of the cycle, where the cycle serves as the boundary. 
\begin{definition}[Length-Constrained Division]\label{defn:LC-division}
     An \textbf{$h$-length $\alpha$-division} of an $h$-length region $\left(H,L_H\right)$ of a graph $G=(V,E)$ with edge lengths $l$, edge weights $w$, and vertex weights $W$ is a set 
     $$\mathcal{H}=\left\{\left(H_1,L_1\right),\left(H_2,L_2\right),\dots\right\}$$
     of $h$-length regions of $H$ satisfying 
    \begin{enumerate}
        \item \textbf{\textit{Complete}}: any $e\not\in E\left(L_H\right)\cup \left(\bigcup_{i}E(L_i)\right)$ appears in exactly $1$ region of $\mathcal{H}$ and $\bigcup_{i}E\left(H_i\right)=E(H)$.
        \item \textbf{\textit{$\alpha$-divided}}: $|\mathcal{H}|=O(\alpha)$ and each region $H_i$ of $\mathcal{H}$ satisfies $$W\left(V\left(H_i\right)\setminus\bound_H\left(H_i\right)\right) \leq \frac{W\left(V(H)\setminus\bound_G(H)\right)}{\alpha}$$ 
        % that is, the total weight of non-boundary vertices is at most a $1/\alpha$ fraction of the total weight of non-boundary vertices of $H$. 
        % \item \textbf{\textit{Few holes}}: any edge in two regions is on a hole of each, and every region has at most $O(1)$ holes. 
        \item \textbf{\textit{Light boundary}}: $w\left(\left(\bigcup_{i}L_{i}\right)\setminus L_H\right)\leq  O\left(\alpha\cdot D^{(h)}(H)\right)$. 
        % \item \textbf{\textit{Partitionable holes}}: $\mathcal{P}=\{s_1,s_2,\dots,s_k\}$ is a partition of $V(\holes(\mathcal{H}_G))$ such that $\holes(\mathcal{H}_G)[s_i]$ is a path satisfying $l(\holes(\mathcal{H}_G)[s_i])\leq h/\beta$ for each $i\in[k]$ and for any region $H$ and the subset of parts $\seg(H)\subseteq \mathcal{P}$ that intersect with $V(\holes(H))$ has size at most $O(\beta)$. 
        % \begin{enumerate}
            % \item \textbf{\textit{Small boundary}}: $|V(\mathcal{S}_H)|\leq O(h)$.
            % \item  
            % i.e.\ if $\bigcup_{i\in[j]} p_i=S[\bound(H)]$ then $j\leq O(\beta)$. 
        % \end{enumerate}
    \end{enumerate}
\end{definition}
\noindent \cite{klein2013structured} gave algorithms for computing $\alpha$-divisions with few holes, which are another type of division satisfying nice properties. Unfortunately, none of their nice properties seem helpful for length-constrained MST since they aren't motivated by length constraints, similarly to how classic planar separators aren't helpful for length-constrained MST.

To compute our length constrained divisions, we use similar ideas from the $O(n\log n)$-time algorithm (Algorithm 1) in \cite{klein2013structured} in repeatedly computing separators and reweighting vertices in order to balance other properties of the beyond total vertex weight. However, we require many modifications to cope with length constraints. Our algorithm (unsurprisingly) uses length-constrained separators instead, and introduces markedly different vertex-weighting regimes than that of \cite{klein2013structured} to achieve the properties of length-constrained regions and divisions described in the previous definitions. Additionally, we need to transform our subinstances in the following manner before recursing on both sides of our separators:
\begin{definition}[Boundary Flattening]\label{defn:h+}
    Given an $h$-length region $\left(H,L_H\right)$ of a graph $G=(V,E)$ with edge lengths $l$, and edge weights $w$ we let $H^{0}$ be the result of setting edge $e$'s length and weight to $0$ for each $e\in E(L_H)$.
    % We let $l_{H^{0(l,w)}},w_{H^{0(l,w)}}$ be the length and weight functions of $H^{0(l,w)}$, respectively (and analogous definitions for $H^{0(l)},H^{0(w)}$).
    % and if $r\not\in V(H)$ then we add a dummy root $r$ in $\bound(H)$ with arbitrary edges of length and weight $0$ such that $S[\bound(H)]$ is connected.
\end{definition}
\noindent Note that given the length-constrained region $(G,\emptyset)$ of graph $G$, we have $G^{0}=G$ since there is no boundary. Flattening the boundary is an important step, as it will allow us to relate the $h$-length-constrained diameter of smaller sub-regions with that of the original region/graph.
\begin{lemma}[Flattening Doesn't Increase Diameter]\label{lemma:shrinking-diam}
    Given an $h$-length region $\left(H,L_H\right)$ of a graph $G$ with edge lengths $l$, edge weights $w$, and vertex weights $W$, we have that $D^{(h)}(H^{0})\leq D^{(h)}(G)$. 
\end{lemma}
\begin{proof}
    We show that for any $u,v\in V(H)$, the minimum weight $h$-length $u,v$ path in $H^{0}$ never has larger weight than that of $G$. Fix some $u,v\in V(H)$ and let $P$ be the minimum-weight $h$-length $u,v$ path in $G$. If $P$ only traverses through $H$ then $P$ is also a path in $H$ with $w_{H^{0}}(P)\leq w(P)$ and we are done. 
    
    Otherwise, $P$ can be broken into subpaths $p_1,p_2,\dots,p_k$ (i.e.\ a partition of its edges) such that each part is a contiguous subpath of $P$ and the parts $p_1,p_2,\dots,p_k$ alternate between containing only edges in $E(H)\setminus E\left(L_H\right)$ and containing only edges not in $E(H)\setminus E\left(L_H\right)$. For each $i$ let $p'_i$ be $p_i$ if $p_i$ is a part containing edges in $H$. Otherwise, let $p'_i$ be the $0$-weight path containing only edges in $L_H$ between $p_i$'s endpoints; this is always possible by the separated boundary property of length-constrained regions (see \Cref{defn:lc-region}). Indeed, if there is no such path in $H^{0}$, then $p_i$'s endpoints must lie on separate connected components of $L_H$. By separated boundary, $p_i$ necessarily contains an edge of $E(H)\setminus E(L_{H})$ to connect the two connected components. Then $P'=\bigcup_{i\in[k]} p'_i$ is a possible $h$-length $u,v$ path in $H^{0}$. Clearly $w(P')\leq w(P)$. Therefore we have $D^{(h)}(H^{0})\leq D^{(h)}(G)$.
\end{proof}
\noindent We are ready to show how to compute a length-constrained division.
\begin{lemma}[Length-Constrained Division Existence and Algorithm]\label{lemma:lc-div}
    Given an $h$-length region $\left(H,L_H\right)$ of a planar graph $G=(V,E)$ with edge lengths $l$, edge weights $w$, vertex weights $W$, and $\alpha>0$, we can compute in polynomial time an $h$-length $\alpha$-division $\mathcal{H}$ of $H$. 
\end{lemma}
\begin{proof}
    We compute an $h$-length separator $P=P_1\cup P_2$ and cycle $C=P\cup e_C$ of $H$ that separates $V(H)$ into $A,B,V(C)$ using \Cref{lemma:lengthseparator}. Letting $H_A=H\left[A\cup V(C)\right],H_B=H\left[B\cup V(C)\right]$, we then recurse on the two flattened regions 
    \begin{align*}
        &\left(H_A^{0}, P\cup\left(E\left(H_A\right)\cap E\left(L_H\right)\right)\right), \\
        &\left(H^{0}_B,P\cup\left(E\left(H_B\right)\cap E\left(L_H\right)\right)\right)
    \end{align*}
    for $\frac{\log\alpha}{\log(3/2)}=O(\log\alpha)$ levels. See \Cref{fig:separate-regions}. However, we alternate between different node weighting regimes for each separator computation. Letting $i$ be the current recursive depth and $\hat{H}$ be the current region, we do the following if
    \begin{enumerate}
        \item $i\mod3=0$: we balance non-boundary vertices of $\hat{H}$ by giving them their weight under $W$.
        \item $i\mod3=1$: we balance connected components of $L_{\hat{H}}$ by contracting each connected component of $L_{\hat{H}}$ into a single vertex and giving those vertices weight $1$ and giving the rest of the vertices weight $0$. However, we uncontract these connected components after computing the separator and before recursing.   
        \item $i\mod3=2$: we balance boundary \textit{lengths} by assigning each vertex $v\in V\left(L_{\hat{H}}\right)$ weight $\phi(v)=\sum_{e\in L_{\hat{H}},
        e\ni v} l(e)$ and giving the rest of the vertices weight $0$.
    \end{enumerate}
    In other words, we alternate between separating the weight of non-boundary vertices, number of connected components of the boundary, and boundary lengths for each resulting region, all using $h$-length separators. We let the regions computed at the bottom of the recursive tree be the regions of $\mathcal{H}=\left\{\left(H_1,L_1\right),\left(H_2,L_2\right),\dots\right\}$. 

    We show that each $\left(H_i,L_i\right)\in\mathcal{H}$ is an $h$-length region of $H$. First note that when we recurse on the two regions, a vertex in either region is a boundary vertex if it is contained in an edge of both regions. This is exactly a vertex belonging to the separator $P$ we took to get these two regions, and we add this to the boundary of both regions. So $V\left(L_i\right)=\bound_H\left(H_i\right)$. We now prove that all three properties in \Cref{defn:lc-region} hold:
    \begin{enumerate}
        \item \textbf{Length-constrained.} Let the potential of a boundary $L_{\hat{H}}$ be $\sum_{v\in V(L_{\hat{H}})}\phi(v)$. So the potential of $L_{\hat{H}}$ is at most $2\cdot O(h)=O(h)$ initially by the assumption that $H$ is an $h$-length region and the length-constrained boundary property. Observe that for each separator we take we can add at most $4h$ to the potential by \Cref{lemma:lengthseparator}. Meanwhile, every third separator we take (when $i\mod 3=2$) reduces the potential by a multiplicative $2/3$ by our node-weighting and the balanced property in \Cref{defn:sep} guaranteed by \Cref{lemma:lengthseparator}. So we repeatedly add $12h$ to the potential and multiply it by $2/3$, and solving the recurrence gives that the added potential is at most $24h$ after $O(\log\alpha)$ levels. Then 
        $$l\left(L_{\hat{H}}\right)< 2\cdot l\left(L_{\hat{H}}\right)=\sum_{v\in V\left(L_{\hat{H}}\right)}\phi(v) \leq O(h)$$
        and this holds for any $L_i$.
        \item \textbf{Few components.} The number of connected components in $L_{\hat{H}}$ is initially at most $O(1)$ by the assumption that $H$ is an $h$-length region and the few components property. Each separator we take can add at most one connected component, while every third separator (when $i\mod3=1)$ reduces the number of connected commponents by a multiplicative $2/3$ by our node-weighting and the balanced property of \Cref{defn:sep} guaranteed by \Cref{lemma:lengthseparator}. So we repeatedly add $1$ connected component and multiply it by $2/3$, and solving the recurrence gives that there are at most $O(1)$ connected components in $L_{\hat{H}}$ after $O(\log\alpha)$ levels and the same holds for any $L_i$.
        \item \textbf{Separated.} Any two connected components $J,K$ of $L_i$ must have belonged to different separators that we took since a separator is a connected subgraph. If there some path between $J,K$ using no edge in $E\left(H_i\right)\setminus L_i$, then there must be an edge with an endpoint in $J$ and an endpoint outside of $H$, contradicting the separation property in \Cref{defn:sep} guaranteed by \Cref{lemma:lengthseparator}. Then any path between $J,K$ in $H$ must use an edge in $E\left(H_i\right)\setminus E\left(L_i\right)$.
    \end{enumerate}
    Now we show that $\mathcal{H}$ is an $h$-length $\alpha$-division of $H$. 
    \begin{enumerate}
        \item \textbf{Complete.} Given a separator $P$ and cycle $C=P\cup e_C$ of $H$, any edge in $H$ must either be in $P$, or on the inside/outside of $C$ but not both, so any edge that was never captured by a separator during the recursive process can be in exactly one region. These are exactly the edges not in $L\cup\left(\bigcup_{i} L_i\right)$. By definition, each edge in $L\cup\left(\bigcup_{i} L_i\right)$ is contained in some region, and we are done. 
        \item \textbf{$\alpha$-divided.} The number of regions is $2^{O(\log\alpha)}=O(\alpha)$ since we make two recursive calls for $O(\log\alpha)$ recursive levels. The total node weight of the non-boundary vertices in each region is reduced by a constant fraction (by the balanced property of \Cref{defn:sep} guaranteed by \Cref{lemma:lengthseparator}) $O(\log\alpha)$ times. This is what we need. 
        \item \textbf{Light boundary.} When we recurse on a side $A$ of a separator of region $H$, we always choose the boundary of the child region $H_A$ to be $P\cup \left(H_A\cap L_H\right)$, and by \Cref{lemma:lengthseparator}, the weight of each $P$ is at most $O\left(D^{(h)}(H)\right)$. By \Cref{lemma:shrinking-diam}, we have $D^{(h)}(\hat{H})\leq D^{(h)}(H)$ for any region $\hat{H}\in\mathcal{H}$. Then it follows that $w\left(\left(\bigcup_{i}L_i\right)\setminus L_H\right) \leq O\left(2^{O(\log\alpha)} D^{(h)}(H)\right)=O(\alpha \cdot D^{(h)}(H))$ since we make two recursive calls for $O(\log\alpha)$ recursive levels. 
    \end{enumerate}
    Lastly, a paragraph on runtime. Every step of this algorithm can be done in polynomial time. The recursive tree has depth $O(\log\alpha)\leq O(\log n)$ recursive levels and at most $n$ children since we have a leaf for each of the $O(\alpha)$ regions in the division. Each recursive call is a polynomial time operation by \Cref{lemma:lengthseparator}, so this is a polynomial time algorithm. 
\end{proof}
\begin{figure}[H]
    \centering
    \begin{subfigure}[t]{.25\textwidth}
      \centering
        \includegraphics[width=.8\linewidth]{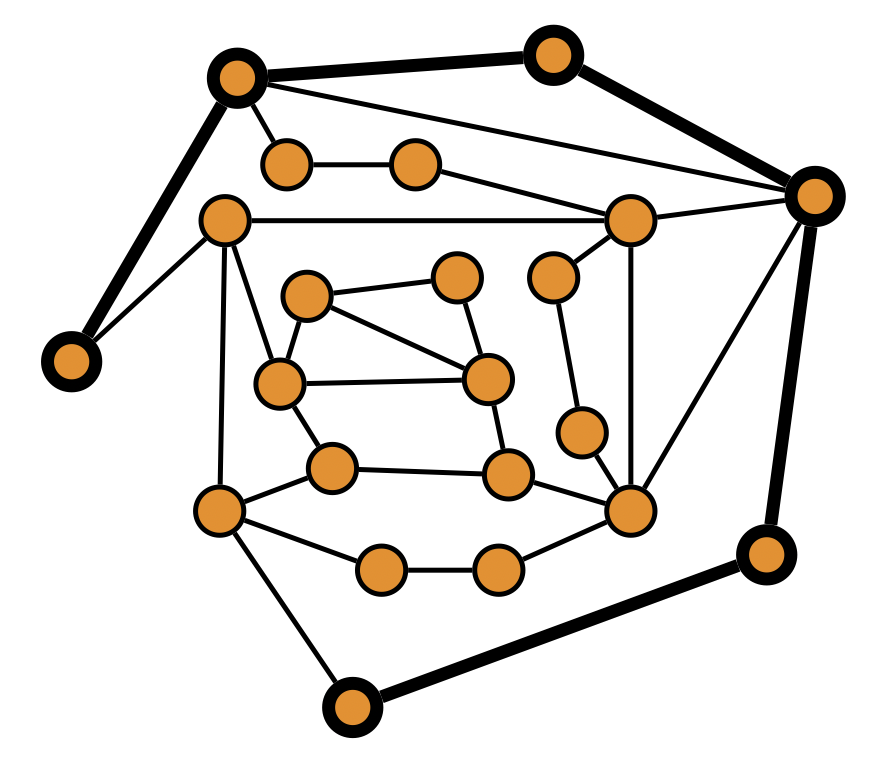}
      \caption{$(H,L_H)$.}\label{subfig:region}
    \end{subfigure}%
    ~
    \begin{subfigure}[t]{.25\textwidth}
      \centering
      \includegraphics[width=.8\linewidth]{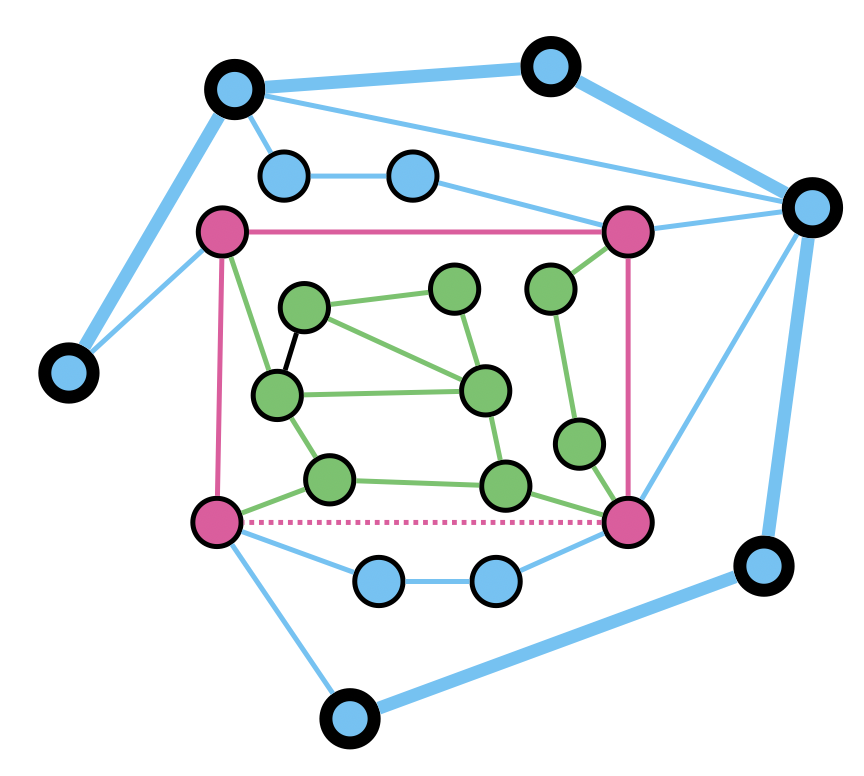}
      \caption{Separator $P$.}\label{subfig:cyclesep-region}
    \end{subfigure}%
    ~
    \begin{subfigure}[t]{.25\textwidth}
      \centering
        \includegraphics[width=.7\linewidth]{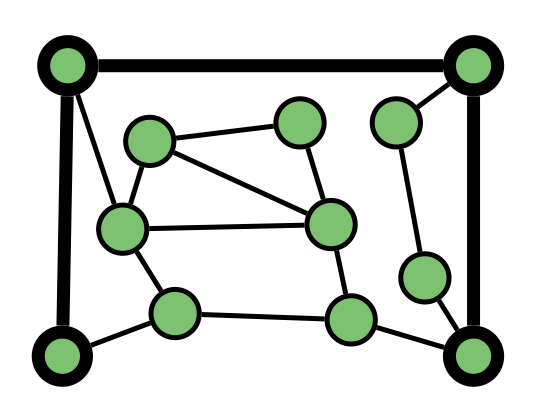}
      \caption{$(H_A,P\cup(H_A\cap L_H))$.}\label{subfig:inside-region}
    \end{subfigure}%
    ~
    \begin{subfigure}[t]{.25\textwidth}
      \centering
      \includegraphics[width=.8\linewidth]{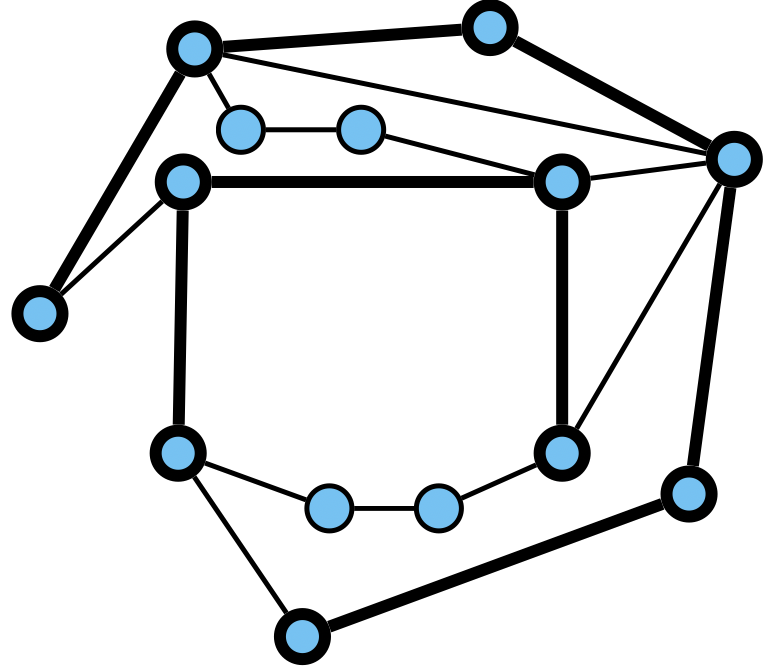}
      \caption{$(H_B,P\cup(H_B\cap L_H))$.}\label{subfig:outside-region}
    \end{subfigure}%
    \caption{A step of our algorithm to compute a length-constrained division. We're given a length-constrained region $(H,L_H)$ where $L_H$ contains the bolded edges/vertices (\Cref{subfig:region}). We take a separator $P$ which contains the pink edges/vertices, $C$ is the union of $P$ and the dotted pink edge $e_C$, and $H_A$ (resp. $H_B$) contains the green (resp. blue) edges/vertices (\Cref{subfig:cyclesep-region}). We recurse on $(H_A,P\cup(H_A\cap L_H)),(H_B,P\cup(H_B\cap L_H))$ where $P\cup(H_A\cap L_H)),P\cup(H_B\cap L_H))$ contain the bolded edges/vertices (\Cref{subfig:inside-region,subfig:outside-region}).}
    \label{fig:separate-regions}
\end{figure}

\noindent
Our main algorithm will recursively compute these length-constrained divisions (that recursively compute length-constrained separators), creating a hierarchy of length-constrained divisions. 

\subsection{Length-Constrained Hierarchies}
A hierarchy is basically a decomposition of a graph, where the leaves are sufficiently small regions.
\begin{definition}[Length-Constrained Division Hierarchy]\label{defn:hierarchy}
    An \textbf{$h$-length $\alpha$-division hierarchy} of a planar graph $G=(V,E)$ with edge lengths $l$, edge weights $w$, and vertex weights $W$ is a rooted tree $\mathcal{T}$ where leaves are $h$-length regions containing no non-boundary vertices and each edge in $E$ appears in some leaf's region. Every node is the union of the regions of its descendants in the tree, and the children of a node $v$ in the tree are the regions in an $h$-length $\alpha$-division $\mathcal{H}_v$.
\end{definition}
\noindent So the root of $\mathcal{T}$ is the length-constrained region $(G,\emptyset)$, and every child node $u$ with parent $v$ in a division hierarchy is associated with an $\alpha$-division of $v^{0}$. We will often refer to a node $\left(H,L_H\right)\in V(\mathcal{T})$ as $H$ for short.

We have the following observation immediately from the definition of division hierarchies. 
\begin{lemma}\label{lemma:hierarchy-tree}
    The depth of an $h$-length $\alpha$-division hierarchy $\mathcal{T}$ of a graph $G=(V,E)$ is at most $\log_{\alpha}W(V)$, and each non-leaf node in the hierarchy has $O(\alpha)$ children. 
\end{lemma}
\begin{proof}
    By the $\alpha$-divided property of \Cref{defn:LC-division}, the total weight of non-boundary vertices a child vertex in $\mathcal{T}$ must be at most a $1/\alpha$ fraction of that of its parent. Since all vertices of $G$ are non-boundary we have that $W(V)$ is the initial total weight of non-boundary vertices. So we can multiply $W(V)$ by $1/\alpha$ at most $\log_\alpha W(V)$ times before it drops below $1$, which is the total weight of non-boundary vertices in a leaf $\mathcal{T}$. This is exactly the depth of $\mathcal{T}$. Also by the $\alpha$-divided property of \Cref{defn:LC-division}, there are $O(\alpha)$ regions in an $\alpha$-division so any non-leaf node in $\mathcal{T}$ has $O(\alpha)$ children, one for each region in the corresponding division. 
\end{proof}
\noindent We will generally let $\operatorname{depth}(\mathcal{T})$ denote the depth of $\mathcal{T}$, and $\operatorname{child}(H)$ be the set of child nodes of a node $H$ in $\mathcal{T}$.
\subsection{From Hierarchies to Length-Constrained Minimum Spanning Tree}
We show the connection between the structures we have spent the past two sections defining and length-constrained MST, beginning with a natural definition. 
\begin{definition}[Restriction of $T$ on $H$]\label{defn:restriction}
    Given an $h$-length MST $T$ and length-constrained region $H$ of a graph $G=(V,E)$, we let $E\left(T_{|H}\right)=E(T)\cap \left(E(H)\setminus E\left(L_H\right)\right)$ and let $\OPT_{|H}=\sum_{e\in E(T_{|H})}w(e)$.
\end{definition}
\noindent This definition allows us to reasonably charge edges of $\OPT$ without double-charging edges that appear in multiple regions.
\begin{lemma}\label{lemma:restriction-sums}
    Given an $h$-length division $\mathcal{H}$ of a graph $G$  we have $\sum_{\left(H,L_H\right)\in\mathcal{H}}\OPT_{|H}\leq\OPT$.
\end{lemma}
\begin{proof}
    This follows from the complete property of length-constrained divisions (\Cref{defn:LC-division}) and the fact that $E\left(T_{|H}\right)$ only includes edges in $E(H)\setminus E\left(L_H\right)$.
\end{proof}
\noindent Recall we fixed an $h$-length MST $T^*$, so future appearance of $\OPT_{|H}$ is in reference to $E\left(T^*_{|H}\right)$.

Note that while we defined all of the previous length-constrained structures to work with general nonnegative vertex weights, it was not because the input graphs of our length-constrained MST instances are vertex-weighted (recall that the formal definition of length-constrained MST never included a vertex weight function). In fact, given an instance of length-constrained MST, we may assume the vertices are unweighted i.e.\ $W:V\to\{1\}$ and $W(V)=n$.

Since the input $h$ for length-constrained MST is really a bound on the distances from the root $r$, we need to slightly adjust our usage of the $h$-constrained diameter. In particular, note that given an $h$-length MST, any pair $u,v$ are certainly reachable by a $2h$-length path, since a possible $u,v$ path is $u$ to $r$ to $v$ and both subpaths have length at most $h$. Hence we will actually be using $2h$-length regions and divisions when we try to solve length-constrained MST, but the extra $2$ will be swallowed by big-Os. We show that we can similarly upper bound the $h$-length-constrained diameter of a region whose boundary is flattened as before, this time in the context of length-constrained MST:

\begin{lemma}[Flattening and $h$-length MSTs]\label{lemma:bounded-radius}
    Any $2h$-length region $\left(H,L_H\right)$ of a graph $G=(V,E)$ with edge lengths $l$ and edge weights $w$
    % of a $2h$-length-constrained division
    satisfies $D^{(2h)}\left(H^{0}\right)\leq \OPT_{|H}$. 
\end{lemma}
\begin{proof}
    By definition all edge weights of $E\left(L_H\right)$ in $H^{0}$ are $0$ and $E\left(T^*_{|H}\right)$ contains the edges of $T^*$ that aren't in $E\left(L_H\right)$. Therefore $E\left(T^*_{|H}\right)$ is a subset of the edges with positive weight in $H^{0}$. Let $P$ be the unique $u,v$ path in $T^*$ of length at most $2h$. We can construct a $2h$-length path $P'$ in $H^{0}$ from $P$ as follows: we can break $P$ into subpaths $p_1,p_2,\dots,p_k$, (i.e.\ partition its edges) such that each part is a contiguous subpath of $P$ and the parts $p_1,p_2,\dots,p_k$ alternate between containing only edges that are in $E(H)\setminus E\left(L_H\right)$, or not. For each $i\in[k]$ let $p'_i$ be $p_i$ if $p_i$ contains edges in $E(H)\setminus E\left(L_H\right)$. Otherwise let $p'_i$ be the $0$-weight path containing only edges in $E\left(L_H\right)$ between $p_i$'s endpoints; this is always possible by the separated boundary property of length-constrained regions. Indeed, if there is no such path in $H^{0}$, then $p_i$'s endpoints must lie on separate connected components of $L_H$. By the separated boundary property of \Cref{defn:lc-region}, $p_i$ necessarily contains an edge of $E(H)\setminus E\left(L_{H}\right)$ to connect the two connected components. Clearly, the union of the $p'_i$'s with positive weight is a subset of $E\left(T^*_{|H}\right)$. Then $P'=\bigcup_{i\in[k]}p'_i$ is a feasible $h$-length $u,v$ path in $H^{0}$ satisfying $w(P')\leq \OPT_{|H}$. This holds for any $u,v$ pair, so the lemma statement follows. 
\end{proof}
\noindent Likewise, we can relate the total edge weight across all region boundaries in a hierarchy to $\OPT$.
\begin{lemma}[Length-Constrained Division Hierarchy Existence and Algorithm]\label{lemma:hierarchy}
    Given a planar graph $G=(V,E)$ with edge lengths $l$, edge weights $w$, and
    $h\geq1,\alpha\geq1$, one can compute in polynomial time a $2h$-length-constrained $\alpha$-division hierarchy $\mathcal{T}$ of $G$ satisfying 
    $$w\left(\bigcup_{H\in V(\mathcal{T})}L_H\right)\leq O\left(\alpha \log_\alpha n \right)\cdot \OPT$$ 
    % $d_\mathcal{S'}\left(r,v\right)\leq O\left(\log n\right)\cdot h$ for every $v\in V$ and 
\end{lemma}
\begin{proof}
    To compute $\mathcal{T}$, we recursively compute $2h$-length-constrained $\alpha$-divisions until we have an instance with no non-boundary vertices. By \Cref{lemma:hierarchy-tree}, the recursive tree has depth $O\left(\log_\alpha n\right)$ and each non-leaf node has $O(\alpha)$ children, so there are $\poly(n)$ nodes in the tree. Computing an $\alpha$-division can be done in polynomial time by \Cref{lemma:lc-div}, so computing $\mathcal{T}$ can be done in polynomial time.

    We prove the statement on weight. Let $l_i$ be the set of nodes at level $i$ in $\mathcal{T}$ where $l_1$ contains the root. Then we have
    % First, two observations: \textbf{(1)} for any level $i$ in the hierarchy $\mathcal{T}$, we have $\sum_{H\in V(\mathcal{T}):H\in\text{ level }i}\OPT_{|H}\leq \OPT$ because each $E(T_{|H})$ is disjoint by definition, and \textbf{(2)} for any node $H\in V(\mathcal{T})$, we have $D^{(2h)}(\transformh\leq \OPT_{|H}$ by $\Cref{lemma:bounded-radius}$. 
    \begin{align*}
        w\left(\bigcup_{H\in V(\mathcal{T})}L_H\right)
        %&\leq \sum_{i\in[\operatorname{depth}(\mathcal{T})]}\sum_{H\in L_i} w_{\transformh}\left(\holes(H)\right) \\
        % &\leq \sum_{i\in[\operatorname{depth}(\mathcal{T})-1]} \sum_{H\in l_{i}} \sum_{\hat{H}\in\operatorname{child}(H)} w_{H\transformh}\left(L_{\hat{H}}\right) \\
        &= \sum_{i\in[\operatorname{depth}(\mathcal{T})-1]} \sum_{H\in l_{i}} w_{H}\left(\left(\bigcup_{\hat{H}\in \operatorname{child}(H)}L_{\hat{H}}\right)\setminus L_H\right) \\
        &\leq \sum_{i\in[\operatorname{depth}(\mathcal{T})-1]} \sum_{H\in l_{i}} O\left(\alpha\cdot D^{(2h)}\left(H^{0}\right) \right) \\
        & \leq \sum_{i\in[\operatorname{depth}(\mathcal{T})-1]} \sum_{H\in l_{i}} O\left(\alpha\cdot \OPT_{|H} \right) \\
        &\leq \sum_{i\in[\operatorname{depth}(\mathcal{T})-1]}  O\left(\alpha\cdot \OPT \right) \\
        &\leq O\left(\alpha\log_\alpha n\right)\cdot\OPT
    \end{align*}
    where 
    the first line is because $G$ is the only region at level $1$ and $L_G=\emptyset$ so we only need to sum over child nodes in $\mathcal{T}$, 
    % the second line is because $\left(\bigcup_{\hat{H}\in \operatorname{child}(H)}L_{\hat{H}}\right)\setminus L_H=\bigcup_{\hat{H}\in \operatorname{child}({H)}}L_{\hat{H}}$, 
    the second line is by the light boundary property of \Cref{defn:LC-division} guaranteed by \Cref{lemma:lc-div},
    the third line is by \Cref{lemma:bounded-radius}, 
    the fourth line is by \Cref{lemma:restriction-sums}, 
    and the fifth line is because there are $O\left(\log_\alpha n\right)$ levels in $\mathcal{T}$ by \Cref{lemma:hierarchy-tree}. 
    % It is clear that $\mathcal{S}$ is the same regardless of $\alpha$. In particular, $\mathcal{S}$ is the union of $O\left(\log n\right)$ levels of length-constrained separators on instances of size decreasing by a constant (excluding the non-tree edges). By induction on the recursive depth and \Cref{coro:generalLCsep}, we have $d_\mathcal{S'}\left(r,v\right)\leq O\left(\log n\right)\cdot h$ for every $v\in V$ and $w(\mathcal{S}')\leq O(\log n)\cdot \OPT$.
\end{proof}
\noindent Our main algorithm will use these hierarchies with $\alpha=\log^{\eps/2}n$. 

\section{An $O(\log^{1+\eps} n)$ Approximation Algorithm with $1+\eps$ Length Slack}\label{sec:realalg}
In this section we present our main algorithm and prove the following theorem.
\thmmainreal*
\noindent 
We start by buying a $2h$-length-constrained $\alpha$-division hierarchy of the input graph; paying an $O\left(\alpha\log_\alpha n\right)$ factor in the weight approximation seems unavoidable by \Cref{lemma:hierarchy-tree}. The natural next step is to add a set of edges of low weight that decreases the distances from the root to be competitive with an optimal solution.

Our approach is based on the following idea: let $H$ be a subgraph with connected components $C_1,C_2,\dots,C_k$ such that $D\left(C_i\right)\leq a\text{ }\forall\text{ }i\in[k]$. If there exists a (cheap) set of edges $S$ such that each $C_i$ has a vertex $v$ such that $d_{H\cup S}(r,v)\leq b$, then $d_{P\cup S}(r,v)\leq a+b$ holds \textit{for every} $v\in V(H)$. Towards this intuition, we break the boundary at each level of our hierarchy into low diameter components (\Cref{subsec:pieces}), reduce each level into an instance of length-constrained \textit{Steiner} tree (\Cref{sec:chandInstance}) where the solution of each instance serves as a candidate set of cheap edges that reduce the length of our solution, and finally we use a standard dynamic programming approach to pick the best solutions (\Cref{sec:DP}). We walk through the intuition of our algorithm in more detail in the following subsections, and the formal description and analysis are given in \Cref{sec:realalgdesc}. 
\subsection{Partitioning the Boundary into Pieces}\label{subsec:pieces}
In order to obtain a collection of low diameter connected components as described above, it is natural to try to break up the length-constrained boundary of length-constrained regions. However, a subgraph having bounded total edge length does not imply that it has low diameter. We use the following fact to help us break up the boundary into (not too many) low diameter components:
\begin{lemma}[Partitioning the Boundary]\label{lemma:low-diam-partition}
    Given a connected graph $G=(V,E)$ with $l(E)\leq h$ and $\beta\geq1$, we can find in polynomial time a $\beta$-partition $\mathcal{P}=\left\{V_1,V_2,\dots,V_k\right\}$ of $V$ such that $|\mathcal{P}|\leq O(\beta)$ and for every $i\in[k]$ we have that $D\left(G[V_i]\right)\leq h/\beta$. 
\end{lemma}
\begin{proof}
    Compute a BFS tree $T$ of $G$ rooted at an arbitrary vertex $r\in V$. Let $A\subseteq E(T)$ be the set of edges of $T$ with length strictly larger than $\frac{h}{4\beta}$. Then $T- A$ is a forest with connected components $T_1,T_2,\dots,T_f$ with $f\leq 4\beta$. This is because removing an edge from $T$ increases the number of subtrees in the forest by $1$ and there are at most $4\beta$ edges with length greater than $\frac{h}{4\beta}$ as otherwise we would have $l(E)>h$, contradicting our assumption. We root each $T_i$ at $r_i$, which is the vertex in $T_i$ with the highest level in $T$.
    
    Given a tree $H$ call a vertex $u\in V(H)$ \textit{minimally far} if it has a descendant $v$ in $V(H)$ such that $d_{H}(u,v)>\frac{h}{4\beta}$, but for any of its children $w\in V(H)$ and any descendant $v\in V(H)$ of $w$, we have $d_{H}(w,v)\leq \frac{h}{4\beta}$. For each $i\in[f]$ let $T_i(r)$ be the connected component of $T_i$ containing $r_i$ (initially $T_i=T_i(r)$ for every $i$). We now perform a greedy procedure: if there exists a minimally far vertex $v$ in $T_i(r)$ and with a parent $t$ then we remove the edge between $t,u$. Upon removing that edge, we break $T_i(r)$ into $T_i(r)$ and a new connected component that we root at $u$. By definition, the new component rooted at $u$ has no minimally far vertices with a parent. We repeat until there are no more minimally far vertices with a parent in $T_i(r)$, and iterate across all $T_i$'s. Let the set of edges we remove in this procedure be $B$. Then $A\cap B=\emptyset$ and $T-(A\cup B)$ is a forest with connected components $T'_1,T'_2,\dots,T'_{f'}$. Observe that any vertex that was ever minimally far with a parent is now a root (i.e.\ vertex with highest level in $T$) of its connected component $T'_j$ in $T-(A\cup B)$; denote it as $r'_j$ for every $j\in[f']$. 
    
    We claim that $f'\leq 4\beta$. Let $\phi_i=l\left(T_i(r)\right)$ be the total edge length of the connected component containing $r_i$ as we perform the greedy procedure and define a potential $\Phi=\sum_{i\in[f]}\phi_i$. If we remove an edge of $T_i$, then there is a minimally far vertex $u$ with parent $t$ in $T_i$. The subtree of $T_i$ rooted at $u$ must have total edge length of at least $\frac{h}{4\beta}$ by definition of minimally far vertices. So when we break off this subtree from $T_i(r)$ by removing the edge between $t,u$, then $\phi_i$ and thus $\Phi$ must go down by at least $\frac{h}{4\beta}$. We have $\Phi< h-\frac{hf}{4\beta}$ initially since we already removed $f$ edges each of length larger than $\frac{h}{4\beta}$ from $T$ in the first step. Then we can remove at most $4\beta-f$ edges across all of the $T_i$'s throughout the greedy procedure. Indeed, if we removed more than $4\beta-f$ edges, then $\Phi<0$ and thus there is some $i$ such that $\phi_i<0$, but this is clearly impossible. So we have $f'\leq f+4\beta-f=4\beta$. 
    
    For any $u,v \in V\left(T'_j\right)$ we have that the unique $u,v$ path in $V\left(T'_j\right)$ at worst goes from $u$ to $r'_j$ to $v$. Then $d_{T'_j}(u,v)\leq 2\left(\frac{h}{4\beta}+\frac{h}{4\beta}\right)=\frac{h}{\beta}$, where we get one $\frac{h}{4\beta}$ for the subpath from $u$ to a child $w$ of $r'_j$ since $r'_j$ was minimally far and another $\frac{h}{4\beta}$ for the $w,r'_j$ edge since we already removed all edges with length greater than $\frac{h}{4\beta}$, and the same thing for the subpath from $r$ to $v$. We let $\mathcal{P}=\left\{V\left(T'_1\right),V\left(T'_2\right),\dots,V\left(T'_{f'}\right)\right\}$ and we are done. Also see \Cref{fig:lowdiam-partition}.
\end{proof}
\begin{figure}[H]
    \centering
    \scalebox{0.7}{ % Scale everything by 80%
    \begin{tikzpicture}[
      level distance=1.5cm,
      sibling distance=2.5cm,
      every node/.style={circle, draw=black, thick, fill=pink, minimum size=8mm},
      edge from parent/.style={draw=black, thick}
    ]
    
    % Tree structure
    \node (root) {}
      child { node (a1) {}
        child { node (a2) {}
          child { node (a3) {} }
        }
      }
      child { node (b1) {}
        child { node (b2) {} }
        child { node (b3) {}
          child { node (b4) {} }
        }
      };
    
    % Labels above nodes
    \node[above=4pt of a1, draw=none, fill=none, rectangle, inner sep=0] {\Large$\leq \frac{h}{4\beta}$};
    \node[above=4pt of b1, draw=none, fill=none, rectangle, inner sep=0] {\Large$\leq \frac{h}{4\beta}$};

    % Left brace
    \draw[thick,decorate,decoration={brace,amplitude=8pt}]
      ($(a3.south west)+(-0.3,-0.2)$) -- ($(a1.north west)+(-0.3,0.1)$)
      node[midway,xshift=-12pt,left, draw=none, fill=none] {\Large$\leq \frac{h}{4\beta}$};
    
    % Right brace
    \draw[thick,decorate,decoration={brace,mirror,amplitude=8pt}]
      ($(a3.south west)+(4.7,-0.2)$) --
      ($(a1.north west)+(4.7,0.1)$)
      node[midway,xshift=12pt,right, draw=none, fill=none] {\Large$\leq \frac{h}{4\beta}$};
    
    \end{tikzpicture}
    }
    \caption{A subtree $T'_i$ rooted at $r'_i$. A $u,v$ path in $T'_i$ at worst uses $\frac{h}{4\beta}$ to go from $u$ to a child $w$ of $r'_i$, $\frac{h}{4\beta}$ to go from $w$ to $r'_i$, $\frac{h}{4\beta}$ to go from $r'_i$ to a child $x$ of $r'_i$, and $\frac{h}{4\beta}$ to go from $x$ to $v$.}
    \label{fig:lowdiam-partition}
\end{figure}
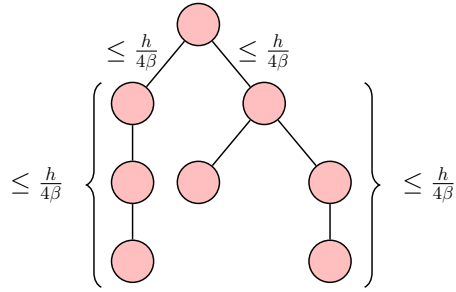
\noindent We will refer to the parts of $\mathcal{P}$ computed by \Cref{lemma:low-diam-partition} as \textit{pieces}. 

Recall from \Cref{defn:lc-region} that a $2h$-length-constrained region $\left(H,L_H\right)$ satisfies $l\left(L_H\right)\leq O(h)$, and $L_H$ has $k\leq O(1)$ connected components. Applying \Cref{lemma:low-diam-partition} to each connected component gives a partition of each component into pieces $\mathcal{P}_{i}=\left\{V_1,V_2,\dots\right\}$ such that $|\mathcal{P}_{i}|\leq O(\beta)$ and the induced subgraph of each piece in each corresponding component has diameter at most $h/\beta$. In particular, none of the pieces across any of the $\mathcal{P}_i$'s overlap on a vertex or edge, since by definition they belonged to different connected components to start out. 

Let $\mathcal{P}_H=\bigcup_{i\in[k]}\mathcal{P}_{i}$. Since $k\leq O(1)$ then $|\mathcal{P}_H|\leq O(1)\cdot O(\beta)=O(\beta)$. We say $\mathcal{P}_H$ is a $\beta$-partition of the boundary $L_H$ into pieces. Then given a $2h$-length-constrained region, we now have a means to break its boundary into connected components each with diameter at most $h/\beta$. The number of pieces being at most $O(\beta)$ will be used later. 

\subsection{From Pieces to a Whole (Length-Constrained Steiner Tree)}\label{sec:chandInstance}
In this subsection we show how to reduce length-constrained MST to many small instances of length-constrained Steiner tree, given a length-constrained division hierarchy. Recall that in length-constrained Steiner tree, we are given the same input along with a terminal subset $U\subseteq V$ and are required to return a tree of $G$ with minimum total edge weight among trees $T$ of $G$ satisfying $d_T(r,t)\leq h$ for each $t\in U$.
\paragraph{Construction.} Our approach with the hierarchy involves iterating across the hierarchy from bottom to top, considering an instance for every parent region and its children. For brevity we will always assume that a length-constrained region $(H,L)$ comes with a $\beta$-partition $\mathcal{P}$ of $L$ into pieces computed by \Cref{lemma:low-diam-partition}. 
% Fix some node $(H,L)$ with a child $(H',L')$ in a $2h$-length-constrained $\alpha$-division hierarchy. Let $\mathcal{P}_H,\mathcal{P}_{H'}$ be partitions of $L,L'$ into pieces by \Cref{lemma:low-diam-partition}, respectively. 
% Now suppose we contract the $O(\beta)$ connected components of $\mathcal{P}_H$ into $O(\beta)$ vertices. 
% by contracting all the edges of a segment besides the one closest to the root (see \Cref{fig:contractsegments}).
% Let $h_i$ (resp. $h_i$') be the minimum distance (under $l$) in our fixed $h$-length MST $T^*$ between $r$ and a vertex in piece $V_i\in\mathcal{P}$ (resp. $V'_i\in\mathcal{P}'$). Suppose we know all of the $h_i$'s and $h_i'$'s. Note that $r$ might not be in $H$, so let us insert a fake root $\textbf{r}$ in $H$. We can define a length-constrained Steiner tree instance on this transformed graph as follows: . Our goal is to find a subtree of minimum weight that connects the root to the $i^{\text{th}}$ piece with a path of length at most $h_i$ for all $i$. We can further make the length bounds for each piece uniformly $h$ by adding a gadget node for each piece, and connecting the $i^{\text{th}}$ gadget node to the $i^{\text{th}}$ piece by an edge of weight $0$ and length $h-h_i$.
\begin{definition}[Pieces to Length-Constrained Steiner Tree]\label{defn:pieces-to-whole}
    Given a node $(H,L)$ with child $(H',L')$ in a $2h$-length-constrained $\alpha$-division hierarchy $\mathcal{T}$ of graph $G$ where $\mathcal{P}=\left\{V_1,V_2,\dots\right\},\mathcal{P}'=\left\{V'_1,V'_2,\dots\right\}$ are $\beta$-partitions of $L,L'$ 
    into pieces by \Cref{lemma:low-diam-partition}, and a function $g$ (resp. $g'$) that maps each piece of $\mathcal{P}$ (resp. $\mathcal{P}'$) to a value between $0$ and $h$, we let the graph $F\left(H,H',g,g'\right)$ ($F$ for short when the underlying regions and functions are obvious) be defined as follows:
    \begin{enumerate}
        \item \textbf{Root}: If $(H,L)=(G,\emptyset)$, we let the root of $F$ be $\tilde{r}=r$ and we let $r$ be the only piece of $\mathcal{P}$. 
        Otherwise we insert a fake root $\tilde{r}$ in $V(F)$ even if $r\in V(H)$ already.
        \item \textbf{Vertices}: $V(F)$ includes $V(H)$ and a vertex $t'_i$ for every piece $V'_i\in \mathcal{P}'$ and let $U$ be the set of all $t'_i$'s. Although $U$ is defined with respect to a specific instance of $F$, for brevity we will always just refer to it as $U$ since it will always be written alongside its corresponding $F$. 
        \item \textbf{Edges}: $E(F)$ includes $E(H)$ with their lengths unchanged but weights of edges with both endpoints in a boundary set to $0$ (and all other weights unchanged), an edge of length $g\left(V_i\right)$ and weight $0$ between $r,v$ for every $v\in V_i$ for every $i\in|\mathcal{P}|$ (let $X$ denote the set of such edges), and an edge of length $h-g'\left(V'_i\right)$ and weight $0$ between $t'_i,v$ for every $v\in V'_i$ for every $i\in|\mathcal{P}'|$ (let $X'$ denote the set of such edges).
    \end{enumerate}
\end{definition}
\noindent See \Cref{fig:pieces-to-whole} of an illustration our construction of $F$. 
\begin{figure}[H]
    \centering
    \includegraphics[width=0.45\linewidth]{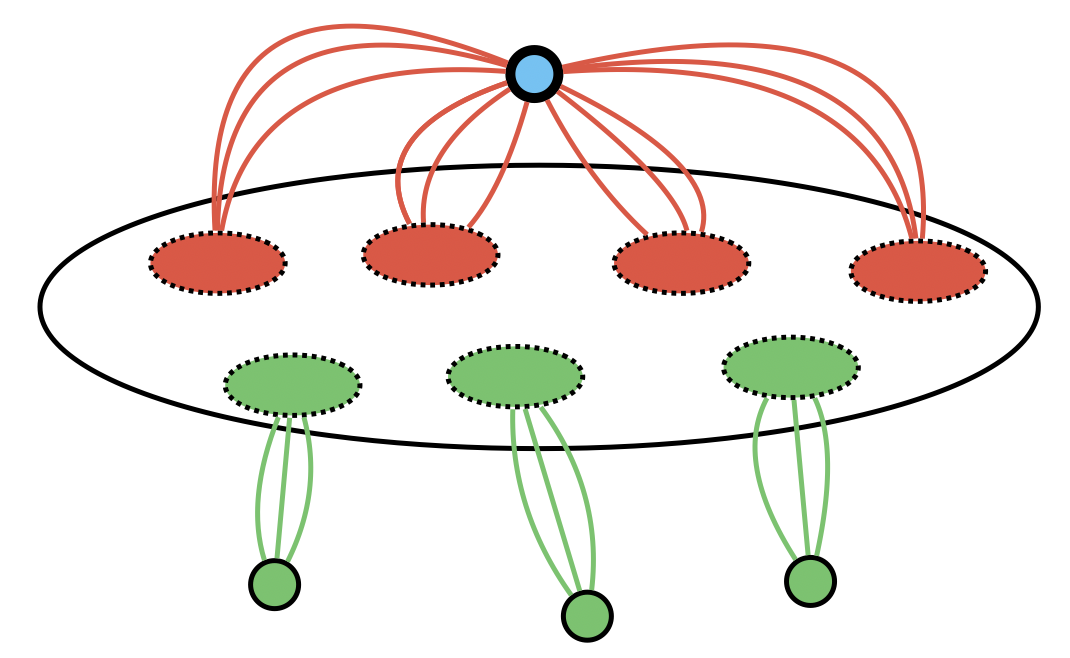}
    \caption{Graph $F\left(H,H',g,g'\right)$, where $H$ is the clear oval, $\tilde{r}$ is the bolded blue vertex, the pieces of $\mathcal{P}$ are the red dotted ovals, the pieces of $\mathcal{P}'$ are the green dotted ovals, $U$ contains the green vertices, $X$ contains the red edges, and $X'$ contains the green edges.}
    \label{fig:pieces-to-whole}
\end{figure}
\noindent There is some flexibility given by the functions $g,g'$. In a fantasy land, $g,g'$ could be the following:
\begin{definition}[Distance to Pieces]\label{defn:h_i}
    Given a length-constrained region $(H,L)$ and a partition $\mathcal{P}$ of $L$ into pieces, we let 
    $h^*\left(\hat{V}\right):=\min_{v\in \hat{V}}\left(d_{T^*}\left(r,v\right)\right)$
    be the minimum distance under $l$ between $r$ and a vertex in $\hat{V}\in\mathcal{P}$ in a fixed $h$-length MST $T^*$.  
\end{definition}
\noindent If we somehow knew $h^*$ ahead of time and let $g,g'=h^*$ to define the edge lengths of $F\left(H,H',g,g'\right)$, then we could define an instance of length-constrained Steiner tree with $t=O(\beta)$ terminals on $F\left(H,H',g,g'\right)$ with $U$ as the terminal set and $h$ as the length bound. 

Any path $P$ between $\tilde{r}$ and a terminal $t'\in U$ must include an edge $e=(\tilde{r},a)\in X$, an edge of $e=(b,t')\in X'$, and a subpath from $a$ to $b$ (whose length we denote as $z$) where $a\in \hat{V}$ and $b\in V'$ for parent piece $\hat{V}$ and child piece $V'$. Notably, the length $z$ subpath contains only edges in $E(H)$, and we did not change those edge lengths. So for any path $P$ between $\tilde{r}$ and a terminal $t'\in U$, a solution $T$ to this instance of length-constrained Steiner tree on $F$ must satisfy 
$$l_{F}(P)=h^*\left(\hat{V}\right)+\left(h-h^*(V')\right)+z\leq h\implies z\leq h^*(V')-h^*\left(\hat{V}\right)$$
% and since $h^*\left(\hat{V}\right)\in[0,h]$ for any piece $\hat{V}$ we have $z\in[0,h]$. 
If we only look at the subset of edges in $F$ shared with $E(H)\setminus E(L\cup L')$ then we obtain a path of length at most $h^*\left(V'\right)-h^*\left(\hat{V}\right)$ in $H$ between these pieces. By an inductive argument, a vertex of piece $V'$ can reach the root via length at most $h^*\left(V'\right)$ for any piece $V'$.
Furthermore, the union of $E\left(T^*_{|H}\right)$ and $X,X'$ is a feasible solution to the length-constrained Steiner tree instance on $F\left(H,H',h^*,h^*\right)$ (we will prove this for an even weaker assumption on $g,g'$ in \Cref{lemma:correct-guess}). It follows that solving the length-constrained Steiner tree instance gives us an edge set whose weight is competitive with $T^*$.
% In \textit{length-constrained Steiner tree}, in addition to the same input as in length-constrained MST, we are also given a terminal set $U\subseteq V,|U|=t$, and the objective now is to return a minimum weight subgraph of $G$ such that there exists a path of length at most $h$ between each terminal and the root\footnote{Algorithms for length-constrained MST are often generalizable to the Steiner tree version, e.g.\ in \cite{ravi1998bicriteria,hershkowitz2024simplelengthconstrainedminimumspanning} and the algorithm we present in this paper.}. 

Then given this instance of length-constrained Steiner tree with $O(\beta)$ terminals, we can use the single-criteria approximation algorithm of \cite{charikar1999approximation} to obtain such a set of edges in $H$ that have good length and weight. Specifically, they showed the following:
\begin{theorem}[Theorem 3.3 in \cite{charikar1999approximation}]\label{thm:chandra}
    For any $\delta>0$, there is an $O\left(t^\delta/\delta^3\right)$ approximation algorithm with $1$ length slack that runs in time $t^{O\left(1/\delta\right)}$ for length-constrained Steiner tree.
\end{theorem}
\noindent Importantly, although $F\left(H,H',g,g'\right)$ is not necessarily planar, \Cref{thm:chandra} works for general graphs. Using their algorithm as a subroutine, we obtain a subset of edges of $E(H)\setminus E(L\cup L'$ with weight $O\left(\beta^\delta/\delta^3\right)\cdot\OPT_{|H}$ from $F$ that exactly matches the length bounds of $T^*$ for some vertex in each piece in $\mathcal{P}'$ in time $\beta^{O\left(1/\delta\right)}$. 
% If $\beta=O\left(\frac{\log n}{\eps}\right)$ like in the warm-up algorithm then we can efficiently compute a single subtree with weight $O\left(\left(\frac{\log n}{\eps}\right)^\delta\right)$ which is already an improvement on the shortcut paths of weight $O\left(\frac{\log n}{\eps}\right)$ for $\delta<1$. 

But if we don't know $h^*$, then we have no way of constructing the correct instance to apply \Cref{thm:chandra} to. We will need a guessing procedure to construct a $g$ that is ``like'' $h^*$. 
\paragraph{Guesses.} 
Clearly we have $h^*\left(V_i\right)\in[0,h]$ for any piece $V_i\in\mathcal{P}$. We could try multiple choices of $g$ to get closer to $h^*$, but we cannot afford to try every value in $[0,h]$ for every piece, as this could lead to exponentially many choices of $g$ (and we still need to run the length-constrained Steiner tree algorithm for each choice of $g$). Instead, we bucket guesses coarsely enough such that there are not too many guesses yet fine-grained enough such that we don't lose too much in length. 

Observe we already give up $h/\beta$ because each piece has diameter at most $h/\beta$ and at worst a vertex has to travel the entire diameter of its piece before exiting the piece. Motivated by this observation, we will use the following bucketing of guesses for each piece:
$$ \mathcal{G}:=\left\{\frac{h}{\beta},\frac{2h}{\beta},\frac{3h}{\beta},\dots,h\right\}$$
% We let $\Vec{g}\in\mathcal{G}^{|\mathcal{P}|}$ be a $|\mathcal{P}|$-length vector where each entry is a value in $\mathcal{G}$, and let $g(V_i)=\Vec{g}_i$. We will try every possible $g$, i.e.\ every possible $\Vec{g}\in\mathcal{G}^{|\mathcal{P}|}$, and refer to $g$ as a \textit{guess}. 
More formally:
\begin{definition}[Guess]\label{defn:guess}
    Given a $\beta$-partition $\mathcal{P}$ of pieces, a guess $g$ is mapping from each piece of $\mathcal{P}$ to a value in $\mathcal{G}:=\left\{\frac{ih}{\beta}\right\}_{i\in[\lceil \beta \rceil]}$.
\end{definition}
\noindent
There are $\beta$ terms in $\mathcal{G}$, so if there are $O(\beta)$ pieces then we need a total of $\beta^{\beta}$ total guesses. This expression needs to be polynomial on $n$ so that the number of guesses is polynomial. 
We will also later prove that if we iterate up a length-constrained $\alpha$-division hierarchy from bottom to top, connecting child pieces to parent pieces in the manner described earlier as we travel upwards, then we essentially add $h/\beta$ to the maximum distance-to-$r$ of our solution $\log_\alpha n$ times. 

So there is a tension between our setting of $\alpha,\beta$: we want $\beta$ to be small enough such that $\beta^\beta=\poly(n)$, we want $\alpha$ to be small because we already pay $\alpha$ in weight with the hierarchy alone, and we want the maximum distance-to-$r$ to be at most $h+h\eps$ (this is equivalent to the final length slack to be $1+\eps$), which requires $\log_{\alpha}n=h\eps\beta$. With this in mind, we will (roughly) set
$$\alpha=\log^\eps n,\beta=\frac{\log n}{\eps^2\log\log n}$$
which achieves $\beta^\beta=n^{O\left(1/\eps^2\right)}$ and $\log_{\alpha}n=\frac{\log n}{\eps\log\log n}=\frac{h\eps}{\frac{h\eps^2\log\log n}{\log n}}=h\eps\beta$. 

With all of these guesses in hand, it is natural to single one out as the best guess.
\begin{definition}[Correct Guess]\label{defn:correct}
    Given a $\beta$-partition of pieces $\mathcal{P}=\left\{V_1,V_2,\dots\right\}$, we say a guess $g$ is \textbf{correct} if $g(V_i)\in\left[h^*(V_i),h^*(V_i)+\frac{h}{\beta}\right)$ for every $i\in[|\mathcal{P}|]$, or equivalently, if $h^*(V_i)\in\left(g(V_i)-\frac{h}{\beta},g(V_i)\right]$ for every $i\in[|\mathcal{P}|]$. 
\end{definition}
\noindent Note that by definition of $\mathcal{G}$ there is always a correct guess for any region and partition of pieces. Next, we introduce some notation for the length-constrained Steiner tree solutions that we compute using \Cref{thm:chandra}.
\begin{definition}\label{defn:lcst}
    We let $\LCST(G,U,h)$ be the solution (i.e.\ a subset of edges) of the length-constrained Steiner tree instance on graph $G$ with terminal set $U$ and length bound $h$ returned by the algorithm of \Cref{thm:chandra}.
    If there is no feasible solution of the instance then we let $\LCST\left(G,U,h\right) = \operatorname{fail}$ and $w\left(\LCST\left(G,U,h\right)\right)=\infty$. 
    Given a parent node $(H,L)$ with child $(H',L')$ in a hierarchy $\mathcal{T}$ and corresponding guesses $g,g'$ we let
    $$\LCST^*\left(F\left(H,H',g,g'\right),U,h\right):=\LCST\left(F\left(H,H',g,g'\right),U,h\right)\cap \left(E\left(H\right)\setminus \left(E\left(L\right)\cup E\left(L'\right)\right)\right)$$
    be the set of edges in the solution that are also edges in $H$ that are not fully contained in $L,L'$. 
\end{definition}
\noindent Since a correct guess can be overestimate by $h/\beta$, we will actually consider the length-constrained Steiner tree instances with length bound $h+h/\beta$, not $h$. The following lemma relates our guesswork with the optimal solution, formalizing some of the intuition provided earlier. 
\begin{lemma}[Correct Guess is Good Enough]\label{lemma:correct-guess}
    Given a parent node $\left(H,L\right)$ with child $\left(H',L'\right)$ in a hierarchy $\mathcal{T}$ and correct guesses $g,g'$, we have that $w\left(\LCST^*\left(F,U,h+h/\beta\right)\right)\leq O\left(\beta^\delta/\delta^3\right)\cdot\OPT_{|H}$.
\end{lemma}
\begin{proof}
    By definition of correct guesses we have $g(V)\geq h^*\left(V\right)$ for any piece of $\mathcal{P}$ (and the analogous statement for $g'$). Then we can construct a feasible solution of the length-constrained Steiner tree instance on $F,U,h$ using the $h$-length MST $T^*$ as follows: let $T$ be a subgraph whose vertex set is $V(T)=V(H)$ and edge set is $E(T)=E\left(T^*_{|H}\right)\cup X\cup X'$ (where $X,X'$ are the edge sets as defined in \Cref{defn:pieces-to-whole}). 
    
     Note that for every child piece $V'$, the edges $E\left(T^*_{|H}\right)$ contains a path to a $v\in V'$ with length at most $h$ as otherwise $T^*$ wouldn't be an $h$-length MST. Let $P$ be the shortest such path to a vertex in $V'$ across all vertices in $V'$; by our fixing of $T^*$ and definition of $h^*$ we have $h^*\left(V'\right)=l(P)$. Then $P$ may go through multiple different parent pieces, but there is eventually a subpath $P_2$ of $P$ that starts at some parent piece $\hat{V}$ and ends at $V'$ without touching any other parent pieces (and let $P_1$ be the entire subpath that precedes $P_2$). In particular, we have 
     \begin{align*}
         g\left(\hat{V}\right)&\leq h^*\left(\hat{V}\right)+h/\beta\leq l\left(P\right)_2+h/\beta,\\
         g'\left(V'\right)&\geq h^*\left(V'\right)= l\left(P_1\right)+l\left(P_2\right)
     \end{align*}
    by definition of correct guesses (\Cref{defn:correct}) and by our choice of $P$. Then the solution $T$ is feasible to the length-constrained Steiner tree instance with graph $F$, terminals $U$, and length bound $h+h/\beta$ because it contains a $\tilde{r},t'$ path that is the union an edge of length $g\left(\hat{V}\right)$ from $X$, the subpath $P_2$, and an edge of length $h-g'\left(V'\right)$ from $X'$. This path has length at most 
    \begin{align*}
        g\left(\hat{V}\right) + l\left(P_2\right) + h-g'\left(V'\right)
        &\leq l\left(P_1\right) + h/\beta + l\left(P_2\right) + h - l\left(P_1\right) - l\left(P_2\right) \\
        &= h+h/\beta
    \end{align*}
    for any terminal $t'\in U$.
    
    Let $\OPT_{(F,U,h+h/\beta)}$ be the weight of an optimal solution of the length-constrained Steiner tree instance on graph $F$ with terminal set $U$ and length bound $h+h/\beta$. Then we have
    \begin{align*}w\left(\LCST^*\left(F,U,h+h/\beta\right)\right) 
        &\leq O\left(\beta^\delta/\delta^3\right)\cdot\OPT_F\\
        &\leq O\left(\beta^\delta/\delta^3\right)\cdot w(T)\\
        &= O\left(\beta^\delta/\delta^3\right)\cdot\beta^\delta \left(w(E(T^*_{|H})+w(X)+w(X')\right)\\
        &= O\left(\beta^\delta/\delta^3\right)\cdot w\left(E(T^*_{|H})\right)\\
        &= O\left(\beta^\delta/\delta^3\right)\cdot\OPT_{|H}
    \end{align*} 
    where the first line is because \Cref{thm:chandra} returns an $O\left(t^\delta/\delta^3\right)$ approximate length-constrained Steiner tree, 
    the second line is because a subgraph of $T$ is a feasible solution to the length-constrained Steiner tree instance on $F,U,h+h/\beta$, 
    the third line is by how we defined $E(T)$, 
    the fourth line is because the only edges with positive weight in $E(T)$ are those in $E\left(T^*_{|H}\right)$, 
    and the fifth line is by \Cref{defn:restriction}. 
\end{proof}
\subsection{Solving Many Length-Constrained Steiner Trees with Dynamic Programming}\label{sec:DP}
We maintain a table $\DP$ containing an entry for every region $H$ and every possible guess $g$. 
\begin{definition}[$\DP$]\label{defn:dp-table}
    Given a length-constrained division hierarchy $\mathcal{T}$, a partition $\mathcal{P}_H$ of $L_H$ into pieces for every $\left(H,L_H\right)\in V(\mathcal{T})$, a region $H\in V(\mathcal{T})$, and a guess $g$ we let 
    \begin{align*}
    \DP[H,g] =
    \begin{cases}
        0 & H\text{ is a leaf in }\mathcal{T} \\
        \displaystyle
        \sum_{H'\in\operatorname{child}(H)} 
        \min_{g'\in \mathcal{G}^{|\mathcal{P}_{H'}|}}\left(
        \DP[H,g'] + w\left(\LCST^*\left(F,U,h+\frac{h}{\beta}\right)\right)
        \right) 
        & \text{o.w.}
    \end{cases}
    \end{align*}
    Each entry $\DP[H,g]$ stores the weight of a set $S$ of edges and implicitly stores $S$ as well, so we will abuse notation and let $E\left(\DP\left[H,g\right]\right)=S$.
\end{definition}
\noindent With the appropriate choices of $\alpha,\beta$ we have that the size of $\DP$ is $\poly(n)$.
% \noindent After plugging in $\alpha$ we have that the total number of nodes in the hierarchy is 
% $\left(\log^\eps n\right)^{\log_{\log^\eps n}n}=\left(\log^\eps n\right)^{\frac{\log n}{\eps\log\log n}}=O(n)$
% $\left(\log^\eps n\right)^{\log_{\log^\eps n}n}=\left(\log^\eps n\right)^{\frac{\log n}{\eps\log\log n}}=2^{\frac{\log n \cdot \eps\log\log n}{\eps\log\log n}}=\Theta(n)$
% by \Cref{lemma:hierarchy-tree}, and the number of guesses for each segment of a region is at most $n^{1-o(1)}$. Hence the size of the table is $\poly(n)$. 

We give some intuition on how the dynamic program works.
In the base case, the region $H$ is a leaf in $\mathcal{T}$ and thus contains no non-boundary vertices, so we don't do anything. 
Otherwise, we try the algorithm in \Cref{thm:chandra} on the length-constrained Steiner tree instance with graph $F$ and terminals $U$ as defined in \Cref{defn:pieces-to-whole} with length bound $h+h/\beta$. Otherwise, we take a minimum across all guesses for each child of $H$. Note that sometimes our guesses may be invalid, i.e.\ there is no tree in the instance satisfying the guessed lengths. In this case we consider this guess a failure and set the entry to $\infty$. The dynamic program connects child pieces to parent pieces for every child of a parent, from the bottom up via our construction of the length-constrained Steiner tree instances. At the top level of the hierarchy is the region $(G,\emptyset)$ whose only piece is defined to be $r$ in \Cref{defn:pieces-to-whole}, so in this case we essentially connect all the pieces that have been inductively connected to the only parent piece of $G$, the root.
% which serves as the terminal of the length-constrained Steiner tree instance. We can solve this instance even without; e.g.\ given our guesses for each segment we can compute the minimum weight $h$-length path to $v$ and add that to the solution of the subproblem. For any other cases, we consider the child regions $H_1,H_2,\dots, H_k$ of region $H$:  
\subsection{Algorithm Description and Analysis}\label{sec:realalgdesc}
We can now put everything together to get the final algorithm. We start by carefully setting our parameters $\xi=\eps/2,\alpha=\log^{\xi}n,\beta=\frac{\log n}{\xi^2\log\log n}$. We compute and add a $2h$-length-constrained $\alpha$-division hierarchy $\mathcal{T}$ of $G$ to our solution, compute a $\beta$-partition into pieces for each node in $\mathcal{T}$, solve the corresponding $\DP$, add the edge set stored by $\argmin_{g}\DP[G,g]$ to our solution, and then return a shortest path tree of the solution rooted at $r$. We provide the pseudocode in \Cref{alg:planarmain2}. 
\begin{algorithm}[H]
    \caption{Main Algorithm}    
    \begin{algorithmic}[1]
        \label{alg:planarmain2}
        \Statex \textbf{Input}: undirected planar graph $G=\left(V,E\right)$ with root $r\in V$, edge length and weight functions $l:E\to\mathbb{Z}_{\geq0},w:E\to \mathbb{Z}_{\geq0}$, length bound $h$, constant $\epsilon>0$.
        \Statex \textbf{Output}: a $(1+\eps)h$-length spanning tree of $G$ rooted at $r$ with weight at most $O\left(\log^{1+\eps}n\right)\cdot\OPT$.
        %$O\left(\frac{\log^{1+\eps+\delta} n}{(\eps^2 \log\log n)^\delta}\right)\cdot\operatorname{OPT}_h$ 
        \Statex
        % \State $T\gets \MST\left(G,r,w,h,\epsilon,\delta\right)$
        \State Set $\xi=\eps/2,\alpha=\log^{\xi}n,\beta=\frac{\log n}{\xi^2\log\log n}$.
        \State Compute a $2h$-length-constrained $\alpha$-division hierarchy $\mathcal{T}$ of $G$ using \Cref{lemma:hierarchy}.
        \State For each $(H,L_H)\in V(\mathcal{T})$, compute a $\beta$-partition $\mathcal{P}_H$ of $L_H$ into pieces using \Cref{lemma:low-diam-partition}.
        \State Solve the dynamic programming table $\DP$ with respect to $\mathcal{T},\{\mathcal{P}_H\}_{H\in V(\mathcal{T})}$ as defined in \Cref{defn:dp-table}, applying \Cref{thm:chandra} with $\delta=\xi$ at each cell. 
        \State Set $T\gets\left(V,\left(\bigcup_{H\in V(\mathcal{T})} E(L_H)\right)\cup E\left(\argmin_g\DP\left[G,g\right]\right)\right)$.
        %E\left(\DP\left[G,\argmin_{g}\DP[G,g]\right]\right)\right)$. 
        \State \textbf{Return } a shortest path tree $\textbf{T}$ of $T$ rooted at $r$.
    \end{algorithmic}
\end{algorithm}
%\begin{algorithm}[H]
%    \caption{$\MST\left(G,r,w,h,\epsilon\right)$}    
%    \begin{algorithmic}[1]
%        \label{alg:planar2}
%            \State Return
%    \end{algorithmic}
%\end{algorithm}

% \subsection{Algorithm Analysis}\label{sec:realalganal}

It remains to prove that \Cref{alg:planarmain2} is an efficient $O\left(\log^{1+\eps}n\right)$ approximation algorithm for length-constrained MST with length slack $1+\eps$.
\begin{lemma}[Runtime]\label{lemma:realalg-runtime}
    \Cref{alg:planarmain2} runs in time $\poly(n)\cdot n^{1/\eps^2}$.
\end{lemma}
\begin{proof}
    Step 2 of the algorithm can be done in time $\poly(n)$ by \Cref{lemma:hierarchy}, step 3 can be done in time $\poly(n)$ by \Cref{lemma:low-diam-partition} and the fact there are $\poly(n)$ vertices partitions to compute (one for each node in the hierarchy), step 4 can be done in time $\poly(n)\cdot n^{O\left(1/\eps^2\right)}$ since by our choices of $\xi,\alpha,\beta$ we have the number of cells in $\DP$ is 
    \begin{align*}
        |V(\mathcal{T})|\cdot \beta^{O(\beta)} &= O\left(\alpha\right)^{\log_\alpha n}\cdot \beta^{O(\beta)}\\
        &= O\left(\log^\xi n\right)^{\frac{\log n}{\xi\log\log n}}\cdot \left(\frac{\log n}{\xi^2\log\log n}\right)^{O\left(\frac{\log n}{\xi^2\log\log n}\right)}\\
        &= \poly(n) \cdot n^{O\left(1/\xi^2\right)} \\
        % &\leq \left(\log^\eps n\right)^{\frac{\log n}{\eps\log\log n}}\cdot \left(\frac{\log n}{\eps^2\log\log n}\right)^{\frac{\log n}{\eps^2\log\log n}}\\
        &= \poly(n)\cdot n^{O\left(1/\eps^2\right)}
    \end{align*}
    and our application of the algorithm of \cite{charikar1999approximation} runs in time $\beta^{O\left(1/\xi\right)}=\log^{O\left(1/\eps\right)} n \leq n^{O\left(1/\eps^2\right)}$ for each cell by \Cref{thm:chandra}, and computing a shortest path tree can be done in time $\poly(n)$. Therefore the running time of \Cref{alg:planarmain2} is $\poly(n)\cdot n^{O\left(1/\eps^2\right)}$.
\end{proof}
\noindent Note that since $\eps$ is already assumed to be a constant, the runtime can be simplified to $\poly(n)$, but we expose the dependence on $\eps$ from our guessing for extra clarity. 
\begin{lemma}[Everything]\label{lemma:realeverything}
    \Cref{alg:planarmain2} returns a spanning tree $\textbf{T}$ of $G$ such that $w(\textbf{T})\leq O\left(\log^{1+\eps}n\right)\OPT$ and for every $v\in V$ we have $d_\textbf{T}(r,v)\leq (1+\eps)h$.  
\end{lemma}
\begin{proof}    
    Feasibility of $\textbf{T}$ is immediate as $\bigcup_{H\in V(\mathcal{T})} E(L_H)$ alone spans all of $V$ by \Cref{defn:hierarchy}.

    Since $\DP$ solves instances in the hierarchy $\mathcal{T}$ from bottom to top, we will reverse our indexing of the levels by letting level $1$ be the level containing the leaves of $\mathcal{T}$ and level $\operatorname{depth}(\mathcal{T})$ be the level containing $(G,\emptyset)$. 
    
    We start with the weight approximation. Let $C'$ be the hidden constant in the approximation guarantee of \Cref{thm:chandra}, i.e.\ \Cref{thm:chandra} guarantees a $C't^\delta/\delta^3$-approximation for length-constrained Steiner tree, and let $C:=C'/\xi^3$. We will prove by induction that given a node $(H,L)\in V(\mathcal{T})$ in level $i$ and a correct guess $g$, it satisfies $\DP[H,g]\leq \alpha\beta^\xi Ci\OPT_{|H}$. 
    
    In the base case of $i=1$ this is trivial, since by \Cref{defn:dp-table} all such entries are set to $0$. 
    So we assume the statement holds for $i-1$ and we get 
    \begin{align*}
        \DP[H,g] &= 
        \sum_{H'\in \operatorname{child}\left(H\right)} \min_{g'}\left(\DP[H',g'] + w\left(\LCST^*\left(F\left(H,H',g,g'\right),U,h+h/\beta\right)\right)\right) \\
        &\leq \sum_{H'\in \operatorname{child}\left(H\right)} \alpha\beta^\xi C(i-1)\OPT_{|H'} + w\left(\LCST^*\left(F\left(H,H',g,g'\right),U,h+h/\beta\right)\right) \\
        &\leq \alpha\beta^\xi C(i-1)\OPT_{|H}+ \sum_{H'\in \operatorname{child}\left(H\right)}w\left(\LCST^*\left(F\left(H,H',g,g'\right),U,h+h/\beta\right)\right) \\
        &\leq \alpha\beta^\xi C(i-1)\OPT_{|H}+ \sum_{H'\in \operatorname{child}\left(H\right)}\beta^\xi C\OPT_{|H} \\
        &\leq \alpha\beta^\xi C(i-1)\OPT_{|H}+ \alpha\beta^\xi C\OPT_{|H} \\
        &= \alpha\beta^\xi Ci\OPT_{|H}
    \end{align*}
    where the first line is by \Cref{defn:dp-table},
    the second line is by induction,
    the third line is by \Cref{lemma:restriction-sums}, 
    the fourth line is by \Cref{lemma:correct-guess}, 
    and the fifth line is because there are $\alpha$ children by \Cref{lemma:hierarchy}.
    So by our choices of $\alpha,\beta,\xi$ we have
    \begin{align*}
    \alpha\beta^\xi Ci\OPT_{|H} &=
    \left(\frac{\log^{2\xi}}{\left(\xi^2\log\log n\right)^{\xi}}\right) Ci\OPT_{|H}\\
        &\leq O\left(\log^{2\xi}n\right)i\OPT_{|H} \\
        &= O\left(\log^{\eps}n\right) i\OPT_{|H}
    \end{align*}
    where the second line is because $C$ is a constant. 
    
    Combining with \Cref{lemma:hierarchy-tree}, we have that there is a guess $g$ for the top region $G$ such that 
    \begin{align*}
        \DP[G,g]
        \leq O\left(\log^{\eps}n\right)\left(\log_\alpha n\right)\OPT 
        = O\left(\log^\eps n\right)\left(\frac{\log n}{\xi\log\log n}\right)\OPT 
        \leq O\left(\log^{1+\eps} n\right)\OPT
    \end{align*}
    % \begin{align*}
    %     w(T) 
    %     &= w\left(\bigcup_{H\in V(\mathcal{T})} L_H\right) + \DP[G,g] \\
    %     &\leq O\left(\alpha\log_\alpha n\right)\OPT+ \DP[G,g] \\
    %     &\leq O\left(\frac{\log^{1+\xi} n}{\xi\log\log n}\right)\OPT+ \DP[G,g] \\
    %     &\leq O\left(\frac{\log^{1+\eps/2}}{\eps\log\log n}\right)\OPT+O\left(\log^{1+\eps} n\right)\OPT \\
    %     &= O\left(\log^{1+\eps}n\right)\OPT
    % \end{align*}
    Now we analyze the length slack. We will prove by induction that given a node $(H,L)\in V(\mathcal{T})$ in level $i$ any vertex $v\in V(H)$, there exists a piece $\hat{V}$ of $L$'s boundary partition such that $v$ is connected to a vertex of $\hat{V}$ with length at most $h-g\left(\hat{V}\right)+i\frac{2h}{\beta}$ in the subgraph $S$ defined as follows: 
    \begin{align*}
        V(S)&=V(H),  \\
        E(S)&=\left(E\left(H\right)\cap\left(\bigcup_{H'\in \operatorname{child}(H)}E\left(\LCST^*\left(F\left(H,H',g,g'\right),U,h+h/\beta\right)\right)\right)\right)\cup E(L)
    \end{align*}    
    In the base case of $i=1$, all vertices in $H$ are boundary and thus belong to some piece in $L$'s boundary partition, so every $v\in V\left(H\right)$ is connected to a vertex of a piece (e.g.\ to itself) with length $0$ which is at most $h-g\left(\hat{V}\right)+2h/\beta$ since $h\geq g\left(\hat{V}\right)$.  
    % Each piece induced on $H$ (i.e.\ $L$) has diameter $h/\beta$, so any vertex in piece $\hat{V}$ in $H$ can reach $\tilde{r}$ via a path of length at most $2h/\beta$, and the claim holds. No edges of $S$ in this case are contributed by $E\left(\LCST\left(F,U,h\right)\right)$ but that was completely fine.  
    %In the base case of $i=2$, we have a region $H$ (with piece partition $\mathcal{P}$ and guess $g$) and fix one of its children $H'$ (with piece partition $\mathcal{P}'$ and guess $g'$).  
    
    So we assume the statement holds for $i-1$. We fix a node $(H,L)$ at level $i$ with guess $g$ and a $v\in V(H)$, and let $S$ be defined as before. If $v\in L$ then this is the same situation as in the base case we are done. Otherwise, $v$ must be in some child region $H'$ (with guess $g'$) of $H$, which we can apply induction to. By our construction of the instance in \Cref{defn:pieces-to-whole}, $\LCST\left(F\left(H,H',g,g'\right),U,h+h\beta\right)$ must contain a path $P$ from $\tilde{r}$ to $v$ that is the union of an edge $(\tilde{r},a)\in X$ with length $g\left(\hat{V}\right)$ where $a$ is in the piece $\hat{V}$ in $L$'s boundary partition, and a subpath $P$ from $a$ to a vertex $b$ in a child piece $V'$, and an edge $(b,t')\in X'$ with length $h-g'\left(V'\right)$. So we have 
    $$ g\left(\hat{V}\right)+l(P)+h-g\left(V'\right)\leq h+h/\beta$$
    which implies
    $$ l(P) \leq g'\left(V'\right) - g\left(\hat{V}\right)+ h/\beta$$
    Meanwhile, in $S$ there is a path from $a$ to $v$ that is the union of $P$, the path from $b$ to some vertex $c\in V'$ of length at most $h/\beta$ (since $b,c$ are in the same piece), and the $c,v$ path $P'$ computed inductively by the child. Recall that $a\in \hat{V}$; then we have 
    \begin{align*}
        d_S(a,v) 
        &\leq l(P)+l(P')+d_S(b,c)\\
        &\leq g'\left(V'\right)-g\left(\hat{V}\right)+h/\beta + h- g'\left(V'\right) + (i-1)\frac{2h}{\beta} + h/\beta\\
        &= h-g\left(\hat{V}\right) + i\frac{2h}{\beta} 
    \end{align*}
    where the second line is by induction.
    Recall that for the top level $(G,\emptyset)$ we let $\tilde{r}=r$ and the only parent piece is $r$. 
    Then at the top level we have for any $v\in V$ that 
    \begin{align*}
        d_S(r,v) 
        &\leq h-g\left(\hat{V}\right) + \operatorname{depth}(\mathcal{T})\cdot\frac{2h}{\beta} \\
        &\leq h + \operatorname{depth}\left(\mathcal{T}\right)\cdot\frac{2h}{\beta} \\
        &\leq h + \left(\log_\alpha n\right)\frac{2h}{\beta} \\
        &= h+\frac{\log n}{\xi\log\log n}\cdot\frac{2h\xi^2\log\log n}{\log n} \\
        &= h+2h\xi \\
        &= h+ h\eps
    \end{align*}
    % \begin{align*}
    %     d_S(r,v) 
    %     &\leq g^1\left(\right)-g^{\operatorname{depth}\left(\mathcal{T}\right)}\left(\right)+\operatorname{depth}\left(\mathcal{T}\right)\cdot\frac{2h}{\beta}\\
    %     &\leq h + \operatorname{depth}\left(\mathcal{T}\right)\frac{2h}{\beta} \\
    %     &\leq h+ \left(\log_\alpha n\right)\frac{2h}{\beta} \\
    %     &= h+\frac{\log n}{\xi\log\log n}\cdot\frac{2h\xi^2\log\log n}{\log n} \\
    %     &= h+2h\xi \\
    %     &\leq h+ h\eps
    % \end{align*}
    where the second line is because $g\left(\hat{V}\right)\geq0$,
    the third line is by \Cref{lemma:hierarchy-tree}, 
    the fourth is by our setting of $\alpha,\beta$, and the sixth line is by our setting of $\xi$.
% Then we have
    
%     \begin{align*}
%         d_S\left(b_v,v\right) 
%          &\leq \frac{h}{\beta}+\sum_{j\in[2,i]} g^{j-1}\left(\right) - g^{j}\left(\right) + \frac{2h}{\beta} \\
%         &= \frac{h}{\beta}+g^1\left(\right) - g^i\left(\right) + (i-1)\frac{2h}{\beta}  \\
%         &\leq \frac{h}{\beta}+h^{* 1}\left(\right) + \frac{h}{\beta} - h^{* i}\left(\right) + (i-1)\frac{2h}{\beta} \\
%         &=h^{* 1}\left(\right) -h^{* i}\left(\right) + i\frac{2h}{\beta}
%     \end{align*}
%     where the first line is by induction, the second line is because the sum telescopes, and the third line is by definition of correct guesses (\Cref{defn:correct}).

    It follows that there is a guess $g$ such that 
    $$T=\left(V,\left( \bigcup_{H\in V(\mathcal{T})} E\left(L_H\right) \cup E\left(\DP\left[G,g\right]\right) \right)\right)$$
    satisfies both $w(T)\leq O\left(\log^{1+\eps}n\right)\OPT$ and $d_T\left(r,v\right)\leq h+h\eps$ for any $v\in V$. To see why the weight bound holds, note that by \Cref{lemma:hierarchy} the total edge weight of the hierarchy is 
    \begin{align*}
        w\left(\bigcup_{H\in V(\mathcal{T})}E\left(L_H\right)\right) 
        &\leq O\left(\alpha\log_\alpha n\right)\OPT \\
        &= O\left(\frac{\log^{1+\xi} n}{\xi\log\log n}\right)\OPT \\
        &= O\left(\frac{\log^{1+\eps/2}n}{\eps\log\log n}\right)\OPT 
    \end{align*}
    and since this term is negligible compared to $\DP[G,g]$ we may conclude that
    \begin{align*}
        w(T) 
        &\leq O\left(\frac{\log^{1+\eps/2}n}{\eps\log\log n}\right)\OPT  + \DP[G,g] \\
        &\leq O\left(\frac{\log^{1+\eps/2}n}{\eps\log\log n}\right)\OPT+ O\left(\log^{1+\eps} n\right)\OPT \\
        &= O\left(\log^{1+\eps} n\right)\OPT
    \end{align*}
    Finally, taking a shortest path tree of $T$ rooted at $r$ can only reduce the total edge weight and does not change any distance-to-$r$. Therefore $w\left(\textbf{T}\right)\leq O\left(\log^{1+\eps}n\right)\OPT$ and $d_{\textbf{T}}\left(r,v\right)\leq h+h\eps$ for any $v\in V$. 
\end{proof}

\noindent \Cref{thm:planarmain2} follows from \Cref{lemma:realalg-runtime,lemma:realeverything}.
%We remark that \Cref{alg:planarmain2} can be slightly adjusted (namely, incorporating non-uniform vertex weights which our length-constrained planar separators/divisions/hierarchies were designed to handle) to give the same guarantees for length-constrained Steiner tree. 

\section{Integrality Gap of the Flow-Based LP Relaxation}\label{sec:intgap}
In this section we give a much simpler algorithm for planar length-constrained MST with the same length slack and a slightly worse weight-approximation guarantee, but is LP-competitive. The LP relaxation we are concerned with is below:
\begin{equation}\label{equ:lcmst-flow}
\begin{split}
    \min\sum_{e\in E}w(e)x_e &\cr\text{ such that}& \\
    \sum_{P\in \mcP^{(h)}(r,v)}f^v_P&= 1 \text{ }\forall\text{ }v\in V- r \\
    \sum_{P\in \mcP^{(h)}(r,v),P\ni e} f^v_P&\leq x_e\text{ }\forall\text{ }e\in E,v\in V-r\\
    x_e&\geq 0\text{ }\forall\text{ }e\in E\\
    f^v_P &\geq 0\text{ }\forall\text{ }v\in V-r, P\in \mcP^{(h)}(r,v)
\end{split}
\end{equation}
\noindent This is a standard flow-based relaxation. We define a flow variable for every $h$-length path from $r$ and guarantee that each non-root vertex gets sent exactly $1$ flow while using the $x_e$'s as the capacities. Our main result in this section is the following theorem:
\begin{theorem}\label{thm:intgap}
    Let $x$ be an optimal fractional solution to the flow-based LP of length-constrained MST with planar graph $G=(V,E)$, root $r$, and length bound $h$; there exists a polynomial-time algorithm that returns a spanning tree $T$ of weight $O\left(\frac{\log^3 n}{\eps}\right)w(x)$ and satisfying $d_T(r,v)\leq(1+\eps)h$ for every $v\in V$.
\end{theorem} 
\noindent 
On the (planar) directed Steiner tree side of things, \cite{chekuri2024directedpolymatroid} proved a similar theorem on the integrality gap of the corresponding cut-based LP relaxation by extending the algorithm of \cite{friggstad2023planardst}. 

Despite the LP having exponentially many variables, an extreme point of \Cref{equ:lcmst-flow} can be found in polynomial time via standard techniques, i.e.\ the ellipsoid method using min-cut computations as the separation oracle.

\subsection{Simpler Hierarchies and Piece-Connecting}
 Let $\textbf{w}(x)=\sum_{e\in E}w(e) x_e$ and $x_{|H}=$ be the restriction of $x$ on region $H$, i.e.\ $x_{|H_e}=x_e$ if $e\in E(H)$ and $0$ otherwise.
 
 We use a much simpler notion of length-constrained regions and divisions.
\begin{definition}[Simple Length-Constrained Region]\label{defn:simple-LC-region}
    A \textbf{simple $h$-length region} of a graph $G=(V,E)$ with edge lengths $l$ is a pair $H,L_H$ where $H$ is a region of $G$ and $L_H$ is a subgraph of $H$ with $V(L_H)=\bound_G(H)$ satisfying $l(L_H)\leq O(h)$. 
\end{definition}
\noindent Then we can define simple length-constrained divisions and simple length-constrained division hierarchies analogously with \Cref{defn:LC-division,defn:hierarchy}, but with simple $h$-length regions. 

\begin{lemma}[Simple Length-Constrained Division Hierarchy Existence and Algorithm]\label{lemma:simple-hierarchy}
    Given a planar graph $G=(V,E)$ with edge lengths $l$, edge weights $w$, $h\geq 1$, and extreme point $x$ to \Cref{equ:lcmst-flow}, one can compute in polynomial time a simple $2h$-length-constrained $3/2$-division hierarchy $\mathcal{T}$ of $G$ satisfying
    $$ w\left(\bigcup_{H\in V(\mathcal{T})} L_H\right)\leq O(\log n)\cdot \textbf{w}(x)$$
\end{lemma}
\begin{proof}
    We first show how to compute a simple $2h$-length constrained $3/2$-division. Given a region $H$, we compute a length-constrained separator $P$ and cycle $C=P\cup e_C$ of $H$ that separates $H$ into $A,B,V(C)$ using \Cref{lemma:lengthseparator}. Let $H'$ be the result of contracting $P$ into a single node, and let the division of $H$ be $\mathcal{H}=\{H'[A], H'[B]\}$. In other words, a $3/2$-division is the result of taking a single level of separation on a region. Also note that we contract the boundary here instead of flattening it. To compute $\mathcal{T}$, we recursively compute simple $2h$-length-constrained $3/2$-divisions until we have an instance with no non-boundary vertices. The recursive tree has depth $O(\log n)$ and each non-leaf node has $2$ children, so there are $\poly(n)$ nodes in the tree. Hence this can all be computed in polynomial time.
    
    To bound the weight, we first claim that $\textbf{w}(x)\geq D^{(h)}(G)$.
    Recall that by definition of $h$-radii, there exists a vertex $v\in V-r$ such that every path $P\in\mcP^{(h)}(r,v)$ satisfies $w(P)\geq D^{(h)}(G)$; fix such a $v$. 
    Then we have
    \begin{align*}
        D^{(h)}(G) &\leq \sum_{P\in\mcP^{(h)}(r,v)}w(P)f^v_P 
        \\&= \sum_{P\in\mcP^{(h)}(r,v)}f^v_P\sum_{e\in P}w(e) 
        \\ &= \sum_{e\in E}w(e)\sum_{P\in\mcP^{(h)}(r,v),P\ni e}f^v_P
        \\ &\leq \sum_{e\in E}w(e)x_e 
        \\ &= \textbf{w}(x)
    \end{align*}
    where the first line is by definition of $v$ and the first constraint of \Cref{equ:lcmst-flow}, and the fourth line is by the second constraint of \Cref{equ:lcmst-flow}.

    Observe that the $h$-radius never increases in later recursive subinstances since we contract the chosen separators into $r$ before recursing. Combining this with the above claim, we have that the weight of each separator path that we buy is $O\left(\textbf{w}(x)\right)$. 
    
    Let $l_i$ be the set of nodes at level $i$ in $\mathcal{T}$ where $l_1$ contains the root. Then the total weight of the hierarchy is 
    \begin{align*}
        w\left(\bigcup_{H\in V(\mathcal{T})}L_H\right) &= \sum_{i\in[\operatorname{depth}(\mathcal{T})-1]} \sum_{H\in l_i} w\left(\left(\bigcup_{\hat{H}\in \operatorname{child}(H)} L_{\hat{H}}\right)\setminus L_H\right) \\
        &\leq \sum_{i\in[\operatorname{depth}(\mathcal{T})-1]} \sum_{H\in l_i} O\left(\textbf{w}(x_{|H})\right) \\
        &\leq \sum_{i\in[\operatorname{depth}(\mathcal{T})-1]} O(\textbf{w}(x)) \\
        &\leq O(\log n)\cdot \textbf{w}(x)
    \end{align*}
    concluding the proof.
\end{proof}

\noindent Our algorithm is similar to \Cref{alg:planarmain2} until we have to connect pieces together. Instead of reducing to many instances of length-constrained Steiner tree with few terminals, we do the following. For a piece $V_i$ we let $P_i$ be the minimum length path in $G$ of weight at most some bound among such paths between any vertex in the parent piece and any vertex in $P_i$, and add it to our solution. We formalize the intuition in the following lemma:
\begin{lemma}[Shortcuts]\label{lemma:shortcuts}
    Suppose we are given $h\geq 1,\epsilon>0,c\leq 1$, a non-leaf node $(H,L)\in V(\mathcal{T})$, the separator $L'$ computed at node $H$ where $\mathcal{T}$ is a simple $h$-length $3/2$-division hierarchy $\mathcal{T}$ of a graph $G=(V,E)$ computed by \Cref{lemma:simple-hierarchy}, a $\beta$-partition $\{V_1,V_2,\dots\}$ of $L'$ into pieces computed by \Cref{lemma:low-diam-partition}, and an extreme point $x$ to \Cref{equ:lcmst-flow}. Let $H'$ be the union of $L$ and a dummy root $r'$ with an edge $(r',u)$ for every $u\in V(L)$ of weight $0$ and some nonnegative length. Then one can construct in polynomial time a set of edges $S_H$ satisfying
    \begin{enumerate}
        \item \textbf{\textit{Shortcuts}}: for any path $P$ from $u\in V(L)$ to a vertex in $V_i$ with all intermediate vertices in $H$ and $w(P)\leq c\cdot\textbf{w}(x_{|H})$, we have $d_{H'\cup S_H}\left(r',v\right) \leq l_{H'}(r',u)+l(P)+h/\beta$ for every $v\in V_i$.
        \item \textbf{\textit{Light}}: $w\left(S_H\right)\leq O(\beta c)\cdot \textbf{w}(x_{|H})$.
    \end{enumerate}
\end{lemma}
\begin{proof}
    By \Cref{lemma:low-diam-partition}, we have that $D(H[V_i])\leq h/\beta$ for every $V_i\in\mathcal{P}_H$. In particular, since $L'$ is a single separator path by \Cref{lemma:simple-hierarchy}, we have that each piece is a subpath of $L$ of length at most $h/\beta$. 
    Let $P_{i}$ be the minimum length path between $r'$ and a vertex in $V_i$ in $H'$ in $H$ of weight at most $c\cdot \textbf{w}(x_{|H})$. Such a path exists since there must be some path with weight at most $D^{(h)}(H)\leq \textbf{w}(x_{|H})$. We can find these paths $P_{i}$ efficiently. Indeed, for each piece, we can perform a shortest path computation between $r$ and each vertex of the segment, where we use the edge lengths as the distance metric and only consider the path to be feasible if its weight is at most $\textbf{w}(x_{|H})$. Then we may set $ S_H = \bigcup_i P_{i}$.

    Fix an $i$ and a vertex $v\in V_i$. One (possibly not shortest) $r',v$ path in $H'\cup S_H$ walks $P_{i}$ and then along $H[V_i]$ to $v$. So the distance $r'$ must travel to reach $v$ in $H'\cup S_H$ is at most
    $ l_{H'}(r',v)+l(P)+h/\beta$, proving the shortcut property. 
    Also observe there are at most $O(\beta)$ paths $P_i$ by \Cref{lemma:low-diam-partition}, and each path has weight at most $c\cdot \textbf{w}(x_{|H})$ by definition. Then $w\left(E\left(S_H\right)\right)\leq O(\beta c)\cdot\textbf{w}(x_{|H})$, proving the light property. 
\end{proof}
\noindent
Observe that since $L'$ is a single path in this case, we can alternatively use a much simpler procedure to partition $L'$ into low diameter pieces (which are all subpaths of $L'$), e.g.\ a greedy algorithm. 
% \begin{figure}[H]
%     \centering
%     \includegraphics[width=0.35\linewidth]{pics/shortcuts.png}
%     \caption{$L\cup S_H$. The root $r$ is the bolded node, the solid black edges form $L$, and the dotted lines represent the shortcut paths $\bigcup_i P_i$.}
%     \label{fig:shortcuts}
% \end{figure}

\subsection{Algorithm Description and Analysis}
We are now ready to describe the algorithm. Our algorithm computes an extreme point solution $x$ to \Cref{equ:lcmst-flow} to obtain $\textbf{w}(x)$, computes a simple $2h$-length $3/2$-division hierarchy, partitions the boundary into pieces, and connects each child piece to a parent piece via the shortest (in length) path whose weight is at most $\textbf{w}(x)$. 
The pseudocode is given in \Cref{alg:planar-intgap}. 

\begin{algorithm}[h]
    \caption{LP-Competitive Algorithm}  
    \begin{algorithmic}[1]
        \label{alg:planar-intgap}
        \Statex \textbf{Input}: undirected planar graph $G=\left(V,E\right)$ with root $r\in V$, extreme point solution $x$ to the LP \Cref{equ:lcmst-flow}, edge length and weight functions $l:E\to\mathbb{Z}_{\geq0},w:E\to \mathbb{Z}_{\geq0}$, length bound $h$, constant $\epsilon>0$.
        \Statex \textbf{Output}: a $(1+\eps)h$-length spanning tree of $G$ rooted at $r$ with weight at most $O\left(\frac{\log^{3}n}{\eps}\right)\cdot\textbf{w}(x)$.
        \Statex
        \State Set $\alpha=3/2,\beta=\frac{\log n}{\eps}$.
        \State Compute a simple $2h$-length-constrained $\alpha$-division hierarchy $\mathcal{T}$ of $G$ using \Cref{lemma:simple-hierarchy}.
        \State For each $(H,L_H)\in V(\mathcal{T})$, compute a $\beta$-partition $\mathcal{P}_H$ of $L_H$ into pieces using \Cref{lemma:low-diam-partition}.
        \State Compute a set of paths $S_H$ for every $v\in V(\mathcal{T})$ top-down using \Cref{lemma:shortcuts} with $c=2\log n$ and the lengths of edges $r',v$ in $H'$ in to be their length in the solution so far.
        \State Set $T\gets \left( V,\left( \bigcup_{H\in V(\mathcal{T})} E\left(L_H\right) \right) \cup \bigcup_{H\in V(\mathcal{T})} E\left(S_H \right) \right)$.
        \State \textbf{Return} a shortest path tree $\textbf{T}$ of $T$ rooted at $r$.
    \end{algorithmic}
\end{algorithm}

% \noindent We can analyze our algorithm now.
\begin{lemma}[Runtime]\label{lemma:runtime-intgap}
    The runtime of \Cref{alg:planar-intgap} is $\poly\left(n\right)$. 
\end{lemma}
\begin{proof}
    Clearly, computing each separator and set of shortcut paths in each instance requires at most $\poly(n)$ time. The recursive depth is $O(\log n)$ and each instance makes at most $2$ recursive calls. Therefore, the runtime of \Cref{alg:planar-intgap} is $\poly(n)$. 
\end{proof}
\begin{lemma}[Everything]\label{lemma:everything-intgap}
    \Cref{alg:planar-intgap} returns a spanning tree $\textbf{T}$ of $G$ such that $w(\textbf{T})\leq O\left(\frac{\log^3 n}{\eps}\right)\textbf{w}(x)$ and for every $v\in V$ we have $d_{\textbf{T}}(r,v)\leq (1+\eps)h$.
\end{lemma}
\begin{proof}
    Feasibility of $\textbf{T}$ is immediate by the same reason as in \Cref{lemma:realeverything}.

    We bound the weight of $T$ below: 
    \begin{align*}
        w(T) &\leq w\left(\bigcup_{H\in V(\mathcal{T})}E(L_H)\right) + w\left(\bigcup_{H\in V(\mathcal{T})}E(S_H)\right) \\
        &\leq O(\log n)\textbf{w}(x)+ w\left(\bigcup_{H\in V(\mathcal{T})}E(S_H)\right) \\
        &\leq O(\log n)\textbf{w}(x)+ \sum_{i\in [\operatorname{depth}(\mathcal{T})]} \sum_{H\in l_i} w\left(S_H\right) \\
        &\leq O(\log n)\textbf{w}(x)+ \sum_{i\in [\operatorname{depth}(\mathcal{T})]} \sum_{H\in l_i} O\left(\frac{\log^2 n}{\eps}\right)\textbf{w}(x_{|H}) \\
        &= O(\log n)\textbf{w}(x)+ O\left(\frac{\log^2 n}{\eps}\right)\sum_{i\in [\operatorname{depth}(\mathcal{T})]} \sum_{H\in l_i} \textbf{w}(x_{|H}) \\
        &\leq O(\log n)\textbf{w}(x)+ O\left(\frac{\log^2 n}{\eps}\right)\sum_{i\in [\operatorname{depth}(\mathcal{T})]}  \textbf{w}(x) \\
        &\leq O(\log n)\textbf{w}(x) + O\left(\frac{\log^3 n}{\eps}\right)\cdot \textbf{w}(x) 
        %\\&\leq O\left(\frac{\log^2 n}{\eps}\right)\textbf{w}(x)
    \end{align*}
    where the second line is by \Cref{lemma:simple-hierarchy}, the fourth line is by \Cref{lemma:shortcuts}, and the seventh line is because $\operatorname{depth}(\mathcal{T})=O(\log n)$. Taking a shortest path tree of $T$ can only lower the weight, so the same weight bound holds for $\textbf{T}$.

    Next, the length bound. 
    Fix an $(H,L)\in V(\mathcal{T})$ at depth $i$ in $\mathcal{T}$ and a $v\in V(H)$. 
    First, observe that by Markov's inequality, there exists a path $P\in\mathcal{P}^{(h)}(r,v)$ such that $w(P\cap E(H_i))\leq O(\log n)\cdot\textbf{w}(x_{|H_j})$ for every region $H_j$ in the hierarchy containing $v$ (where we let $H=H_i$). 
    We prove by induction on $i$ that $d_{T}(r,v)\leq l(P)+\frac{h\eps i}{\log n}$. The base case of $i=1$ is trivial because we haven't contracted any boundary to compute the shortcut paths. Let $u$ be the last vertex in $P$ that belongs to $H$'s parent boundary. Then we have 
    \begin{align*}
        d_{T}(r,v) &\leq d_{T}(r,u) + d_{H^*\cup S_H}(r,v) \\
        &\leq l(P[r,u])+\frac{h\eps(i-1)}{\log n} + d_{H^*\cup S_H}(r,v) \\
        &\leq l(P[r,u])+\frac{h\eps (i-1)}{\log n} + \min_{\substack{u,v\text{ paths } P',\\w(P')\leq \log n\cdot\textbf{w}(x_{|H})}}l(P')  + \frac{h\eps}{\log n} \\
        &= l(P[r,u])+\frac{h\eps i}{\log n} + \min_{\substack{u,v\text{ paths } P',\\w(P')\leq \log n\cdot\textbf{w}(x_{|H})}}l(P') \\
        &\leq l(P)+\frac{h\eps i}{\log n}
    \end{align*}
    where the second line is by induction and the fact that $u$ is a vertex belonging to a region containing $v$ at depth $i-1$ in $\mathcal{T}$ (and hence $w(P\cap E(H_{i-1})\leq O(\log n)\cdot \textbf{w}(x_{|H_{i-1}})$, the third line is by \Cref{lemma:shortcuts} and definition of $u$, and the fifth line is because $u$ must exist in $P$. Since $d_{T^*}(r,v)\leq h$ and the $\operatorname{depth}(\mathcal{T})=O(\log n)$, we have that the longest distance-to-root in the output of \Cref{alg:planar-intgap} is at most $h+\frac{h\eps\log n}{\log n}=h+h\eps=(1+\eps)h$. Finally, since taking a shortest path tree rooted at $r$ does not change any of the distances from $r$, the same length bound holds for $\textbf{T}$. 
\end{proof}

\noindent \Cref{thm:intgap} follows from \Cref{lemma:runtime-intgap,lemma:everything-intgap}. 
% Observe that our proof of \Cref{lemma:everything-intgap} actually shows that for any $v\in V$ we have $d_{T^*}(r,v)\leq (1+\eps)d_{T}(r,v)$ for \textit{every} optimal solution $T$, which is much stronger than a general length bound of $(1+\eps)h$. 

% We also remark that the flow-based LP for directed Steiner tree was shown to have an integrality gap of $\Omega(n^{0.0418})$ in general graphs \cite{li2024polynomial}. It would be interesting if a bicriteria lower bound was true for the integrality gap of \Cref{equ:lcmst-flow} on general graphs.

We also remark that a slight modification to \Cref{alg:planar-intgap} gives an $O\left(\frac{\log^2 n}{\epsilon}\right)$ approximation algorithm (that is crucially not approximate with respect to the linear program) with the following stronger length guarantee: the length of the $r,v$ path in the algorithm's solution is at most $1+\epsilon$ times the length of the $r,v$ path in a canonical optimal solution. 

\section{Hardness of Approximation from Group Steiner Tree}\label{sec:gst-reduction}
In this section we prove a new bicriteria lower bound on the hardness of length-constrained MST for general graphs via a reduction from group Steiner tree, improving upon the decades-old single-criteria lower bound implied by a reduction from set cover \cite{naor1997retractedshallowlight}. 

In group Steiner tree, we are given a graph with edge weights, a root, and $k$ disjoint groups $g_1, g_2, \dots, g_k$ where each $g_i\subseteq V$\footnote{Disjointness of the groups is a standard assumption. In particular, for any vertex $v$ belonging to multiple groups, we can turn $v$ into a star by connecting a dummy vertex to $v$ with weight $0$ edges for each group containing $v$, assigning each dummy vertex to a unique group that $v$ belongs to, and removing $v$ from all groups.}. The goal is to find a minimum-cost tree that connects the root to at least one node from each group. Group Steiner tree was proven to be hard to approximate within $\Omega\left(\log^2 n\right)$ even on trees by \cite{halperin2003polyloginapx}:
\begin{theorem}[\cite{halperin2003polyloginapx}]\label{thm:gstlowerbound}
    For every fixed $\eps>0$, group Steiner tree cannot be approximated $O\left(\log^{2-\eps}n\right)$ in poly-time unless $\textsf{NP}\subseteq \textsf{ZTIME}\left(n^{\poly(\log n)}\right)$, even for trees.
\end{theorem}
% but also admits a polynomial time $O\left(\log^3 k\right)$ approximation \cite{fakcharoenphol2003tight}. 
\noindent A corollary to this theorem and our reduction is the following theorem:
\thmlowerbound*
\begin{proof}
\textbf{The reduction.} Given an instance of group Steiner tree involving a graph $F=(V,E)$ with root $r$ and disjoint groups $g_1,g_2,\dots,g_k$, we construct an instance of length-constrained MST tree as follows: Let $G$ be an undirected graph such that 
\begin{enumerate}
    \item \textbf{Vertices}: $V(G)$ is the union of $V(F)$ and a set copies of $V'=\left\{v'_1,v'_2,\dots\right\}$ of each vertex $v\in \bigcup_{i\in[k]}g_i$.
    \item \textbf{Edges}: $E(G)$ is the union of
    \begin{enumerate}
        \item the set $P$ of edges containing an edge of length $h$ and weight $0$ between each copy vertex in $V'$ and its original counterpart in $V(F)$ (hence we ``pull'' each vertex in a group away by length $h$),
        \item the set $C$ of edges containing an edge of length $0$ and weight $0$ between every pair of copy vertices whose original counterparts belong to the same group (hence we make a disjoint clique for each pulled group), 
        \item the set $S$ of edges containing an edge of length $h$ and weight $0$ between the root $r$ and every vertex in $V(F)$ (these are the ``spanning'' edges), 
        \item and $E(F)$, where for each $e\in E(F)$ we set $l_G(e)=0,w_G(e)=w_F(e)$.
    \end{enumerate} 
\end{enumerate}

% For every edge $e$ with both endpoints in $F'$, we set $w_G(e)=w_F(e)$ and $l_G(e)=1$, and for every edge $e$ with a group vertex as an endpoint we set $w_G(e)=0$ and $l_G(e)=n$. Finally, we let the root $r$ be the same as in the group Steiner tree instance, set the terminals to be $g_1,g_2,\dots,g_k$, and set $h=2n$. 

\noindent Given a solution $T_G$ to the instance of length-constrained MST with $G,h$, we transform it into a solution $T_F$ to the instance of group Steiner tree with $F,\{g_1,g_2,\dots,g_k\}$ by taking the edges $E\left(T_G\right)\cap E(F)$. 
See \Cref{fig:gst-reduction} for an illustration of the reduction, which also shows that the reduction does not preserve planarity.

\textbf{Correctness of the reduction.} By construction, we have that if $T_F$ is feasible/optimal then so is $T_G$, so it suffices prove the opposite direction. For feasibility, fix any $g_i$. $T_G$ is feasible, $g_i$ must have some copy vertex in $V'$ that is connected to $r$ in $T_G$ via an $h$-length path. Such a path uses exactly one edge in $P$ which has length $h$ by definition of $P$. Then the rest of the path cannot use any edges of $S$ since those edges all have length $h$, so the rest of the path must be a path from some vertex in $g_i$ to $r$ using only edges in $E(F)$. This is what we need. Now suppose that $T_G$ is optimal but $T_F$ isn't, i.e.\ some tree $T_F''$ of $F$ has strictly less weight than $T_F$ and is feasible for the group Steiner tree instance on $F,r,\{g_1,g_2,\dots,g_k\}$. The only edges with positive weight in $G$ are those in $E(F)$. Then we can create a subgraph of $G$ using edges $E\left(T_F''\right)\cup \left(E\left(T_G\right)\cap P\right)\cup C\cup S$ with strictly less weight than $T_G$, contradicting $T_G$'s optimality.

\textbf{Lower bound.} 
Notice that if we ignore the length bound $h$, the only ways to get a solution with better weight in $G$ is by either purchasing a weight $0$ path between different pulled cliques using edges of $P$ (which necessarily has length at least $3h$), or purchasing a weight $0$ path from the root to any pulled clique via an edge in $S$ (which necessarily has length $2h$). Combining this with \Cref{thm:gstlowerbound}, we get the lower bound of $2$ on the length slack and $\Omega\left(\log^{2-\eps}n\right)$ on the weight approximation in \Cref{thm:newlowerbound}.
% Notice that if we ignore the length bound $h$, the only way to get a solution with better weight in $G$ is to purchase a weight $0$ path between the terminals e.g.\ from $g_1$ to $g_2$ in the case of there being two terminals. However, this necessarily increases the length of the solution by at least $2n$, so the length of the solution must be at least $4n=2h$. 
\end{proof}
\begin{figure}[h]
    \begin{subfigure}[b]{.25\textwidth}
      \centering
        \includegraphics[width=1.05\linewidth]{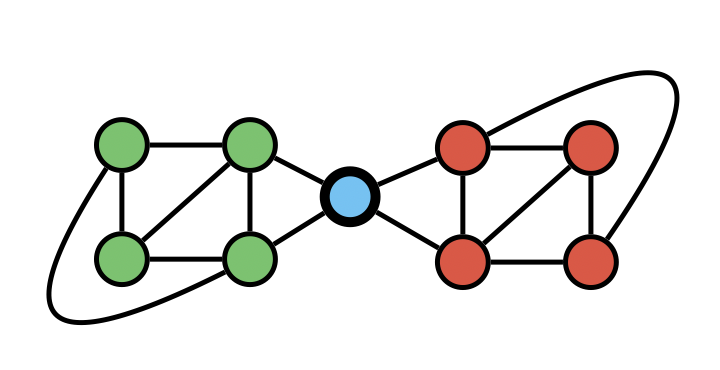}
      \caption{$F$.}\label{subfig:gst-instance}
    \end{subfigure}%
    ~
    \begin{subfigure}[b]{.25\textwidth}
      \centering
      \includegraphics[width=1.05\linewidth]{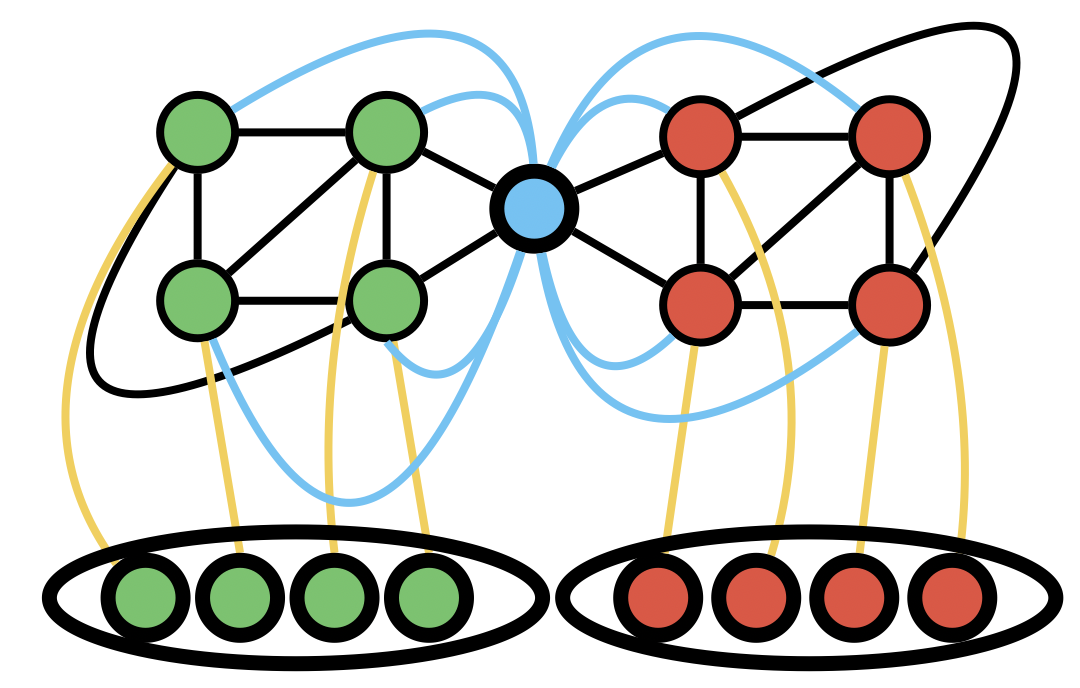}
      \caption{$G$.}\label{subfig:gstlcmst-instance}
    \end{subfigure}%
    ~
    \begin{subfigure}[b]{.25\textwidth}
      \centering
        \includegraphics[width=1.05\linewidth]{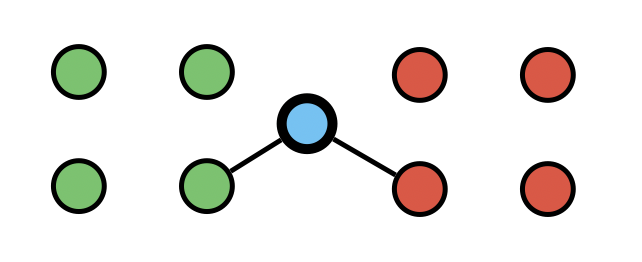}
      \caption{Solution of $F$.}\label{subfig:gst-solution}
    \end{subfigure}%
    ~
    \begin{subfigure}[b]{.25\textwidth}
      \centering
      \includegraphics[width=1.05\linewidth]{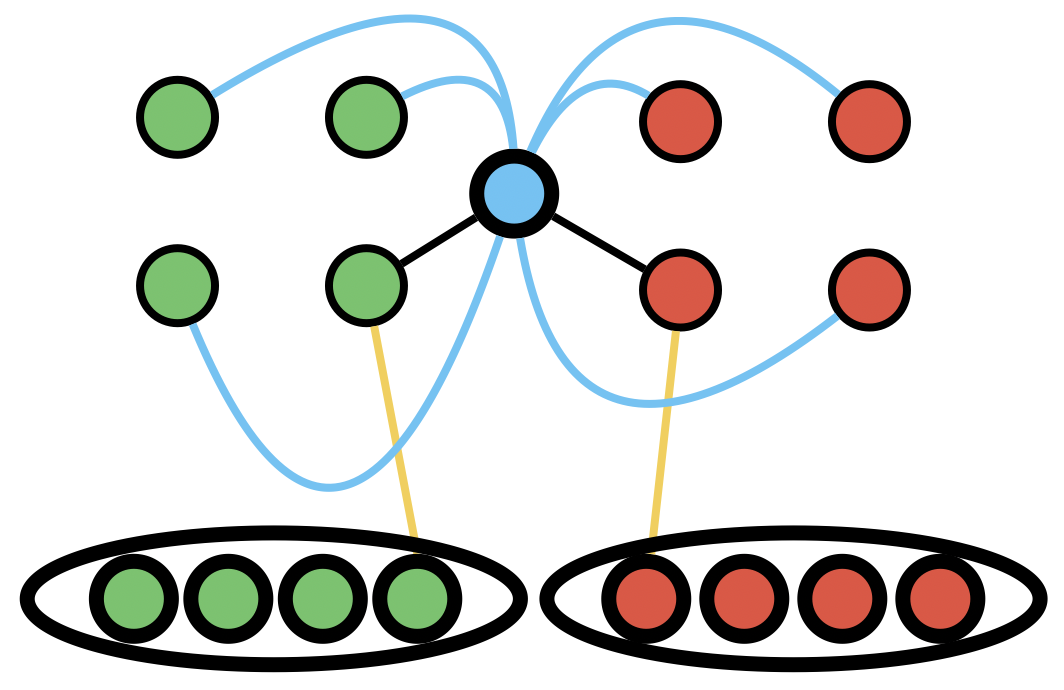}
      \caption{Solution of $G$.}\label{subfig:gstlcmst-solution}
    \end{subfigure}%
\caption{Given an instance of group Steiner tree on $F$ where $g_1$ contains green vertices and $g_2$ contains red vertices and the root $r$ is the bolded blue vertex (\Cref{subfig:gst-instance}), we transform it into an instance of length-constrained MST of $G,h$ where the bolded green/red vertices form the set of copy vertices $V'$, yellow edges form the set $P$, blue edges form the set $S$, and there is a clique of edges within both bolded ovals containing the bolded red/green vertices (\Cref{subfig:gstlcmst-instance}). A solution for $G$ where there is a spanning subtree of edges within both bolded ovals (\Cref{subfig:gstlcmst-solution}) is transformed into a solution for $F$ (\Cref{subfig:gst-solution}).}
\label{fig:gst-reduction}
\end{figure}

\noindent Observe that the same reduction can be made for length-constrained Steiner tree, where the construction of $G$ is identical except there is no set $S$ of edges. By the same argument this implies the same hardness but with length slack $s$ for $s<3$. 

% \newpage
\appendix

\section{Variants of Our Algorithm}\label{sec:variants}
Here we describe how to adjust \Cref{alg:planarmain2} to improve the approximation guarantee in quasipolynomial time, and to work for length-constrained Steiner tree.
\subsection{A Quasipolynomial Time Algorithm}\label{sec:qpt}

We show how to improve the weight approximation from the main algorithm if we are willing to give up a quasipolynomial runtime. Basically, we can afford to be less careful with the setting of $\alpha,\beta$ if we only require quasipolynomial time. Specifically, we can obtain the following theorem:
\begin{restatable}{theorem}{qptthmmain}\label{thm:qptplanarmain}
    There is an $O\left(\log n\right)$ approximation algorithm with $1+o(1)$ length slack that runs in quasipolynomial time for planar length-constrained MST.
\end{restatable}
\begin{proof}[Proof Sketch]
    Previously, $\alpha,\beta$ had the following roles: we guessed $h^*$ for each piece in buckets of (roughly) $1/\beta$ in order to have $\beta^{\beta}$ total guesses, where $\beta=\beta(\eps)$ had to be chosen such that $\beta^\beta$ was polynomial on $n$. Next, we needed
$$\frac{h\log_{\alpha}n}{\beta}=h\eps\implies \alpha=\exp\left(\frac{\ln n}{\eps\beta}\right)$$
for the length slack. Lastly, we charged $\alpha$ to the weight of to connect our Steiner trees to the separator hierarchy, and charged $\beta^\delta$ to the weight due to \Cref{thm:chandra}.

    So we want to minimize $\alpha,\beta,\epsilon$ all at once (although the $\delta$ gives us some leeway on setting $\beta$ as we will see soon). If we only care about the total number of guesses to be quasipolynomial, then we could guess every integer in $[h]$ for each piece, which only requires $\beta$ be at most some $\poly(\log n)$. The dynamic programming table defined as before with this many guesses has quasipolynomial size. Now suppose $\alpha=O(1),\beta=\log^c n$; we have $\eps=1/\log^{c-1}n$ (as we require $\alpha=\exp\left(\frac{\ln n}{\eps\beta}\right)$), and this is $o(1)$ for any $c>1$. Then the total weight of our solution is
    $$\overbrace{O(1)}^{\alpha}\cdot \overbrace{O(\log n)}^{\log_{\alpha}n} \cdot \overbrace{O(\log^{c\delta}n)}^{\beta^\delta}\cdot\OPT$$
    so we just need to set $\delta<1/c$ to beat the lower bound of \Cref{thm:newlowerbound}. But we can do even better since we are already giving up quasipolynomial time: we can apply \Cref{thm:chandra} with $\delta$ as low as $\frac{1}{c\log\log n}$, giving a quasipolynomial time subroutine obtaining a $\poly(\log\log n)$ approximation of the subinstance. Then we can run \Cref{alg:planarmain2} in quasipolynomial time by setting $\eps=1/\log n,\delta=\frac{1}{2\log\log n},\alpha=O(1),\beta=\log^2n$. Noting that $\log^{1+\frac{1}{\log\log n}}n=O(\log n)$ we obtain \Cref{thm:qptplanarmain}.
\end{proof}

\noindent Under stronger assumptions, this shows that length-constrained MST is easier on planar graphs than on general graphs for the class of quasipolynomial time algorithms as well. We have the following theorem due to \cite{grandoni2019log2} and our reduction in \Cref{sec:gst-reduction}:
\begin{theorem}
    Length-constrained MST cannot be $o(\log^2 n/\log\log n)$-approximated in quasipolynomial time with length slack $s$ for $s<2$ unless $\textsf{NP}\subseteq\bigcap_{\eps\in(0,1)}\textsf{ZTIME}\left(2^{n^\eps}\right)$, or the Projection Game Conjecture is false.
\end{theorem}

% \noindent Further setting $\epsilon=1/\log\log n$ immediately gives the following:
% \begin{corollary}    There is an $\Tilde{O}\left(\log n\right)$ approximation algorithm with $1+o\left(1\right)$ length slack that runs in quasipolynomial time for planar length-constrained MST. \end{corollary}

\subsection{Planar Length-Constrained Steiner Tree}\label{sec:steiner}
Our algorithms are generalizable to length-constrained Steiner tree. In length-constrained Steiner tree, we are given the same input as in length-constrained MST along with a terminal subset $U\subseteq V,|U|=t$ and are required to return a tree of $G$ with minimum total edge weight among trees $T$ of $G$ satisfying $d_T(r,u)\leq h$ for each $u\in U$.

\subsubsection{Extending \Cref{alg:planarmain2}}
\begin{theorem}\label{coro:lcsteiner}
    There is a polynomial time 
    % $O\left(\frac{\log^{1+\eps}k}{(\eps^2\log\log k)^\delta}\right)$ 
    $O\left(\log^{1+\eps}t\right)$ approximation algorithm with $1+\eps$ length slack for planar length-constrained Steiner tree, where $t$ is the number of terminals. 
    %$O\left(\frac{\log^{1+\eps+\delta}k}{(\eps^2\log\log k)^\delta}\right)$
\end{theorem}
\begin{proof}[Proof Sketch]
    % We give a high-level overview of the generalization to length-constrained Steiner tree. 
    Here $T^*,\OPT$ denote the optimal solution and its weight to our instance of length-constrained Steiner tree on graph $G=(V,E)$, terminals $U$, and length bound $h$.
    % and we will compute our mixture metric using $D^{(h)}(U)\leq\OPT$. 
    Recall that our algorithm for length-constrained MST computed a length-constrained division hierarchy assuming that the vertices were uniformly weighted, although all of the length-constrained separator/division/hierarchy technology that we defined works for general nonnegative vertex weight functions. 
    
    We compute the hierarchy using \Cref{lemma:hierarchy} with the weight function $W$ that assigns all terminals weight $1$ and all other vertices weight $0$. This implicitly uses \Cref{lemma:lc-div} (which implicitly uses \Cref{lemma:lengthseparator}). In particular, in the step of the algorithm in \Cref{lemma:lc-div} where we balance non-boundary vertex weights, we now balance the number of non-boundary \textit{terminals}, and we obtain divisions whose regions contain at most a $1/\alpha$ fraction of all non-boundary terminals. Note the hierarchy has depth $\log_\alpha t$.

    In the dynamic programming, we add to our guessing framework a guess for each parent piece, given a node in the hierarchy. Specifically, for each parent piece, we guess whether or not $T^*$ uses the piece, i.e.\ if a vertex in the piece is in $T^*$. Then we define the subinstances of length-constrained Steiner tree for each guess and region by adding edges to a parent piece only if we guess that $T^*$ uses it, and the terminals of the subinstance are defined as before. This adds a multiplicative $2^{O\left(\beta\right)}=\poly(n)$ to the total number of guesses, and thus the size of the dynamic program is still $\poly(n)$. 

    The arguments for runtime, length, and weight are roughly the same. However, we actually obtain an $O\left(\log^{1+\eps}t\right)$ approximation by setting $\alpha,\beta$ to be in terms of $t$ rather than $n$.
\end{proof}

\subsubsection{Extending \Cref{alg:planar-intgap}}

\begin{theorem}
    Let $x$ be an optimal fractional solution to the flow-based LP of length-constrained Steiner tree with graph $G$, terminals $U,|U|=t$, root $r$, and length bound $h$; there exists a polynomial-time algorithm that returns a spanning tree $T$ of weight $O\left(\frac{\log^3 t}{\eps}\right)w(x)$ and satisfying $d_T(r,v)\leq (1+\eps)h$ for every $v\in U$. 
\end{theorem}
\begin{proof}[Proof Sketch]
    Given an instance of length-constrained Steiner tree with graph $G=(V,E)$ and terminals $U\subseteq V$, we modify the LP by replacing every $\forall\text{ } v\in V-r$ in \Cref{equ:lcmst-flow} with $\forall\text{ }v\in U$. Then we use \Cref{alg:planar-intgap} with the following modifications: \textbf{(1)} we compute the hierarchy using \Cref{lemma:hierarchy} with the weight function $W$ that assigns all terminals weight $1$ and all non-terminals weight $0$, \textbf{(2)} we use $\textbf{w}(x_{|H})$ in place of $D^{(h)}(H)$ in the mixture metric, \textbf{(3)} we prune away vertices whose distance under $w$ is at most $w(\textbf{x}_{|H})\cdot \log t$ before computing separators and scale up the remaining flow by $1+1/\log t$.\footnote{This is the same pruning and scaling operation in \cite{chekuri2024directedpolymatroid} that prevents the separator from having unbounded weight while maintaining a feasible flow.} For the shortcuts, we only need to consider adding shortcuts to pieces that contain a terminal, rather than every possible piece. As a result, a very similar analysis follows where things are in terms of $t$ rather than $n$.
\end{proof}

\section{Length-Constrained MST is Basically Directed Steiner Tree}\label{sec:dst-reduction}
In this section we prove the following theorem:
\begin{theorem}
    There exists a polynomial time $\alpha$-approximation algorithm for length-constrained MST with length slack $1$ if and only if there exists a polynomial time $\alpha$-approximation algorithm for directed Steiner tree. 
\end{theorem}
\noindent
And this theorem is immediately implied by the following three lemmas we prove in the following three subsections:
\begin{lemma}[Length-Constrained MST is Basically Length-Constrained Steiner Tree]\label{lemma:lcmst-lcst}
    There exists a polynomial time $\alpha$-approximation algorithm for length-constrained MST with length slack $1$ if and only if  there exists a polynomial time $\alpha$-approximation algorithm for length-constrained Steiner tree.
\end{lemma}
\begin{lemma}[$\rightarrow$]\label{lemma:lcmst-dst}
    There exists a polynomial time $\alpha$-approximation algorithm for length-constrained MST with length slack $1$ if there exists a polynomial time $\alpha$-approximation algorithm for directed Steiner tree. 
\end{lemma}
\begin{lemma}[$\leftarrow$]\label{lemma:dst-lcmst}
    There exists a polynomial time $\alpha$-approximation algorithm for directed Steiner tree if there exists a polynomial time $\alpha$-approximation algorithm for length-constrained MST with length slack $1$. 
\end{lemma}
\subsection{Length-Constrained MST is Basically Length-Constrained Steiner Tree}
We show that if we are restricted to length slack $1$, then length-constrained MST is equivalent to length-constrained Steiner tree. 
\begin{proof}[Proof of \Cref{lemma:lcmst-lcst}]
    Length-constrained MST is a special case of length-constrained Steiner tree, so it suffices to reduce length-constrained Steiner tree to length-constrained MST. 
    
    \textbf{The reduction.} We use the same trick with the set $S$ of edges in the reduction in the proof of \Cref{thm:newlowerbound}. Specifically, given an instance of length-constrained Steiner tree with graph $F$, terminals $U\subseteq V(F)$, and length bound $h$, we construct an instance of length-constrained MST with graph $G$ such that $V(G)=V(F)$ and $E(G)$ is the union of the set $S$ of edges containing an edge of length $h$ and weight $0$ between the root $r$ and every vertex in $V(F)\setminus U$ with $E(F)$, where we keep the same lengths and weights for edges in $E(F)$. 

    Given a solution $T_G$ to the instance of length-constrained MST with $G,h$, we transform it into a solution $T_F$ to the instance of length-constrained Steiner tree with $F,U,h$ by taking the edges $E\left(T_G\right)\cap E(F)$. A length-constrained Steiner tree solution $T'_F$ of $F,U,h$ gives a solution $T'_G$ to the length-constrained MST instance with $G,h$ by taking the union of $E\left(T'_F\right)\cup S$, removing redundant edges in $S$ (i.e.\ $r,v$ edges in $S$ such that there already exists an $r,v$ path in $E\left(T'_F\right)$) from there to form a tree.

    \textbf{Correctness of the reduction.} This is immediate by construction.
\end{proof}
\subsection{Reduction to Directed Steiner Tree}
We give a reduction from length-constrained minimum spanning tree to the directed Steiner tree problem. However, the reduction does not preserve planarity, preventing us from simply using e.g.\ the algorithm of \cite{friggstad2023planardst}, and thus motivating the main algorithm of this paper. 
\begin{proof}[Proof of \Cref{lemma:lcmst-dst}]
\textbf{The reduction.} Given an instance of length-constrained MST with graph $F=(V,E)$, root $r$, and diameter bound $h$, we construct a directed layered graph $G$ rooted at $r$ as follows: 
\begin{enumerate}
    \item\textbf{Layers}: there are $h+1$ layers indexed from $0$ to $h$. 
    \item \textbf{Vertices}: layer $0$ contains just $r$, and the remaining layers $1$ through $h$ contain a copy of $V(F)-r$. For a vertex $v\in V(F)-r$ we let $v_i$ be the copy of $v$ in layer $i$.
    \item \textbf{Edges}: for every $(r,v)\in E(F)$, we add a directed edge of the same weight from $r$ to $v_1$ in $H$. For $(u,v)\in E(F)$ where $u,v\neq r$ we add a directed edge of the same weight from $u_i$ to $v_{i+1}$, for every $i\in[h]$. We also add a directed edge of weight $0$ from $v_i$ to $v_{i+1}$ for every $i\in[h-1]$. 
    \item \textbf{Terminals}. In this instance of directed Steiner tree, we let $r$ be the root and all the vertex copies in layer $h$ to be the terminals $U$. 
\end{enumerate}
Given a solution $T_G$ to the instance of directed Steiner tree with $G,U,r$, we show how to transform it into a solution $T_F$ to the instance of length-constrained MST with $F,r,h$. But first, some notation: we let $T(v)$ be the unique root to $v$ path in a tree $T$. The first step of the transformation is to remove all of the weight $0$ edges from $T_G$. Let $T''_G\subseteq T_G$ be the new decremented tree. Now for every $v\in V(G)-r$ we have a path $T''_G\left(v_{i_v}\right)$ where $i_v\leq h$ and no two vertices in $T''_G\left(v_{i_v}\right)$ are copies of the same vertex in $V(F)$. Then $T''_G\left(v_{i_v}\right)$ is equivalent to a simple path from $r$ to $v$ in $F$ that is sure to exist by our construction of $G$ from $F$, and we add this path to $T_F$. We do this for all $v\in V(F)-r$\footnote{We remark that a reduction from length-constrained \textit{Steiner} tree to directed Steiner tree can be done by basically using this construction. The only difference is that instead of making the entire layer $h$ be the set of terminals, we only set the copies that correspond to terminals from the length-constrained Steiner tree instance to be the terminals.}. 
% Given a solution $T'_F$ to the instance of length-constrained MST with $F,r,h$, we can get a solution $T'_G$ to the instance of directed Steiner tree with $G,U,r$ by taking the edge from $u,v$ in $F$ if $T'_G$ contains the edge between $u$'s copy in layer $i$ and $v$'s copy in layer $i+1$ for any $i$.

We give a small example showing that this reduction does not preserve planarity. Let $F$ be the union of $K_4$ and the root $r$, which is connected to two arbitrary vertices of the $K_4$, and let $h=2$. Clearly $F$ is planar since it is $\left(K_5,K_{3,3}\right)$-free, while $G$ is not since it contains $K_{3,3}$ as a minor. See \Cref{fig:nonplanarDST}.

\textbf{Correctness of the reduction.} By construction, we have that if $T_F$ is feasible/optimal then so is $T_G$, so it remains to show that if $T_G$ is feasible/optimal for its instance then so is $T_F$. Observe that any edge in $T_G$ is directed from some layer $i$ to $i+1$, so every $T_G(v_h)$ must visit each layer exactly once.
% It is clear that if $T_F$ is feasible then by our transformation $T_G$ must also be feasible. 
If $T_G$ is feasible yet $T_F$ isn't, then $T_F$ either violates the $h$ bound because it isn't connected or there is some $v$ such that $|T_F(v)|>h$. Indeed, by our transformation this means that some terminal is not connected via a path to $r$ in $T_G$, or $T_F$ contains a cycle (we can WLOG assume this is not the case). So $T_G$ must not be feasible either. 

Observe that $w\left(T_F\right)=w\left(T_G\right)$ since for every edge in $T_F(v)$, there is an edge of identical weight in $T_G(v)$, and the edges between identical vertex copies in $G$ have weight $0$.
% So if $T_F$ is optimal and $T_G$ isn't, then there is some solution $D$ of the directed Steiner tree instance such that $w(D)<w(T_G)=w(T_F)$. If there was a $v_h$ such that $w\left(T_G(v_h)\right)<w(T_F(v))$ then there is an $h$-length path to $v$ in $T_F$ that we could replace $T_F(v)$ with, contradicting $T_F$'s optimality. 
So if $T_G$ is optimal and $T_F$ isn't, then there is some solution $L$ of the length-constrained MST instance such that $w(L)<w\left(T_F\right)=w\left(T_G\right)$. Then using the transformation described above on $L$, we can construct a solution of the directed Steiner tree instance of strictly less weight than $T_G$, contradicting $T_G$'s optimality. Therefore if $T_G$ is optimal then so is $T_F$. 
\end{proof}
\begin{figure}[H]
    \begin{subfigure}[b]{.25\textwidth}
      \centering
        \includegraphics[width=.7\linewidth]{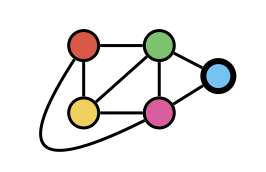}
      \caption{$F$.}\label{subfig:lcmst-instance}
    \end{subfigure}%
    ~
    \begin{subfigure}[b]{.25\textwidth}
      \centering
      \includegraphics[width=.7\linewidth]{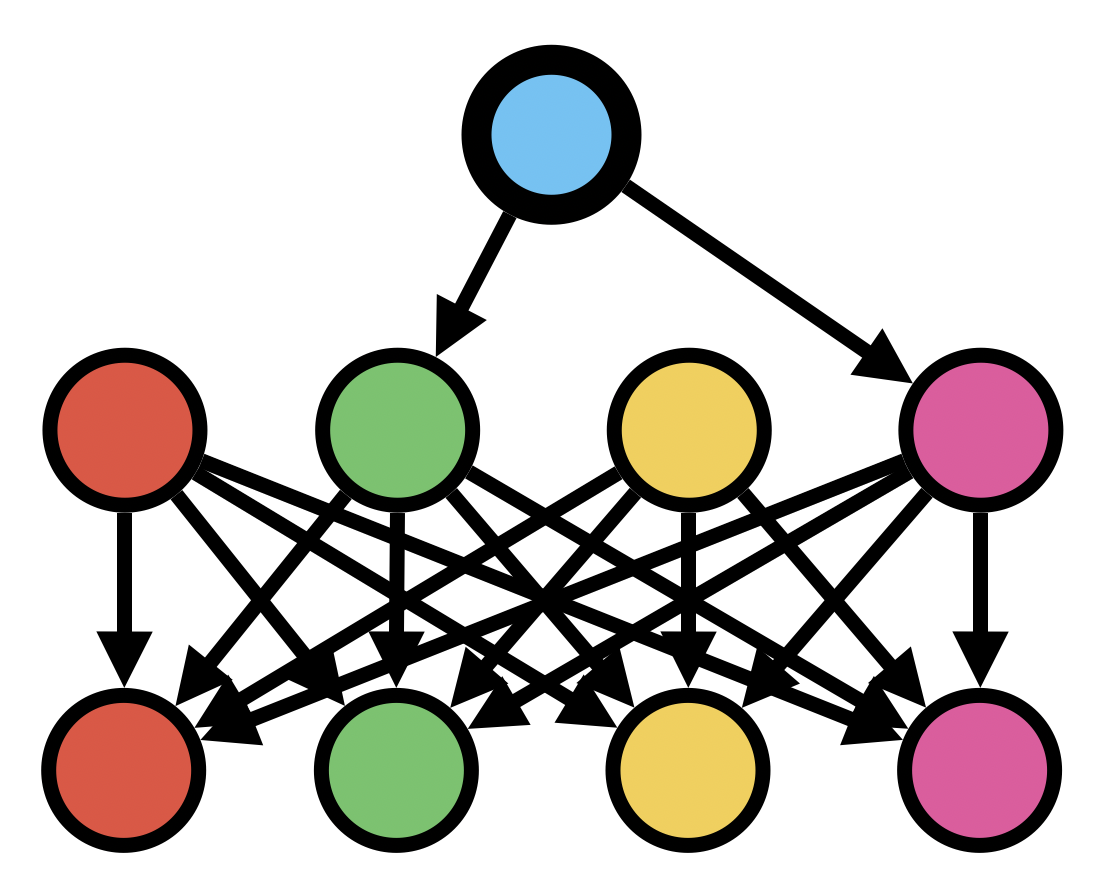}
      \caption{$G$.}\label{subfig:dst-instance}
    \end{subfigure}%
    ~
    \begin{subfigure}[b]{.25\textwidth}
      \centering
        \includegraphics[width=.7\linewidth]{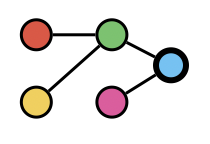}
      \caption{Solution of $F$.}\label{subfig:lcmst-solution}
    \end{subfigure}%
    ~
    \begin{subfigure}[b]{.25\textwidth}
      \centering
      \includegraphics[width=.7\linewidth]{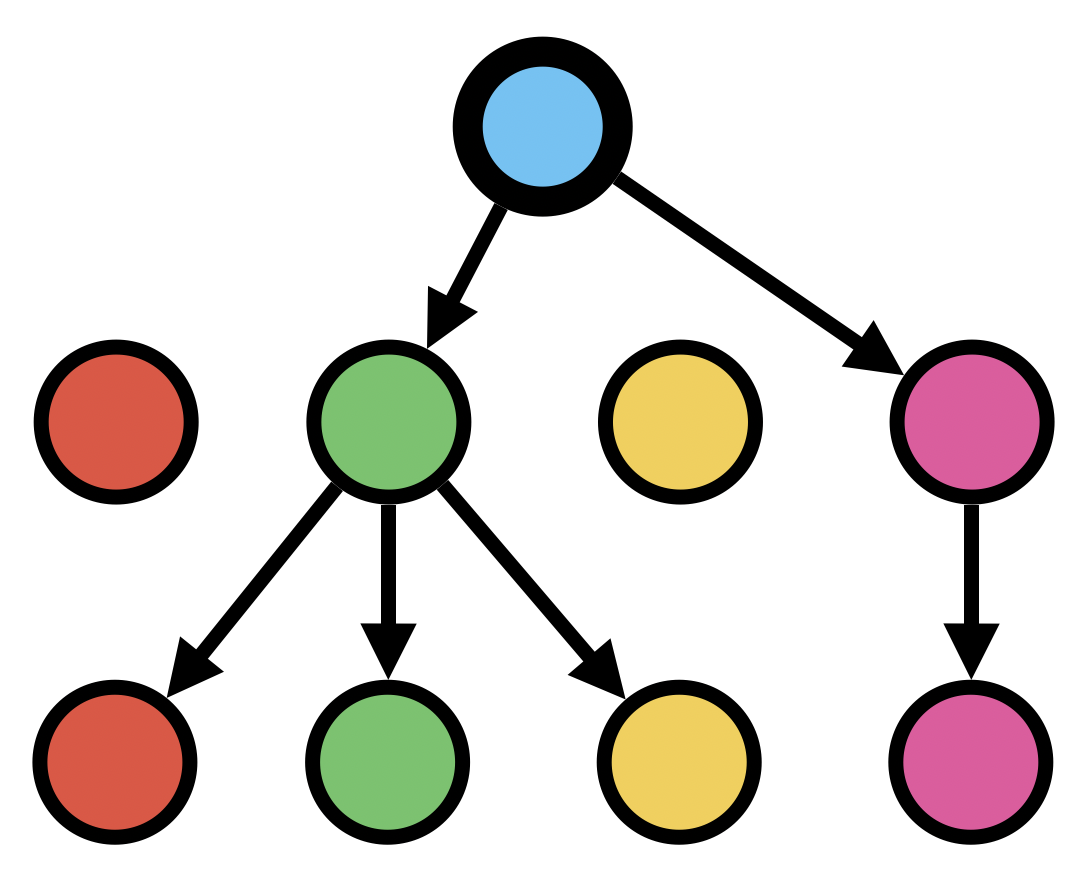}
      \caption{Solution of $G$.}\label{subfig:dst-solution}
    \end{subfigure}%
\caption{Given an instance of length-constrained MST on graph $F$ with the root as the bolded blue vertex and $h=2$ (\Cref{subfig:lcmst-instance}), we transform it into a directed $h$-layered graph $G$ where all edges are directed downwards (\Cref{subfig:dst-instance}). We solve the directed Steiner tree instance on $G$ (\Cref{subfig:dst-solution}), and transform it into a solution for the length-constrained MST instance on $F$ (\Cref{subfig:lcmst-solution}).}
\label{fig:nonplanarDST}
\end{figure}
\noindent 
Our reduction’s runtime depends linearly on $h$ (assuming integer lengths); this dependence is likely necessary since $h$ often (provably necessarily) appears in the runtime of algorithms for length-constrained problems. In particular, \cite{kociumaka2022bellman} showed that computing $h$-length-constrained distances requires $\Omega(hm)$ time under the APSP hypothesis. It would be surprising if there was a reduction with faster runtime than what's required to compute length-constrained distances.

\subsection{Reduction from Directed Steiner Tree}
We wrap things up with a reduction from directed Steiner tree to length-constrained Steiner tree, which combined with our reduction to length-constrained MST, shows that directed Steiner tree is equivalent to length-constrained MST (with length slack $1$). 
\begin{proof}[Proof of \Cref{lemma:dst-lcmst}]
    \textbf{The reduction.} Given an instance of directed Steiner tree involving a directed graph $F$ with root $r$ and terminal set $U$, we construct an instance of length-constrained Steiner tree with undirected layered graph $G$ as follows:
    \begin{enumerate}
        \item \textbf{Layers}: There are $n$ layers indexed from $0$ to $n-1$.
        \item \textbf{Vertices}: In layer $i$, there is a copy of each vertex in $V(F)$, and for each terminal in $U$, there is a special copy, for every $i\in[0,n-1]$. 
        \item \textbf{Edges}: For every directed edge $e=(u\to v)$ in $E(F)$, we have an edge of length $1$ and weight $w_F(e)$ between the copy of $u$ in layer $i$ and the copy of $v$ in layer $i+1$ for every $i\in[0,n-2]$. For every terminal $t\in U$, we add an edge of length $n-i$ and weight $0$ between its copy in layer $i$ and its special copy in layer $i$, for every $i\in[0,n-1]$. For every terminal $t\in U$, we add a clique of edges of length and weight $0$ among all $n$ of its special copies.
        \item \textbf{Terminals}: We let all of the special copies of each terminal in $U$ be the terminals $U_G$.
        \item \textbf{Length constraint}: We set the length constraint to be $h=n$.
    \end{enumerate}
    Given a solution $T_G$ to the instance of length-constrained Steiner tree with $G,r,h$, we can get a solution $T_F$ to the instance of directed Steiner tree with $F,U,r$ by taking the directed edge from $u$ to $v$ in $F$ if $T_G$ contains an edge between a copy of $u$ and a copy of $v$. See \Cref{fig:dst-to-lcst}.
    % Meanwhile, given a solution $T'_F$ to the instance of directed Steiner tree with $F,U,r$, we can get a solution $T'_G$ to the instance of length-constrained Steiner tree with $G,r,h$ by taking the edge between $u$'s copy in layer $i$ and $v$'s copy in layer $i+1$ in $G$ if $(u\to v)$ is a directed edge in $F$ such that $u$ is $i$ directed hops away from $r$ in $T'_F$, and for each terminal $t$ that is $i$ directed hops away from $r$ in $T'_F$, we take the edge between $t$'s copy in layer $i$ and $t$'s special copy in layer $i$ in $G$. See \Cref{fig:dst-to-lcst}.
    
    \textbf{Correctness of the reduction.} By construction, we have that if $T_F$ is feasible/optimal then so is $T_G$. Also by construction, if $T_G$ is feasible then so is $T_F$. Since the only edges with positive weight in $G$ are edges in $E(F)$, we have that if $T_G$ is optimal then so is $T_F$.
    % $T'_F$ is feasible if $T_G'$ is, so it remains to show that $T_F'$ is optimal if $T_G'$ is. Suppose not, then there is a length-constrained Steiner tree solution $L$ with strictly less weight than $T_G'$ and $T_F'$, since the only positive weight edges in $G$ are copies of edges in $F$, but this contradicts $T_F'$'s optimality.

    Note this was a reduction to length-constrained Steiner tree. By \Cref{lemma:lcmst-lcst}, we arrive at the lemma statement.
\end{proof}
\begin{figure}[H]
    \begin{subfigure}[b]{.25\textwidth}
      \centering
        \includegraphics[width=.5\linewidth]{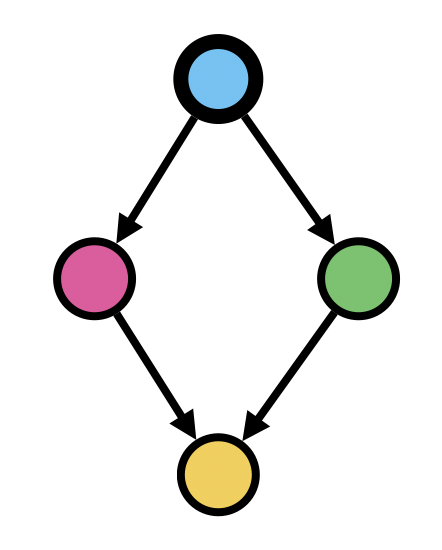}
      \caption{$F$.}\label{subfig:dst-to-lcst-init}
    \end{subfigure}%
    ~
    \begin{subfigure}[b]{.25\textwidth}
      \centering
      \includegraphics[width=.9\linewidth]{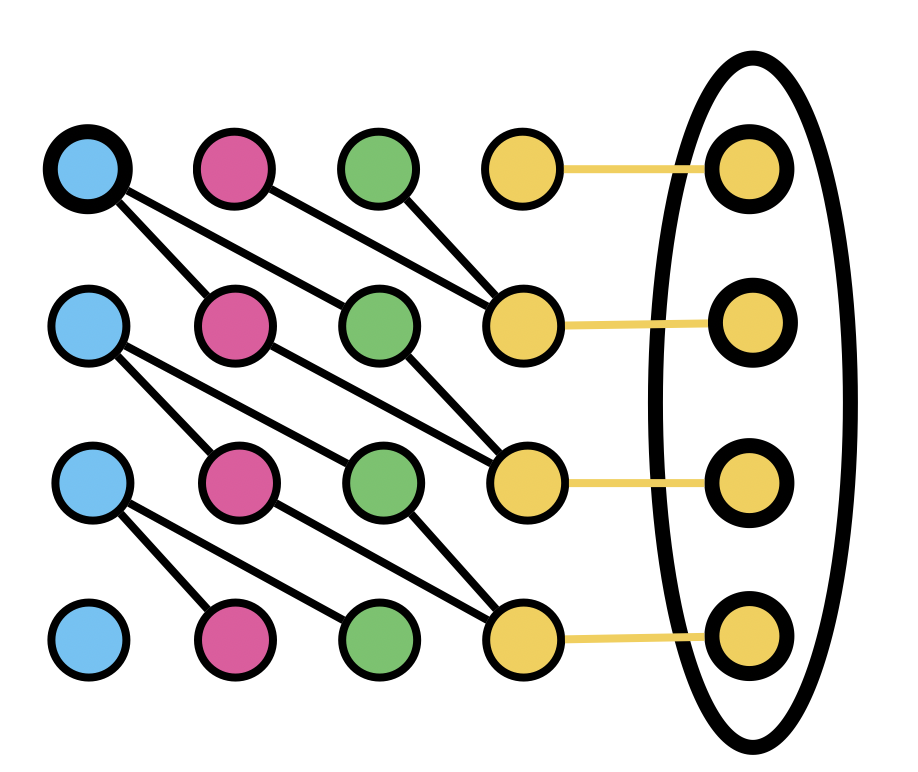}
      \caption{$G$.}\label{subfig:dst-to-lcst}
    \end{subfigure}%
    ~
    \begin{subfigure}[b]{.25\textwidth}
      \centering
        \includegraphics[width=.5\linewidth]{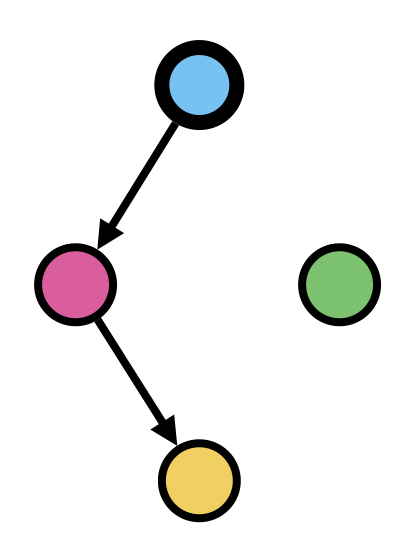}
      \caption{Solution of $F$.}\label{subfig:dst-to-lcmst-init-sol}
    \end{subfigure}%
    ~
    \begin{subfigure}[b]{.25\textwidth}
      \centering
      \includegraphics[width=.9\linewidth]{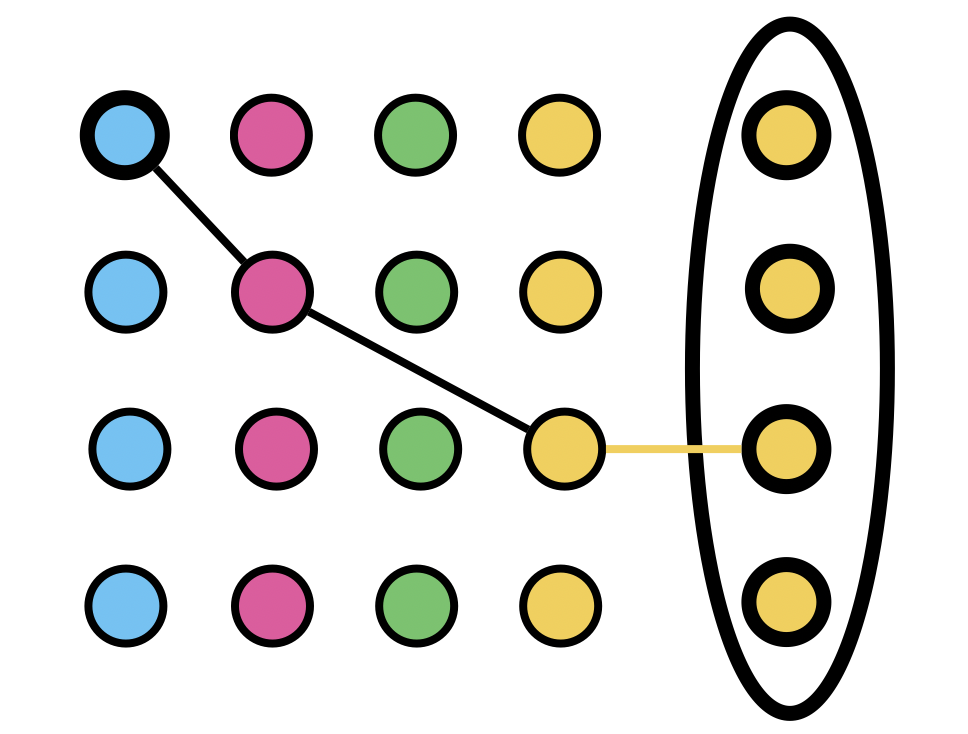}
      \caption{Solution of $G$.}\label{subfig:dst-to-lcmst-sol}
    \end{subfigure}%
\caption{Given an instance of directed Steiner tree on graph $F$ where the root is the bolded blue vertex and (single) terminal is the yellow vertex (\Cref{subfig:dst-to-lcst-init}), we transform it into an $n$-layered graph $G$ where black edges have length $1$ and yellow edges at layer $i$ have length $n-i$ (\Cref{subfig:dst-to-lcst}). We solve the length-constrained Steiner tree instance on $G$ (\Cref{subfig:dst-to-lcmst-sol}), and transform it into a solution for the directed Steiner tree instance on $F$ (\Cref{subfig:dst-to-lcmst-init-sol}).}
\label{fig:dst-to-lcst}
\end{figure}

% \newpage
\bibliographystyle{alpha}
\bibliography{references}

\end{document}

%% file: prologue.tex
\usepackage{amsmath}
\usepackage{amssymb}
\usepackage{amsthm}
\usepackage{thmtools}
\usepackage{thm-restate}
\usepackage{bm}
\usepackage{graphicx}
\usepackage{caption}
\usepackage[margin = 1in]{geometry}
\usepackage{natbib}
\setcitestyle{numbers,square}
\usepackage{color}
\usepackage{media9}
\usepackage{subcaption}
\usepackage{comment}
\usepackage{enumitem}
\usepackage[nameinlink,capitalize,nosort]{cleveref}
\crefname{claim}{Claim}{Claims}
\usepackage{palatino}

\usepackage{algorithm}
\usepackage[noend]{algpseudocode}

\usepackage[colorinlistoftodos]{todonotes}
\usepackage[most]{tcolorbox}  

\definecolor{mygreen}{rgb}{0,0.65,0.66}

\newcommand{\OPT}{\operatorname{OPT}}

\newcommand{\DP}{\operatorname{DP}}

\newcommand{\LCST}{\operatorname{LCST}}
\newcommand{\wmix}{w_{mix}}
\newcommand{\dmix}{d_{mix}}

\graphicspath{{imgs/}}

\makeatletter
\def\mathcolor#1#{\@mathcolor{#1}}
\def\@mathcolor#1#2#3{%
  \protect\leavevmode
  \begingroup
    \color#1{#2}#3%
  \endgroup
}
\makeatother

\newcommand*{\vsepfbox}[1]{%
  \begingroup
    \sbox0{\fbox{#1}}%
    \setlength{\fboxrule}{0pt}%
    \mbox{\kern-\fboxsep\fbox{\unhbox0}\kern-\fboxsep}%
  \endgroup
}

\theoremstyle{plain} \numberwithin{equation}{section}
\newtheorem{theorem}{Theorem}[section]
\numberwithin{theorem}{section}

\newtheorem{lemma}[theorem]{Lemma}

\newtheorem{definition}[theorem]{Definition}

\theoremstyle{definition}

\usepackage{thmtools} 
\usepackage{thm-restate}

\DeclareMathOperator*{\argmin}{argmin}

\makeatletter
\newcommand{\subalign}[1]{%
  \vcenter{%
    \Let@ \restore@math@cr \default@tag
    \baselineskip\fontdimen10 \scriptfont\tw@
    \advance\baselineskip\fontdimen12 \scriptfont\tw@
    \lineskip\thr@@\fontdimen8 \scriptfont\thr@@
    \lineskiplimit\lineskip
    \ialign{\hfil$\m@th\scriptstyle##$&$\m@th\scriptstyle{}##$\hfil\crcr
      #1\crcr
    }%
  }%
}
\makeatother

\newcounter{note}[section]